\documentclass[11pt]{article}

\usepackage{amsmath,amssymb,bbm,amsthm}
\usepackage{fullpage}
\usepackage{thm-restate,color,xcolor,xspace,graphicx}
\usepackage{algorithm}
\usepackage{algorithmic}

\usepackage{thmtools} 
\usepackage{thm-restate}
\usepackage{comment}

\def\BAL#1\EAL{\begin{align*}#1\end{align*}}
\def\BALN#1\EALN{\begin{align}#1\end{align}}
\def\BG#1\EG{\begin{gather}#1\end{gather}}

\usepackage{xcolor}
\usepackage{nameref}
\definecolor{ForestGreen}{rgb}{0.1333,0.5451,0.1333}
\definecolor{DarkRed}{rgb}{0.8,0,0}
\definecolor{Red}{rgb}{1,0,0}
\usepackage[linktocpage=true,
pagebackref=true,colorlinks,
linkcolor=DarkRed,citecolor=ForestGreen,
bookmarks,bookmarksopen,bookmarksnumbered]
{hyperref}

\usepackage{hyperref,cleveref}
\usepackage{enumitem}

\newcommand{\eat}[1]{}

\declaretheorem[numberwithin=section]{theorem}
\declaretheorem[numberlike=theorem]{lemma}
\declaretheorem[numberlike=theorem,name=Lemma]{lem}
\declaretheorem[numberlike=theorem]{fact}

\declaretheorem[numberlike=theorem,name=Proposition]{prop}
\declaretheorem[numberlike=theorem]{corollary}
\declaretheorem[numberlike=theorem,name=Corollary]{cor}
\declaretheorem[numberlike=theorem]{claim}

\declaretheorem[numberlike=theorem,style=definition]{definition}
\declaretheorem[numberlike=theorem,name=Definition,style=definition]{defn}

\newtheorem*{theorem*}{Theorem}

\newcommand{\e}{\epsilon}
\newcommand{\cS}{{\cal S}}
\newcommand{\tO}{\tilde{O}}

\global\long\def\polylog{\mathrm{polylog}}
\global\long\def\almostfair{\textsc{AlmostFair}}
\global\long\def\Otil{\tilde{O}}
\global\long\def\boundary#1{N\left\langle #1\right\rangle }
\global\long\def\fbar{\bar{f}}
\global\long\def\Deltabar{\overline{\Delta}}
\global\long\def\one{\boldsymbol{1}}
\global\long\def\ex{\mathsf{def}}
 \global\long\def\exx{\overline{\mathsf{def}}}
\global\long\def\poly{\mathrm{poly}}

\global\long\def\vol{\textup{vol}}

\ifdefined\ShowComment
\def\thatchapholtext#1{\textcolor{purple}{#1}}
\def\thatchaphol#1{\marginpar{$\leftarrow$\fbox{T}}\footnote{$\Rightarrow$~{\sf\textcolor{purple}{#1 --Thatchaphol}}}}
\def\danupon#1{\textcolor{orange}{DN: #1}}

\def\jason#1{\marginpar{$\leftarrow$\fbox{J}}\footnote{$\Rightarrow$~{\sf\textcolor{blue}{#1 --Jason}}}}
\def\note#1{#1}
\def\alert#1{\textcolor{red}{#1}}
\else
\def\thatchapholtext#1{}
\def\thatchaphol#1{}
\def\danupon#1{}
\def\shen#1{}
\def\sorrachai#1{}
\def\jason#1{}
\def\note#1{} 
\def\alert#1{}
\fi

\begin{document}

\title{Near-Linear Time Approximations for Cut Problems via Fair Cuts}

\author{}
\date{}
\author{Jason Li\thanks{Simons Institute for Theory of Computing, UC Berkeley. Email: {\tt jmli@alumni.cmu.edu}} \and Danupon Nanongkai\thanks{Max Planck Institute for Informatics \& University of Copenhagen \& KTH. Email: {\tt danupon@gmail.com}} \and Debmalya Panigrahi\thanks{Department of Computer Science, Duke University. Email: {\tt debmalya@cs.duke.edu}} \and Thatchaphol Saranurak\thanks{University of Michigan, Ann Arbor. Email: {\tt thsa@umich.edu}}}
\maketitle

\pagenumbering{gobble}

\begin{abstract}
We introduce the notion of {\em fair cuts} as an approach to leverage approximate $(s,t)$-mincut (equivalently $(s,t)$-maxflow) algorithms in undirected graphs to obtain near-linear time approximation algorithms for several cut problems. Informally, for any $\alpha\geq 1$, an $\alpha$-fair $(s,t)$-cut is an $(s,t)$-cut such that there exists an $(s,t)$-flow that uses $1/\alpha$ fraction of the capacity of \emph{every} edge in the cut. (So, any $\alpha$-fair cut is also an $\alpha$-approximate mincut, but not vice-versa.) We give an algorithm for $(1+\epsilon)$-fair $(s,t)$-cut in $\tilde{O}(m)$-time, thereby matching the best runtime for $(1+\epsilon)$-approximate $(s,t)$-mincut [Peng, SODA '16]. We then demonstrate the power of this approach by showing that this result almost immediately leads to several applications: 
\begin{itemize}
    \item the first nearly-linear time $(1+\epsilon)$-approximation algorithm that computes all-pairs maxflow values (by constructing an approximate Gomory-Hu tree). Prior to our work, such a result was not known even for the special case of Steiner mincut [Dinitz and Vainstein, STOC '94; Cole and Hariharan, STOC '03]; 
    \item the first almost-linear-work subpolynomial-depth parallel algorithms for computing $(1+\epsilon)$-approximations for all-pairs maxflow values (again via an approximate Gomory-Hu tree) in unweighted graphs; 
    \item the first near-linear time expander decomposition algorithm that works even when the expansion parameter is polynomially small; this subsumes previous incomparable algorithms [Nanongkai and Saranurak, FOCS '17; Wulff-Nilsen, FOCS '17; Saranurak and Wang, SODA '19].
\end{itemize}
\end{abstract}

\clearpage
\tableofcontents
\clearpage

\pagenumbering{arabic}

\section{Introduction}
\label{sec:introduction}

In the $(s,t)$-mincut problem, we are given an $n$-vertex $m$-edge graph $G=(V, E)$ with integer edge weights $w:E\rightarrow \mathbb{Z}_+$ bounded by $U$. The goal is to minimize the sum of the weight of edges whose removal make $s$ unable to reach $t$. Unless stated otherwise, the input graphs are assumed to be undirected throughout the paper.

The $(s,t)$-mincut problem and its dual---$(s,t)$-maxflow---are among the most fundamental tools in graph algorithms and optimization. %
In particular, many reductions have been recently developed to show that if  $(s,t)$-mincut (equivalently, $(s,t)$-maxflow) can be solved in almost or nearly linear time, then so are a number of fundamental graph problems. 
These problems include vertex connectivity~\cite{LiNPSY21} and Gomory-Hu tree~\cite{AbboudKLPST21-GHtreeSubcubic} in unweighted graphs, deterministic global mincut and Steiner mincut~\cite{LiP20}, edge connectivity augmentation and edge splitting-off~\cite{CenLP22}, and hypergraph global mincut~\cite{ChekuriQ21,MukhopadhyayN21-submodular}.

All these results require {\em exact} $(s,t)$-mincut algorithms. In other words, these reductions cannot exploit {\em approximate} $(s, t)$-mincut algorithms which can offer many advantages. 
For example, while the best randomized {\em $(1+\epsilon)$-approximate} $(s, t)$-mincut algorithm takes {\em nearly-linear}\footnote{By nearly-linear time, we mean a running time of $\tO(m)$.} time on weighted graphs~\cite{Peng16} (and {\em almost-linear}\footnote{By almost-linear time, we mean a running time of $m^{1+o(1)}$.} time for deterministic algorithms~\cite{Sherman13,KelnerLOS14}), the fastest {\em exact} algorithms
require $\tilde O\left(\min(m+n^{3/2}, m^{\frac{3}{2}-\frac{1}{328}},m^{4/3+o(1)}U^{1/3})\right)$ time \cite{GaoLP21,LiuS20,BrandLLSSSW21minimum,BrandLNPSS0W20-matching}\footnote{Throughout, we use $\tilde O$ to hide $\poly\log(n)$.} and are all inherently randomized.\footnote{\label{foot:indep result}In an independent result \cite{ChenKLPGS22}, an almost-linear time randomized algorithm has been shown for the $(s,t)$-mincut problem. Even when this independent result is taken into account, the best $(1+\epsilon)$-approximation algorithms are still superior to the best exact algorithm with respect to time complexities and randomness requirements.}

Moreover, in many popular models of computation such as parallel computing, distributed computing, etc., computing exact $(s,t)$-mincut is still far from efficient, and using approximation algorithms might be the only alternative. For example, it is known that the $(1+\epsilon)$-approximation algorithm  (implied by \cite{chang2019improved,Sherman13}) on undirected unweighted graphs requires almost-linear work and sub-polynomial depth in PRAM. In contrast, we are far from emulating this result for exact algorithms. In fact, the first small step toward solving exact $(s,t)$-mincut with  almost-linear work and sub-polynomial depth would be doing so for the much simpler problem of $(s,t)$-reachability. And, the latter would involve  
breaking a major $\Omega(\sqrt{n})$ depth barrier.%
\footnote{This is due to the reduction from directed maxflow to undirected maxflow (see e.g.~\cite{madry2011graphs}) which works in the parallel setting. The reduction implies that if we can solve $(s,t)$-mincut exactly on undirected unweighted graphs in $O(W)$ work and $O(D)$ depth, then we can solve $(s,t)$-mincut exactly on {\em directed} unweighted graphs with $\tilde O(W)$ work and $\tilde O(D)$ depth. The latter captures the $st$-reachability problem as a special case.}
Another example is in the distributed setting (the CONGEST model), where a nearly optimal algorithm for computing $(1+\epsilon)$-approximate $(s,t)$-mincut exists \cite{ghaffari2015near} while no nontrival algorithm is known for the exact version.  
These advantages of approximate $(s,t)$-mincuts motivate a natural question: 
{\em Can the existing reductions work with approximate $(s,t)$-mincut algorithms instead of the exact ones?}


To answer the above question, let us discuss first {\em why} many reductions work only with exact $(s,t)$-mincut. 
A crucial property of exact $(s,t)$-mincuts in undirected graphs that is used by these reductions (e.g., for Gomory-Hu tree, deterministic global mincut, Steiner cut, edge connectivity augmentation, and edge splitting-off) is the following {\em uncrossing property}: 
%
%
\begin{quote}
	(Uncrossing Property) For any vertices $s$ and $t$, let $X\subset V$ be an $(s,t)$-mincut. Then, for any $u,v\in X$, there exists $Y\subset X$ that is a $(u, v)$-mincut.
\end{quote}
The uncrossing property is very useful from an algorithmic perspective since it gives a natural recursive tool -- after finding an $(s,t)$-mincut, we can recurse on each side of the cut to find a $(u,v)$-mincut for every pair of vertices $(u,v)$ on the same side of the cut. Indeed, the uncrossing property is more generally true for symmetric, submodular minimization problems and is at the heart of most of the beautiful structure displayed by undirected graph cuts and other symmetric, submodular functions. 
The uncrossing property, however, does {\em not} hold for $(1+\epsilon)$-approximate mincuts in general. This is the main bottleneck that prevents these reductions from being robust to approximation. As a result, for these problems, we fail to exploit the benefits of $(1+\epsilon)$-approximate $(s,t)$-mincut algorithms.

\subsection{Our contributions} 
We subvert the above bottleneck by introducing a more robust notion of approximate mincuts called {\em fair cuts}. Informally, an $\alpha$-fair $(s,t)$-cut is an $(s,t)$-cut such that there exists an $(s,t)$-flow $f$ that uses $1/\alpha$ fraction of the capacity of every edge in the cut. (The reader should think of $\alpha$ as being close to $1$.) Formally:

\begin{defn}[Fair Cut]\label{def:fair cut intro}
	Let $G=(V,E)$ be an undirected graph with edge capacities $c\in\mathbb{R}_{>0}^{E}$.
	Let $s,t$ be two vertices in $V$. For any parameter $\alpha\ge 1$,
	we say that a cut $(S,T)$ is a \emph{$\alpha$-fair $(s,t)$-cut}
	if there exists a feasible $(s,t)$-flow $f$ such that $f(u,v)\ge\frac{1}{\alpha}\cdot c(u,v)$
	for every $(u,v)\in E(S,T)$ where $u\in S$ and $v\in T$.
\end{defn}

Observe that a $1$-fair $(s,t)$-cut is an exact $(s,t)$-mincut. Moreover, an $\alpha$-fair $(s,t)$-cut is also an $\alpha$-approximate $(s,t)$-mincut. However, {\em not} all $\alpha$-approximate $(s,t)$-mincuts are $\alpha$-fair $(s,t)$-cuts.\footnote{As a simple example, consider a path $v-s-t$ on three vertices. Clearly, the cut $\{s\}$ contains both edges and is therefore a $2$-approximate $(s,t)$-mincut. However, there is no $(s,t)$-flow that can saturate both edges to fraction $\frac{1}{2}$. To motivate our choice of terminology (fair cuts), note that if an $(s,t)$-cut is a $\alpha$-approximate $(s,t)$-mincut, it follows by flow-cut duality that any $(s,t)$-maxflow will {\em cumulatively} saturate the edges of the cut to a fraction $\ge \frac{1}{1+\alpha}$. But, as we saw in the previous example, this saturation need not be {\em fair} in the sense that some edges might not be saturated at all. In this context, a $\alpha$-fair cut demands the additional property that {\em each} edge be saturated to a fraction $\ge \frac{1}{\alpha}$ (in the sense of ``max-min'' fairness).}
In other words, a set of $\alpha$-fair cuts is a proper subset of $\alpha$-approximate cuts and a superset of exact $(s,t)$-mincuts. 

We show that the notion of fair cuts allow us to combine the key features of both approximate cuts and exact cuts.  
First, fair cuts admit a property for approximate cuts that is analogous to uncrossing for exact mincuts, which we prove in \Cref{sec:uncrossing} for completeness.
\begin{restatable}[Approximate Uncrossing Property]{lemma}{Uncrossing}\label{lem:uncrossing-property}
	For any vertices $s$ and $t$, let $(S,T)$ be an $\alpha$-fair $(s,t)$-mincut. Then, for any $u,v\in S$, there exists $R\subset S$ such that $(R,V\setminus R)$ is an $\alpha$-approximate $(u,v)$-mincut.
\end{restatable}%
Second,  while computing a fair cut can be harder than an approximate mincut (since any fair cut is an approximate mincut but not vice-versa),
we give a nearly-linear time algorithm for computing a $(1+\epsilon)$-fair $(s,t)$-mincut.

\begin{theorem}[Fair Cut]\label{thm:fair}
	Given a graph $G=(V,E)$, two vertices
	$s,t\in V$, and $\epsilon\in(0,1]$, we can compute with high probability
	a $(1+\epsilon)$-fair $(s,t)$-cut in $\Otil(m/\epsilon^{3})$ time.
\end{theorem}

We note that the only reason why our algorithm is randomized is because
we use the \emph{congestion approximator} by \cite{RackeST14,Peng16}. This
can be made deterministic based on an algorithm by \cite{Chuzhoy20det}, but the
running time would be $m^{1+o(1)}/\epsilon^3$ instead.
Moreover, we remark that although we will focus on $(1+\epsilon)$-fair $(s, t)$-cuts, the corresponding $(s, t)$-flow can be obtained from a fair cut in $\tO(m/\epsilon)$ time using a standard application of a $(1+\epsilon)$-approximate max-flow algorithm of Sherman~\cite{Sherman2017area}.

\subsection{Applications} 

We demonstrate the power of fair cuts by using it to improve the time complexity of several problems.

\paragraph{Gomory-Hu Tree.} The Gomory-Hu  (GH) tree is a compact representation of a $(u,v)$-mincut (and therefore, $(u,v)$-maxflow values) between every pair of vertices $(u,v)$ of a graph, and has a large number of applications.
It captures fundamental questions such as global, $(s,t)-$, and Steiner mincuts as special cases. There has been much progress on exact and approximation algorithms for this problem recently (e.g.,~\cite{LiP21,AbboudKT20focs,AbboudKT20soda,AbboudKT21,AbboudKLPST21-GHtreeSubcubic,AbboudKT21-focs,AbboudKT22-soda,LiPS21,Zhang-simple,Zhang-weighted}). The fastest among these is the {\em $(1+\epsilon)$-approximation} algorithm by Li and Panigrahi \cite{LiP21} whose time complexity is equal to poly-logarthmic calls to any {\em exact} $(s,t)$-mincut algorithm, i.e. $\tilde O\left(\min(m+n^{3/2}, m^{\frac{3}{2}-\frac{1}{328}},m^{4/3+o(1)}U^{1/3})\right)$. 

By replacing the exact max-flow calls by our $(1+\epsilon)$-fair cut algorithm in~\cite{LiP21}, we get a nearly-linear time algorithm for approximating the Gomory-Hu tree (which is equivalent to finding all-pairs maxflow values by known reductions, e.g., \cite{AbboudKT20focs}):

\begin{theorem}[Nearly-linear time Gomory-Hu tree]
	\label{thm:ghtree}
	For any $\e > 0$, there is a $\tO(m\cdot \poly(1/\epsilon))$-time randomized algorithm that constructs, with high probability, a $(1+\e)$-approximate Gomory-Hu tree in weighted undirected graphs.
\end{theorem}

Prior to our work, a nearly-linear time (approximation) algorithm was not known even for the special case of the {\em Steiner mincut problem}. In this problem~\cite{DinitzV94,ColeH03,HariharanKP07,BHKP07,LiP20}, we are interested in finding a cut of minimum value that disconnects a given set of terminal vertices. For this problem, Li and Panigrahi~\cite{LiP20} gave an exact algorithm using poly-logarithmic exact max-flow calls. Before our work, no improvement in the running time was known if we allow $(1+\e)$-approximation instead of the exact Steiner mincut. Since the Steiner mincut problem is a minimal generalization of global and $(s,t)$-mincuts, our paper is the first to obtain nearly-linear time (approximation) algorithms for cut problems that go beyond these two problems.

\medskip\noindent{\em Parallelization.}
Since the use of exact max-flow is the only bottlenect to parallelize the approximate GH tree algorithm of  \cite{LiP21}, the following parallel algorithm also follows. 

\begin{theorem}[Parallel GH-tree]\label{thm:ghtree-parallel-intro}
	For any $\e > 0$, there is a $\tO(m^{1+o(1)}/ \poly(\epsilon))$-work $(m^{o(1)}/ \poly(\epsilon))$-depth randomized algorithm that constructs, with high probability, a $(1+\e)$-approximate Gomory-Hu tree in unweighted undirected graphs.
\end{theorem}

We are not aware of prior work on parallel GH algorithms (except some experiments, e.g. \cite{MaskeCD20,CohenRD17}). This is likely because previous GH trees algorithms, even the approximate ones \cite{LiP21}, inherently require solving max-flow exactly, which is well beyond current techniques in the parallel setting.

\paragraph{Expander Decomposition.}
In the last decade, numerous fast graph algorithms are based on fast algorithms for computing an expander decomposition. For some examples of such applications, see e.g.~\cite{SpielmanT04,KelnerLOS14,Sherman13,NanongkaiSW17,chang2019improved,bernstein2020fully}.

We say that a (weighted) graph $G = (V,E)$ is a $\phi$-expander if for every cut $(S,V\setminus S)$, we have that the cut size $\delta(S) \ge \min\{\vol(S),\vol(V\setminus S)\}$ where the volume of $S$ is $\vol(S)=\sum_{v \in S} \deg(v)$.
A $(\epsilon,\phi)$-expander decomposition of $G$ is a partition $\{V_1,\dots,V_k\}$ of vertices such that each $G[V_i]$ is a $\phi$-expander and $\sum_i \delta(V_i) \le \epsilon \cdot \vol(V)$, i.e., the total weight of edges crossing the partition is at most $\epsilon$-fraction.

There are two incomparable fastest algorithms for computing expander decompositions. First, \cite{NanongkaiS17,Wulff-Nilsen17} gave $m^{1+o(1)}$-time algorithms that computes a $(\phi n^{o(1)},\phi)$-expander decomposition for any $\phi>0$. These subpolynomial factors are sometimes undesirable.
Second, \cite{SaranurakW19} gave a $\tO(m/\phi)$-time algorithm that computes a $(\tO(\phi),\phi)$-expander decomposition for any $\phi>0$. This algorithm is slower than the first one when $\phi < 1/n^{0.1}$.
Using fair cuts, we obtain an algorithm that subsumes both these sets of results and is optimal up to poly-logarithmic factors. 

\begin{theorem}[Near-linear expander decomposition]
	For any $\phi>0$, there is a randomized $\tilde{O}(m)$-time algorithm that with high probability computes a $(\tO(\phi),\phi)$-expander in weighted undirected graphs.
\end{theorem}

\paragraph{Open problems.}  We believe that our notion of fair cuts opens up many interesting directions for future research. We mention some examples. (i) A natural goal is to extend our efficient $(1+\epsilon)$-fair $(s,t)$-cut to other computational models, such as the distributed (CONGEST) setting, where exact $(s,t)$-mincut algorithms are much slower/inefficient compared to approximate $(s,t)$-mincut algorithms. This will lead to efficient algorithms for approximating, e.g., Gomory-Hu tree and Steiner mincut in these models as well. 
(ii) The notion of {\em fair vertex cuts} can be defined in a similar fashion to fair (edge) cuts defined in this paper. It would be interesting to design an efficient algorithm for finding a fair vertex cut and use it to obtain nearly-linear time algorithms for approximating the vertex connectivity and hypergraph global mincut. These results can also be extended to other computational models. (iii) We also hope that the notion of fair cuts can be extended to more general contexts such as the minimization of symmetric, submodular functions. In turn, this will significantly improve our understanding of the approximation-efficiency tradeoff in minimization problems defined for these function classes.

\paragraph{Independent result.} Our result is obtained independently from the recently announced almost-linear time bound for min-cost flow by Chen, Kyng, Liu,  Peng, Gutenberg, and Sachdeva~\cite{ChenKLPGS22}. 
Plugging this result into existing reductions in  \cite{AbboudKLPST21-GHtreeSubcubic,LiNPSY21,CQ20,MukhopadhyayN20} help solve problems such as GH tree and vertex connectivity in unweighted graphs, approximate GH tree in weighted graphs, and hypergraph global mincut in $m^{1+o(1)}$ time. 
Even assuming this result, our algorithms are faster in both randomized and deterministic settings; for the latter, our running time is $m^{1+o(1)}$ whereas the best exact $(s,t)$-mincut algorithm takes $\tO(m \min(\sqrt{m}, n^{2/3}))$ time~\cite{GoldbergR98}. Finally, our algorithms can be adapted to other models such as parallel computation whereas this is well beyond existing techniques for exact $(s,t)$-mincut.

\section{Overview of Techniques}

\subsection{Computing Fair Cuts (Proof Idea of \Cref{thm:fair})}

Our key subroutine for computing fair cuts is called $\almostfair$.
Here, we describe at a high-level what the $\almostfair$ subroutine does,
how to use it for computing fair cuts, and finally how to obtain the $\almostfair$
subroutine itself.

Say we are given an $(s,t)$-cut $(S,T)$ which may be far from being
fair. The $\almostfair$ subroutine works on one side of the $(s,t)$-cut,
say $T$, and returns a partition $(P_{t},T')$ of $T$ such that $t\in T'$.
We think of $P_{t}$ as the part that we ``prune'' out of $T$.
Our first guarantee is that the remaining part $T'$ is ``almost
fair'' in the following sense: each boundary edge in $E(S,T')$ (i.e., those edges that are not in $E(P_t,T')$) can simultaneously send flow 
of value at least $(1-\beta)$-fraction of its capacity to $t$, for a small parameter
$\beta$ that we can choose. This guarantee alone would have been
weak if the pruned set $P_{t}$ is so big that there are few edges left
in $E(S,T')$. However, the second guarantee of $\almostfair$ says
that, if $P_{t}$ is big, then $(V\setminus T',T')$ is actually a
much smaller $(s,t)$-cut than the original cut $(S,T)$ in terms of cut value. More
formally, we have $\delta_{G}(T')\le\delta_{G}(T)-\beta\cdot \delta_{G}(S,P_{t})$
meaning that the decrease in the cut size is at least $\beta$ times
the total capacity of $E(S,P_{t})$. 

With these two guarantees of $\almostfair$, given any $(s,t)$-cut
$(S,T)$, we can iteratively improve this cut to make it fair as
follows. We call $\almostfair$ on both $S$ and $T$ and obtain $(P_{s},S')$
and $(P_{t},T')$. Let's consider two extremes. If both pruned sets
$P_{s}$ and $P_{t}$ are tiny, then there is an $(s,t)$-flow that
almost fully saturates \emph{every} edge in $E(S',T')$. This certifies
$(S,T)$ is very close to being fair as $P_{s}$ and $P_{t}$ are
tiny. However, if either $P_{s}$ or $P_{t}$ is very big, say $P_{t}$,
then $(S\cup P_{t},T')$ is an $(s,t)$-cut of much smaller value than
the original cut $(S,T)$. Therefore, this is progress too and we
can recursively work on this new cut $(S\cup P_{t},T')$. To make the
intuition on these two extremes work, we iteratively call $\almostfair$
using a parameter $\beta$ that increases slightly in every round. The full algorithm
is presented in \Cref{sec:from_almost_to_fair}.

Now, let us sketch the $\almostfair$ subroutine itself.
This subroutine is based on Sherman's algorithm for
computing a $(1+\epsilon)$-approximate max-flow~\cite{Sherman13} (for any $\epsilon > 0$), 
which in turn uses the
multiplicative weight update (MWU) framework.\footnote{Sherman's 
original presentation in \cite{Sherman13} does not explicitly use the MWU framework.
Although this alternative interpretation was already known to experts, 
our MWU-based presentation of his algorithm is arguably simpler and more intuitive.} 
Given the $t$-side $T$ of an $(s,t)$-cut, if we call Sherman's
algorithm where the demand is specified so that each boundary edge
should send flow at its full capacity to sink $t$, then the algorithm
would either return a flow satisfying this demand with congestion
$(1+\epsilon)$ or return a ``violating'' cut certifying that the
demand is not feasible. In the former case,
this would satisfy the guarantee of $\almostfair$ where $P_{t}=\emptyset$
after scaling down the flow by a $(1+\epsilon)$ factor. Unfortunately,
in the latter case, the algorithm does not guarantee the existence
of the flow that we want. The reason behind this problem is that whenever
the algorithm detects a violating cut, the algorithm is just terminated.
In a more general context, this holds for most (if not all) MWU-based
algorithms for solving linear programs; in each round of the MWU algorithm,
whenever ``the oracle'' certifies that the linear program is infeasible, then we just terminate the whole algorithm. 

Interestingly, we fix this issue by ``insisting on continuing''
the MWU algorithm. Once we detect a violating cut, we include it into
the pruned set, cancel the demand inside this pruned set, and continue
updating weights in the MWU algorithm. After the last round, the
flow constructed via MWU indeed sends flow from each remaining boundary edge
that is not pruned out, which is exactly our goal. The detailed algorithm
is presented in \Cref{sec:almost_fair}.

\subsection{From Fair Cuts to Approximate Isolating Cuts}
We believe that the notion of fair cuts can be useful in many contexts
since it offers a more robust alternative to approximate mincuts. In this 
paper, we first use it to obtain an approximate isolating cuts algorithm.
We define the isolating cuts problem first.

\begin{defn}
Given a weighted, undirected graph $G = (V, E)$ and a subset of terminals 
$S = \{s_1, s_2, \ldots, s_k\}$, the goal of the isolating cuts problem is to find
a set of disjoint sets $S_1, S_2, \ldots, S_k$ such that for each $i$,
the cut $(S_i, V\setminus S_i)$
is a mincut that separates $s_i\in S_i$ from the remaining terminals 
$S\setminus \{s_i\}\subseteq V\setminus S_i$. If $S_i$ is a
$(1+\epsilon)$-approximate mincut separating $s_i$ from the remaining 
terminals, then the corresponding problem is called the $(1+\epsilon)$-approximate
isolating cuts problem.
\end{defn}

Using fair cuts, we obtain a near-linear time algorithm for approximate
isolating cuts.

\begin{theorem}
\label{thm:iso}
There is an algorithm for finding $(1+\epsilon)$-approximate isolating cuts 
that takes $\tO(m\cdot \poly(1/\epsilon))$ time.
\end{theorem}

Li and Panigrahi~\cite{LiP20} gave an algorithm for finding {\em exact} 
isolating cuts using $O(\log n)$ $(s,t)$-max-flow/mincut computations that crucially relies 
on the uncrossing property of $(s,t)$-mincuts. This property ensures
that if we take a minimum isolating cut $X$ containing a terminal
vertex $s$ and a crossing mincut $Y$, then their
intersection $X\cap Y$ or difference $X\setminus Y$ (depending on which set 
the terminal vertex $s$ is in) is also a minimum isolating cut. 
This allows partitioning of the 
graph by removing edges corresponding to a set of mincuts, such that each 
terminal is in one of the components of this partition. For each terminal,
the corresponding minimum isolating cut is now obtained by simply contracting 
the rest of the components and running a max-flow algorithm on this contracted
graph. The advantage of this contraction is that the total size of all the graphs 
on which we are running max-flows is only a constant times the size of the overall
graph.

Unfortunately, approximate mincuts don't satisfy this uncrossing property,
which renders this method unusable if we replace exact mincut subroutines
with faster $(1+\epsilon)$-approximate mincuts. But, if we instead used 
fair cuts, then we can show the following: if $X$ is a $(1+\epsilon)$-approximate
minimum isolating cut containing terminal $s$ and $Y$ is a $(1+\alpha)$-fair cut, 
then either $X\cap Y$ or $X\setminus Y$ (whichever set contains $s$) is a 
$(1+\epsilon)(1+\alpha)$-approximate minimum isolating cut. This allows
us to use the framework in \cite{LiP20}. Since the number of fair cuts
we remove in forming the components is only $O(\log n)$, the multiplicative
growth in the approximation factor can be offset by scaling the 
parameter in fair cuts by the same logarithmic factor. The advantage 
in using fair cuts over exact mincuts is that the running time of the 
former is near-linear by \Cref{thm:fair}, which helps establish 
\Cref{thm:iso}. The details of this algorithm are presented in 
\Cref{sec:isolating}.

\subsection{From Approximate Isolating Cuts to Approximate GH-trees}

Finally, we use approximate isolating 
cuts to obtain an approximate GH tree algorithm. 
\cite{LiP21} gives a recursive algorithm for computing 
an approximate GH tree but using exact isolating cuts.
We observe that the latter can be replaced by approximate
isolating cuts provided
the approximation is {\em one-sided} in the following sense:
the ``large'' recursive subproblem needs to preserve mincut 
values exactly. But, in general, 
if we use the approximate isolating cuts subroutine
as a black box, this would not be the case. To alleviate
this concern, we augment the approximate isolating cuts 
procedure using an additional fairness condition for 
the isolating cuts returned by the algorithm. This fairness
condition ensures that although we 
do not have one-sided approximation, the approximation 
factor in the ``large'' subproblem can be controlled using
a much finer parameter than the overall approximation
factor of the algorithm, which then allows us to run 
the recursion correctly. The details of the GH tree 
algorithm establishing \Cref{thm:ghtree} are presented in 
\Cref{sec:ghtree}.

\subsection{From Fair Cuts to Near-linear time Expander Decomposition}
Via fair cuts, we will speed up the algorithm by \cite{SaranurakW19} for computing a $(\Otil(\phi),\phi)$-expander decomposition in time $\Otil(m/\phi)$ to $\Otil(m)$. There are two main steps in the algorithm by \cite{SaranurakW19}: the cut-matching step and the trimming step. The cut-matching step can be solved in $\tO(m)$ time simply by applying the near-linear-time approximate maxflow algorithm by \cite{Peng16}.\footnote{For reader who are familiar with \cite{SaranurakW19}, their algorithm applies the push-relabel flow algorithm that takes $\Otil(m/\phi)$ time, instead of using an $\Otil(m)$-time approximate max flow algorithm, because the push-relabel algorithm has fewer log factors in the running time.}
The harder step to speed up is the trimming step. 
However, we observe that the cut problem needed to be solved in this step is actually a one-sided version of the fair cut problem, which is an easier problem.\thatchaphol{Should I explain this more?} By calling our fair cut algorithm, we immediately obtain a $\Otil(m)$-time algorithm for the trimming step.  See details in \Cref{sec:expdecomp}.

\section{Preliminaries}

Given a undirected capacitated/weighted graph $G=(V,E)$ with edge capacities/weights
is $c\in\mathbb{R}_{\ge0}^{E}$ and an edge set $E'\subseteq E$, we let $c(E')=\sum_{e\in E'}c(e)$
be the total capacity of $E'$. 
For simplicity, we assume that the ratio between the largest and lowest edge capacities or weights are $\poly(n)$. For any disjoint sets $S,T\subseteq V$,
we let $\delta_{G}(S)=c(E(S,V\setminus S))$ denote the cut size of
$S$ and $\delta_{G}(S,T)=c(E(S,T))$ denote the total capacity of
edges from $S$ to $T$. For any distinct vertices $s$ and $t$, let $\lambda_G(s,t)$ be the minimum-weight $s$-$t$ cut. We sometimes omit $G$ when it is clear from the context.
\paragraph{Flow.}

A \emph{flow} $f:V\times V\rightarrow\mathbb{R}$
satisfies $f(u,v)=-f(v,u)$ and $f(u,v)=0$ for $\{u,v\}\notin E$.
The notation $f(u,v)>0$ means that mass is routed in the direction
from $u$ to $v$, and vice versa. The \emph{congestion} of $f$ is
$\max_{\{u,v\}\in E}\frac{|f(u,v)|}{c(e)}$. If the congestion is
at most $1$, we say that $f$ \emph{respects the capacity} or $f$
is \emph{feasible}. For each vertex $u\in V$, the \emph{net flow
out of vertex} $u$, denoted by $f(u)=\sum_{v\in V}f(u,v)$, is the
total mass going out of $u$ minus the total mass coming into $u$.
More generally, for any vertex set $S\subseteq V$, we can define
the \emph{net flow out of $S$ }as $f(S)=\sum_{u\in S}f(u)=\sum_{u\in S,v\in V}f(u,v)$.
The \emph{net flow out from $S$ to $T$} is denoted by $f(S,T)=\sum_{u\in S,v\in T}f(u,v)$.
Observe that we always have $f(V)=0$. 

A \emph{demand function }$\Delta:V\rightarrow\mathbb{R}$ is a function
where $\sum_{v\in V}\Delta(v)=0$. We say that flow $f$ \emph{satisfies}
demand $\Delta$ if $f(v)=\Delta(v)$ for all $v\in V$. For any $S\subseteq V$,
let $\Delta(S)=\sum_{v\in S}\Delta(v)$ be the \emph{total demand}
on $S$. Observe $\Delta(V)=f(V)=0$. By the max-flow min-cut theorem,
we have the following:
\begin{fact}
\label{fact:mfmc}For any $\epsilon\ge0$, $|\Delta(S)|\le\epsilon\cdot\delta(S)$
for all $S\subseteq V$ iff there is a flow with congestion $\epsilon$
that satisfies $\Delta$. 
\end{fact}

For a flow $f$ and a demand function $\Delta$, 
define the \emph{excess}
$\Delta^{f}$ as $\Delta^{f}(v)=\Delta(v)-f(v)$ for every $v\in V$.
We think of excess as a remaining demand function. We say that $f$
\emph{$\epsilon$-satisfies} $\Delta$ if $|\Delta^{f}(S)|\le\epsilon\cdot\delta(S)$
for all $S\subseteq V$. That is, by \Cref{fact:mfmc}, there exists a flow
$f_{aug}$ with congestion $\epsilon$ where $f+f_{aug}$ satisfies
$\Delta$. Note that $f$ $0$-satisfies $\Delta$ iff $f$ satisfies
$\Delta$. 

For any two vertices $s,t\in V$, an \emph{$(s,t)$-cut} $(S,T)$
is a cut such that $s\in S$ and $t\in T$. An \emph{$(s,t)$-flow}
$f$ obeys $f(v)=0$ for all $v\neq s,t$. Similarly, an \emph{$(s,t)$-demand
function} $\Delta$ obeys $\Delta(v)=0$ for all $v\neq s,t$. That
is, an $(s,t)$-demand
function is satisfied only by an $(s,t)$-flow.

\paragraph{Congestion Approximators.}

When we want to argue that flow $f$ $\epsilon$-satisfies a demand
function $\Delta$, it can be inconvenient to ensure that $|\Delta^{f}(S)|\le\epsilon\cdot\delta(S)$
for all $S\subseteq V$ because there are exponentially many sets.
Surprisingly, there is a collection $\mathcal{S}$ of linearly many
sets of vertices, where if $|\Delta(S)|\le\epsilon\cdot\delta(S)$ for each $S\in{\cal S}$,
then this is also true for all $S\subseteq V$ with some $\polylog(n)$
blow-up factor. Moreover, $\mathcal{S}$ can be computed in near-linear
time.
\begin{theorem}
[Congestion approximator \cite{RackeST14,Peng16}]\label{thm:congest}There
is a randomized algorithm that, given a graph $G=(V,E)$ with $n$
vertices and $m$ edges, constructs in $\tilde{O}(m)$ time with high probability a laminar
family $\mathcal{S}$ of subsets of $V$ such that 
\begin{enumerate}
\item ${\cal S}$ contains at most $2n$ sets, 
\item each vertex appears in $O(\log n)$ sets of ${\cal S}$, and
\item for any demand function $\Delta$ on $V$, if $|\Delta(S)|\le\delta(S)$
for all $S\in\mathcal{S}$, then $|\Delta(R)|\le\gamma_{\cal S}\delta(R)$
for all $R\subseteq V$ for a \emph{quality} factor $\gamma_{\mathcal{S}}=O(\log^{4}n)$.
\end{enumerate}
\end{theorem}

\paragraph{Graphs with Boundary Vertices. }

Given a set $U\subseteq V$, let $G\{U\}$ denote the following ``\emph{induced
subgraph with boundary vertices}'': 
start with induced subgraph $G[U]$, and for each edge $e=(u,v)\in E(U,V\setminus U)$
with endpoint $u$ in $U$, create a new vertex $x_{e}$ and add the
edge $(x_{e},u)$ to $G\{U\}$ of the same capacity as $e$. Let $N_{G}\{U\}$
be the vertex set of $G\{U\}$ and define $N_{G}\langle U\rangle=N_{G}\{U\}\setminus U$.
We call vertices in $N_{G}\langle U\rangle$ \emph{boundary vertices}.
We simply write $N\{U\}$ and $N\langle U\rangle$ instead of $N_{G}\{U\}$
and $N_{G}\langle U\rangle$ when the context is clear. Observe that
the degree $\deg_{G\{U\}}(x_{e})$ of each boundary vertex $x_{e}\in N\langle U\rangle$
in $G\{U\}$ is simply the capacity $c(e)$ of edge $e$. We will
use this notation very often in the paper.

\paragraph{Boundary Demand Functions. }

In our context, the sink node $t\in U$ is usually given. The \emph{full
$U$-boundary demand function $\Delta_{U}:V(G\{U\})\rightarrow\mathbb{R}$
}is defined such that
\[
\Delta_{U}(v)=\begin{cases}
\deg_{G\{U\}}(v) & \mbox{if }v\in N\langle U\rangle\\
0 & \mbox{if } v\in U\setminus t\\
-\Delta(N\langle U\rangle) & \mbox{if } v=t.
\end{cases}
\]
That is, any flow satisfying $\Delta_{U}$ sends flow from each boundary
vertex of $G\{U\}$ at full capacity to $t$. We also write $\Delta_{U,t}$
when it is not clear from the context what $t$ is. More generally,
given any demand function $\Delta':V(G\{U\})\rightarrow\mathbb{R}$,
we say that $\Delta'$ is a \emph{$U$-boundary demand function }if
$\Delta'(v)=0$ for all $v\in U\setminus t$, $\Delta'(t)=-\Delta(N\langle U\rangle)$.
That is, $\Delta'$ is completely determined once we specify the demand
values on boundary vertices $N\langle U\rangle$.

\paragraph{One-Sided Fair Cut.}
Finally, the following ``one-sided'' version of a fair cut (\Cref{def:fair cut intro}) will be useful.
\begin{defn}[One-sided Fair Cut]
Let $G=(V,E)$ be an undirected graph with edge capacities $c\in\mathbb{R}_{>0}^{E}$.
Let $s$ be a vertex in $V$. For any parameter $\alpha\ge 1$,
we say that a cut $(S,T)$ is an \emph{$s$-sided $\alpha$-fair cut}
if there exists a feasible flow $f$ such that
 \begin{enumerate}
     \item $f(v)=0$ for all $v\in S\setminus\{s\}$
     \item $f(u,v)\ge\frac{1}{\alpha}\cdot c(u,v)$
for every $(u,v)\in E(S,T)$ where $u\in S$ and $v\in T$.
 \end{enumerate}
In other words, the flow $f$ sends flow from $s$ to the boundary $E(S,T)$ in a way that almost saturates every edge in $E(S,T)$, but we do not care about the behavior of $f$ beyond $E(S,T)$.
\end{defn}

Clearly, an $\alpha$-fair $(s,t)$-cut is an $s$-sided $\alpha$-fair cut since we can take the same flow $f$ that witnesses the $\alpha$-fair $(s,t)$-cut. However, we will only require the one-sided version in our isolating cuts application in \Cref{sec:isolating}.

\section{Almost Fair Cuts via Multiplicative Weight Updates}
\label{sec:almost_fair}

The key subroutine used for proving \Cref{thm:fair} is the algorithm
below. 
\begin{theorem}
[Almost Fair Cuts]\label{thm:almost fair}There is an algorithm $\almostfair(G,U,t,\epsilon,\beta)$
that, given a graph $G=(V,E)$ with a sink node $t\in V$, a set $U\subseteq V$
where $t\in U$, and parameters $\beta\ge0$ and $\epsilon>0$, returns
a partition $(P,U')$ of $U$ where $t\in U'$ with the following
properties:
\begin{enumerate}
\item \label{thm:almost fair:cut}
 $\delta_{G}(U')\le\delta_{G}(U)-\beta\delta_{G}(P,V\setminus U)$ (equivalently, $\delta_{G}(P,U')\le(1-\beta)\delta_{G}(P,V\setminus U)$),
and
\item \label{thm:almost fair:flow}There exists a flow $f'_{sat}$ in $G\{U'\}$
with congestion $(1+\epsilon)$ satisfying a $U'$-boundary demand
function $\Delta'$ such that 
\begin{align*}
\Delta'(v) & =(1-\beta)\deg_{G\{U'\}}(v) & \text{for all old boundary vertices }v\in N\langle U'\rangle\cap N\langle U\rangle\\
|\Delta'(v)| & \le(1+\epsilon)\deg_{G\{U'\}}(v) & \text{for all new boundary vertices }v\in N\langle U'\rangle\setminus N\langle U\rangle
\end{align*}
\end{enumerate}
The algorithm takes $\Otil(|E(G\{U\})|/\epsilon^{2})$ time and is
correct with high probability.\footnote{Note that the guarantee that $|\Delta'(v)| \le(1+\epsilon)\deg_{G\{U'\}}(v)$  for all new boundary vertices $v\in N\langle U'\rangle\setminus N\langle U\rangle$ in fact follows from the guarantee that $f'_{sat}$ has congestion $(1+\epsilon)$. We state both guarantees explicitly for convenience.}
\end{theorem}

The rest of this section is for proving \Cref{thm:almost fair}. For
convenience, we write $H=G\{U\}$ and let $n$ and $m$ denote the
number of vertices and edges in $H$ throughout this section. Let
$B$ be the incidence matrix of $H$. Observe that, for any flow $f$
on $H$, we have $(Bf)_{v}=f(v)$ is the net flow out of $v$. We
can view $Bf$ as a vector in $\mathbb{R}^{V(H)}$. Define 
\[
\Delta=(1-\beta)\Delta_{U}
\]
as the full $U$-boundary demand function on $G\{U\}$ after scaled
down by $(1-\beta)$ factor. For any $U'\subseteq U$, the
\emph{restriction} $\Delta|_{U'}$ of $\Delta$ is a $U$-boundary
demand function obtained from $\Delta$ by zeroing out the entries
on $\boundary U\setminus\boundary{U'}$, i.e., the boundary vertices
of $U$ which are not boundaries of $U'$, and then setting the
entry on $t$ so that $\sum_{v\in V(H)}\Delta|_{U'}(v)=0$. Similarly,
we view $\Delta$ and also $\Delta|_{U'}$ as vectors in $\mathbb{R}^{V(H)}$.

\subsection{Algorithm}
\label{sec:alg almostfair}
\paragraph{Initialization.}

We start by computing a congestion approximator ${\cal S}$ of $H$
with quality $\gamma_{{\cal S}}=O(\log^{4}n)$ using \Cref{thm:congest}.
For a technical reason, it is more convenient if no set in ${\cal S}$
contains sink $t$. 
From now, we will assume this, which is justified by the following observation:
\begin{prop}
\label{prop:no t}
Given the family ${\cal S}$ from \Cref{thm:congest} and a vertex $t$, there is a linear time algorithm that returns another family ${\cal S}'$ with the same guarantee as in \Cref{thm:congest} but with additional guarantee that each set $S \in {\cal S}'$ does not contains $t$.
\end{prop}

\begin{proof}
Replace each set $S\in{\cal S}$ where $t\in S$ with its complement
$V(H)\setminus S$. Observe that ${\cal S}$ is now a larminar family
on $V(H)\setminus t$ where $|{\cal S}|$ does not change, and the
number of sets containing each vertex may increase only by $O(\log n)$.
Hence, the first and second properties of \Cref{thm:congest} still
hold. The third property still holds as well because $|\Delta(S)|=|\Delta(V(H)\setminus S)|$
for all $S$.
\end{proof}
Our algorithm is based on the Multiplicative Weight Update framework
and so it works in rounds. For round $i$, we maintain \emph{weights}
$w_{S,\circ}^{i}\ge0$ for each $S\in\mathcal{S}$ and $\circ\in\{+,-\}$
and define the potential $\phi^{i}\in\mathbb{R}^{V(H)}$ where
\[
\phi_{v}^{i}=\sum_{S\ni v}\frac{1}{\delta_{H}(S)}(w_{S,+}^{i}-w_{S,-}^{i})
\]
for each vertex $v$. As no set $S\in{\cal S}$ contains $t$, we
will always have $\phi_{t}^{i}=0$ for all $i$. Initially, we set
$w_{S,\circ}^{1}=1$ for all $S\in\mathcal{S}$, $\circ\in\{+,-\}$.

The algorithm also maintains a decremental subset $V^{i}$ where $t\in V^{i}\subseteq V^{i-1}$
for all $i$. We initialize $V^{0}$ as follows. First, set $V^{0}=V(H)$.
While there exists $S\in\mathcal{S}$ where $\Delta|_{V^{0}}(S)>\delta_{H}(S)$,
which certifies that there is no feasible flow on $H$ satisfying $\Delta|_{V^{0}}$
by \Cref{fact:mfmc}, we update $V^{0}\gets V^{0}\setminus S$ (in
particular, the function $\Delta|_{V^{0}}$ changes). Let $D^{0}$ contain all the vertices we removed from $V^{0}$. Now, we are ready to state
the main algorithm.

\paragraph{Main Algorithm.}

For round $i=1,2,\ldots,T$ where $T=\Theta(\log(n)/\alpha^{2})$
and $\alpha=\epsilon/\gamma_{{\cal S}}$, we do the following:
\begin{enumerate}
\item Define $f^{i}$ on $H$ such that for each edge $(u,v)$, $f^{i}(u,v)$
flows from high potential to low potential at maximum capacity. That
is, for every edge $(u,v)$ in $H$,
\[
f^{i}(u,v)=\begin{cases}
c(u,v) & \text{if }\phi_{u}^{i} > \phi_{v}^{i}\\
0 & \text{if }\phi_{u}^{i} = \phi_{v}^{i}\\
-c(u,v) & \text{if }\phi_{u}^{i}<\phi_{v}^{i}.
\end{cases}
\]
\item Using \Cref{lem:1}, compute a \emph{deletion set} $D^{i}\subseteq V(H)\setminus t$
and set $V^{i}\gets V^{i-1}\setminus D^{i}$, where $D^{i}$ satisfies
the following:\label{Di} 
\begin{align*}
\text{if }D^{i}\neq\emptyset,\text{ then }\Delta|_{V^{i-1}}(D^{i}) & >\delta_{H}(D^{i})\text{, and}\\
\langle\phi^{i},\Delta|_{V^{i}}\rangle=\langle\phi^{i},\Delta|_{V^{i-1}\setminus D^{i}}\rangle & \le\langle\phi^{i},Bf^{i}\rangle.
\end{align*}
\item For each $S\in{\cal S}$, let 
\[
r_{S}^{i}=\frac{(\Delta|_{V^{i}})^{f^{i}}(S)}{\delta_{H}(S)}=\frac{\Delta|_{V^{i}}(S)-f^{i}(S)}{\delta_{H}(S)}
\]
be the \emph{relative total excess} at $S$ compared to the cut size
in round $i$.
\item Update the weights \label{enu:Update weights}
\[
w_{S,+}^{i+1}=w_{S,+}^{i}\cdot e^{\alpha r_{S}^{i}}\quad\text{and}\quad w_{S,-}^{i+1}=w_{S,-}^{i}\cdot e^{-\alpha r_{S}^{i}}.
\]
\end{enumerate}
After $T$ rounds, we compute the \emph{pruned set} $P=\cup_{i=0}^{T}D^{i}$
and let $U'=U\setminus P$. Finally, we return the partition $(P,U')$.

\subsection{Correctness}

We prove that the partition $(P,U')$ outputted by our algorithm satisfies
the requirement in \Cref{thm:almost fair}. The first important thing
to understand our algorithm is to formally see how it is captured
by the Multiplicative Weight Update (MWU) algorithm, which we recall
below:
\begin{theorem}
[Multiplicative Weights Update \cite{AroraHK12}]\label{thm:mwu} Let $J$
be a set of indices, and let $\alpha\le1$ and $\omega>0$ be parameters.
Consider the following algorithm: 
\begin{enumerate}
\item Set $w_{j}^{(1)}\gets1$ for all $j\in J$ 
\item For $i=1,2,\ldots,T$ where $T=O(\omega^{2}\log(|J|)/\alpha^{2})$: 
\begin{enumerate}
\item The algorithm is given a ``gain'' vector $g^{i}\in\mathbb{R}^{J}$
satisfying $\|g^{i}\|_{\infty}\le\omega$ and $\langle g^{i},w^{i}\rangle\le0$\label{mwu1} 
\item For each $j\in J$, set $w_{j}^{i}\gets w_{j}^{i-1}\exp(\alpha g_{j}^{i})=\exp(\alpha\sum_{i'\in[i]}g_{j}^{i'})$ 
\end{enumerate}
\end{enumerate}
At the end of the algorithm, we have $\frac{1}{T}\sum_{i\in[T]}g_{j}^{i}\le\alpha$
for all $j\in J$.\footnote{More generally, for any value $\mathsf{val}$, if we have $\langle g^{i},w^{i}\rangle\le\mathsf{val}$
for all $i$, the MWU algorithm guarantees that $\frac{1}{T}\sum_{i\in[T]}g^{i}_j\le\mathsf{val}+\alpha$, for all $j$.
Here, we use a special case when $\mathsf{val}=0$.}
\end{theorem}

To apply \Cref{thm:mwu} into our setting, we define $J=\mathcal{S}\times\{+,-\}$.
That is, we work with indices $(S,+)$ and $(S,-)$ for $S\in\mathcal{S}$.
We use the same weights $w^{i}$ and error parameter $\alpha$ as
the algorithm, and we set $\omega=2$. For each iteration $i$ and
$S\in\mathcal{S}$, we define 
\[
g_{S,\pm}^{i}=\pm r_{S}^{i}=\pm\frac{\Delta|_{V^{i}}(S)-f^{i}(S)}{\delta_{H}(S)}.
\]
Observe that the weights $w_{S,\pm}^{i}$ are updated in Step \ref{enu:Update weights}
exactly as $w_{S,\pm}^{i}\gets w_{S,\pm}^{i-1}\exp(\alpha g_{S,\pm}^{i})$.
With this setting, we show that our gain vector $g^{i}$ indeed satisfies
the condition in Step \ref{mwu1} of \Cref{thm:mwu}.
\begin{lem}
For each $i$, we have $\|g^{i}\|_{\infty}\le2$ and $\langle g^{i},w^{i}\rangle\le0$. 
\end{lem}

\begin{proof}
To show $\|g^{i}\|_{\infty}\le2$, we have 
\[
|g_{S,\pm}^{i}|=\left|\frac{\Delta|_{V^{i}}(S)-f^{i}(S)}{\delta_{H}(S)}\right|\le\left|\frac{\Delta|_{V^{i}}(S)}{\delta_{H}(S)}\right|+\left|\frac{f^{i}(S)}{\delta_{H}(S)}\right|\le1+1,
\]
To see why the last inequality holds, we have (1) $\Delta|_{V^{0}}(S)\le\delta_{H}(S)$
for all $S\in{\cal S}$ by the initialization of $V^{0}$, (2) $\Delta|_{V^{i}}(S)\ge0$
for all $i$ because $t\notin S$, and (3) $\Delta|_{V^{i}}(S)$ may
only decrease as $V^{i}$ is a decremental set. Also, we have $|f^{i}(S)|\le$$\delta_{H}(S)$
because each $f^{i}$ respects the capacity.

To show $\langle g^{i},w^{i}\rangle\le0$, first observe that $\langle g^{i},w^{i}\rangle=\langle\phi^{i},\Delta|_{V^{i}}\rangle-\langle\phi^{i},Bf^{i}\rangle$
exactly.
\begin{align*}
\langle g^{i},w^{i}\rangle & =\sum_{S\in\mathcal{S}}(g_{S,+}^{i}w_{S,+}^{i}+g_{S,-}^{i}w_{S,-}^{i})\\
 & =\sum_{S\in\mathcal{S}}(w_{S,+}^{i}-w_{S,-}^{i})r_{S}^{i}\\
 & =\sum_{S\in\mathcal{S}}\frac{w_{S,+}^{i}-w_{S,-}^{i}}{\delta_{H}(S)}\left(\Delta|_{V^{i}}(S)-f^{i}(S)\right)\\
 & =\sum_{S\in\mathcal{S}}\frac{w_{S,+}^{i}-w_{S,-}^{i}}{\delta_{H}(S)}\sum_{v\in S}\left(\Delta|_{V^{i}}(v)-(Bf^{i})_{v}\right)\\
 & =\sum_{v\in V(H)}\left(\Delta|_{V^{i}}(v)-(Bf^{i})_{v}\right)\sum_{S\ni v}\frac{w_{S,+}^{i}-w_{S,-}^{i}}{\delta_{H}(S)}\\
 & =\sum_{v\in V(H)}\left(\Delta|_{V^{i}}(v)-(Bf^{i})_{v}\right)\phi_{v}^{i}\\
 & =\langle\phi^{i},\Delta|_{V^{i}}\rangle-\langle\phi^{i},Bf^{i}\rangle.
\end{align*}
Since the deletion set $D^{i}$ from Step \ref{Di} is designed to
guarantee that $\langle\phi^{i},\Delta|_{V^{i}}\rangle\le\langle\phi^{i},Bf^{i}\rangle$,
we have that $\langle g^{i},w^{i}\rangle\le0$.
\end{proof}
From the above, we have verified that our algorithm is indeed captured
by the MWU algorithm. Now, we derive the implication of this fact.
Only for analysis, we define the average flow $\bar{f}=\frac{1}{T}\sum_{i=1}^{T}f^{i}\in\mathbb{R}^{E(H)}$
on $H$ and the average $U$-boundary demand function $\overline{\Delta}=\frac{1}{T}\sum_{i=1}^{T}\Delta|_{V^{i}}\in\mathbb{R}^{V(H)}$
on $H$. 
\begin{lem}
\label{lem:fbar sat Deltabar}We have $\bar{f}$ $\epsilon$-satisfies
$\overline{\Delta}$ in $H$.
\end{lem}

\begin{proof}
Define $\bar{r}=\frac{1}{T}\sum_{i=1}^{T}r^{i}\in\mathbb{R}^{{\cal S}}$.
First, we prove that $|\bar{r}_{S}|\le\alpha$ for all $S\in{\cal S}$.
This is because 
\[
\pm\bar{r}_{S}=\frac{1}{T}\sum_{i\in[T]}\pm r_{S}^{i}=\frac{1}{T}\sum_{i\in[T]}g_{S,\pm}^{i}\le\alpha
\]
where the last inequality is precisely the guarantee of the MWU algorithm
from \Cref{thm:mwu}. Next observe that the excess is 
\[
\overline{\Delta}^{\bar{f}}(S)=\overline{\Delta}(S)-\bar{f}(S)=\bar{r}_{S}\delta_{H}(S).
\]
Therefore, we have that $|\overline{\Delta}^{\bar{f}}(S)|\le\alpha\delta_{H}(S)$
for all $S\in{\cal S}$. Since ${\cal S}$ is a congestion approximator,
it follows by \Cref{thm:congest} that 
\[
|\overline{\Delta}^{\bar{f}}(S)|\le\gamma_{{\cal S}}\alpha\delta_{H}(S)=\epsilon\delta_{H}(S)
\]
for all $S\subseteq V(H)$. This precisely means that $\bar{f}$ $\epsilon$-satisfies
$\overline{\Delta}$.
\end{proof}
Now, we are ready to prove \Cref{thm:almost fair:flow} of \Cref{thm:almost fair}.
By \Cref{lem:fbar sat Deltabar}, there exists a flow $\fbar_{aug}$
in $H$ with congestion $\epsilon$ such that $\fbar_{sat}:=\bar{f}+\bar{f}_{aug}$
satisfies $\Deltabar$. We define $f'_{sat}$ as the restriction of
$\fbar_{sat}$ into $G\{U'\}$. That is, for each new boundary vertex
$x_{e}\in\boundary{U'}\setminus\boundary U$ where $u$ is its unique neighbor, we set $f'_{sat}(x_{e},u)=\fbar_{sat}(e)$.
For every other edge $e\in E(G\{U'\})$, we set $f'_{sat}(e)=\bar{f}_{sat}(e)$.
Let $\Delta'$ be a $U'$-boundary demand function where, for each
$U'$-boundary vertex $v\in\boundary{U'}$, we set $\Delta'(v)=f'_{sat}(v)$
as the net flow out of $v$ via $f'_{sat}$. 
\begin{lem}
We have
\begin{enumerate}
\item $f'_{sat}$ is a flow in $G\{U'\}$ with congestion at most $(1+\epsilon)$
that satisfies $\Delta'$. 
\item $\Delta'$ is a $U'$-boundary demand function where 
\begin{align*}
\Delta'(v) & =(1-\beta)\deg_{G\{U'\}}(v) & \text{for all old boundary vertices }v\in N\langle U'\rangle\cap N\langle U\rangle\\
|\Delta'(v)| & \le(1+\epsilon)\deg_{G\{U'\}}(v) & \text{for all new boundary vertices }v\in N\langle U'\rangle\setminus N\langle U\rangle
\end{align*}
\end{enumerate}
\end{lem}

\begin{proof}
(1) As $f'_{sat}$ is a restriction of $\bar{f}_{sat}$ into $G\{U'\}$,
then the congestion of $f'_{sat}$ is at most that of $\fbar_{sat}$
which is $(1+\epsilon)$. To see why $f'_{sat}$ satisfies $\Delta'$,
we have that $\Delta'(v)=f'_{sat}(v)$ for all $U'$-boundary vertex
$v\in\boundary{U'}$ by construction. For non-boundary vertex $v\in U'\setminus t$,
we have $f'_{sat}(v)=\fbar(v)=0=\Delta'(v)$. So $f'_{sat}(v)=\Delta'(v)$
for all $v\neq t$. This implies that $f'_{sat}(t)=\Delta'(t)$ too
and so $f'_{sat}$ satisfies $\Delta'$.

(2) For each new boundary vertex $v\in N\langle U'\rangle\setminus N\langle U\rangle$,
we have $\Delta'(v)=f'_{sat}(v)$ and so $|\Delta'(v)|\le(1+\epsilon)\deg_{G\{U'\}}(v)$
because $f'_{sat}$ has congestion $(1+\epsilon)$ in $G\{U'\}$.
For each old boundary vertex $v\in N\langle U'\rangle\cap N\langle U\rangle$,
we have $\Delta'(v)=f'_{sat}(v)=\bar{f}_{sat}(v)$. As $\bar{f}_{sat}$
satisfies $\Deltabar$, we have $\bar{f}_{sat}(v)=\overline{\Delta}(v)$.
But $\overline{\Delta}(v)=(1-\beta)\deg_{G\{U\}}(v)$ as, for every
$i$, $\Delta|_{V^{i}}(v)=(1-\beta)\deg_{G\{U\}}(v)$ for every $v\notin P$.
Therefore, $\Delta'(v)=(1-\beta)\deg_{G\{U'\}}(v).$
\end{proof}

This proves \Cref{thm:almost fair:flow} of \Cref{thm:almost fair}.
It remains to prove \Cref{thm:almost fair:cut} of \Cref{thm:almost fair}. 
\begin{lem}
\label{lem:pruned cut}$\delta_{G\{U\}}(P)\le\Delta(P)$.
\end{lem}

\begin{proof}
First observe that $\delta_{H}(D^{0})\le\Delta(D^{0})$ because every
time we remove a set $S$ from $V^{0}$, we have $\delta_{H}(S)<\Delta|_{V^{0}}(S)$
and we can charge $\delta_{H}(S)$ to the decrease of $\Delta|_{V^{0}}(S)$.
Next, the sets $D^{i}$ for $i\ge1$ satisfy $\delta_{H}(D^{i})\le\Delta|_{V^{i-1}}(D^{i})$,
so 
\[
\delta_{G\{U\}}(P)=\delta_{H}(P)\le\sum_{i\ge0}\delta_{H}(D^{i})\le\Delta(D^{0})+\sum_{i\ge1}\Delta|_{V^{i-1}}(D^{i})=\Delta(P).
\]
\end{proof}
\begin{cor}
$\delta_{G}(U')\le\delta_{G}(U)-\beta\cdot\delta_{G}(P,V\setminus U)$.
\end{cor}

\begin{proof}
We have $\delta_{G\{U\}}(P)=\delta_{G}(P,U')$ and $\Delta(P)=(1-\beta)\delta_{G}(V\setminus U,P)$.
By adding $\delta_{G}(V\setminus U,U')$ into both sides of the inequality
of \Cref{lem:pruned cut}, we have 
\begin{align*}
\delta_{G}(V\setminus U,U')+\delta_{G}(P,U') & \le\delta_{G}(V\setminus U,U')+\delta_{G}(V\setminus U,P)-\beta\delta_{G}(V\setminus U,P)
\end{align*}
which concludes the proof because $\delta_{G}(U')=\delta_{G}(V\setminus U,U')+\delta_{G}(P,U')$
and $\delta_{G}(V\setminus U,U')+\delta_{G}(V\setminus U,P)=\delta_{G}(V\setminus U)=\delta_{G}(U).$
\end{proof}
This proves the correctness of \Cref{thm:almost fair}.

\subsection{Running Time}
\label{sec:time almostfair}
Here, we explain some implementation details and analyze the total
running time. Computing the congestion approximator ${\cal S}$ takes
$\Otil(m)$ by \Cref{thm:congest}. The step which ensures that no
set in ${\cal S}$ contains $t$ is at most $O(n\log n)$ time because
$t$ was contained in at most $O(\log n)$ sets $S$ and the complement
of $S$ has size at most $n$. 

Next, we explain how to implement the initialization of $V^{0}$ efficiently.
Observe that, for any $S\in{\cal S}$, if $\Delta|_{V^{0}}(S)>\delta_{H}(S)$,
then we set $V^{0}\gets V^{0}\setminus S$ and then we have  $\Delta|_{V^{0}}(S)=0$.
Otherwise, if $\Delta|_{V^{0}}(S)\le\delta_{H}(S)$,
then it remains so forever because $\Delta|_{V^{0}}(S)$ is monotonically
decreasing when $V^{0}$ is a decremental set. In any case, for each $S \in \cal S$, we only
need to compare $\Delta|_{V^{0}}(S)$ with $\delta_{H}(S)$ once,
which takes time at most $O(|S|+|E_{H}(S,V(H))|)$. So the total time
is $O(m\log n)$ because ${\cal S}$ can be partitioned into $O(\log n)$
layers of disjoint subsets by the second property of \Cref{thm:congest}.

In round $i$ of the main algorithm, computing $f^{i}$ takes $O(m)$
time. Using the fact that ${\cal S}$ is a laminar family and ${\cal S}$
contains $O(n)$ sets, we can compute $r_{S}^{i}$ for all $S\in{\cal S}$
in $O(n)$ time, and so we can compute the weights $w_{S,\circ}^{i+1}$
for all $S\in\mathcal{S}$, $\circ\in\{+,-\}$ in $O(n)$. The most
technical step is Step \ref{Di} whose implementation details is shown
at the end of the section. 
\begin{lem}
\label{lem:1} The ``deletion set'' $D^{i}\subseteq V(H)\setminus t$
from Step \ref{Di} can be computed in $O(m+n\log n)$ time. 
\end{lem}

In total, the running time is $\Otil(m)+T\cdot O(m+n\log n)$ time.
Recall that $m=|E(H)|=O(|E(G\{U\})|)$ where $T$ is the number of rounds. So we conclude the running
time analysis: 
\begin{lem}
The total running time of the algorithm for \Cref{thm:almost fair}
is at most $\Otil(|E(G\{U\})|/\epsilon^{2})$.
\end{lem}

\subsection{Proof of \Cref{lem:1}}

In this section, we show how to construction $D^{i}\subseteq V(H)\setminus t$
where 
\begin{align}
\text{if }D^{i}\neq\emptyset,\text{ then }\Delta|_{V^{i-1}}(D^{i}) & >\delta_{H}(D^{i})\label{i1}\\
\langle\phi^{i},\Delta|_{V^{i}}\rangle=\langle\phi^{i},\Delta|_{V^{i-1}\setminus D^{i}}\rangle & \le\langle\phi^{i},Bf^{i}\rangle.\label{i2}
\end{align}
If $\langle\phi^{i},\Delta|_{V^{i-1}}\rangle\le\langle\phi^{i},Bf^{i}\rangle$,
then we simply set $D^{i}=\emptyset$, which trivially fulfills both
conditions. For the remainder of the proof, we assume that $\langle\phi^{i},\Delta|_{V^{i-1}}\rangle>\langle\phi^{i},Bf^{i}\rangle$. 

For real number $x$, define $V_{>x}=\{v\in V(H):\phi_{v}^{i}>x\}$.
Fix some large number $M>{\displaystyle \max_{v\in N\{U\}}|\phi_{v}^{i}|}$.
We first prove the chain of relations
\begin{equation}
\int_{x=-M}^{M}\Delta|_{V^{i-1}}(V_{>x})dx=\langle\phi^{i},\Delta|_{V^{i-1}}\rangle>\langle\phi^{i},Bf^{i}\rangle=\int_{x=-M}^{M}\delta_{H}(V_{>x})dx.\label{eq:chain}
\end{equation}
We start with 
\begin{align*}
\int_{x=-M}^{M}\Delta|_{V^{i-1}}(V_{>x})dx & =\int_{x=-M}^{M}\left(\sum_{v\in V(H)}\Delta|_{V^{i-1}}(v)\cdot\one\{\phi_{v}^{i}>x\}\right)dx\\
 & =\sum_{v\in V(H)}\Delta|_{V^{i-1}}(v)\int_{x=-M}^{M}\one\{\phi_{v}^{i}>x\}dx\\
 & =\sum_{v\in V(H)}\Delta|_{V^{i-1}}(v)(\phi_{v}^{i}-(-M)).
\end{align*}
Since $\sum_{v\in V(H)}\Delta|_{V^{i-1}}(v)=0$ by construction, this
is equal to 
\[
\sum_{v\in V(H)}\Delta|_{V^{i-1}}(v)\,\phi_{v}^{i}=\langle\phi^{i},\Delta|_{V^{i-1}}\rangle.
\]
By definition of the flow $f^{i}$, 
\begin{align*}
\langle\phi^{i},Bf^{i}\rangle & =\sum_{(u,v)\in E(H)}c_{H}(u,v)|\phi_{u}^{i}-\phi_{v}^{i}|\\
 & =\sum_{(u,v)\in E(H)}c_{H}(u,v)\int_{x=-M}^{M}\one\{(u,v)\in\partial_{H}(V_{>x})\}dx\\
 & =\int_{x=-M}^{M}\sum_{(u,v)\in E(H)}c_{H}(u,v)\one\{(u,v)\in\partial_{H}(V_{>x})\}dx\\
 & =\int_{x=-M}^{M}\delta_{H}(V_{>x})dx.
\end{align*}
Together with the assumption $\langle\phi^{i},\Delta|_{V^{i-1}}\rangle>\langle\phi^{i},Bf^{i}\rangle$,
we obtain (\ref{eq:chain}).

Let $x^{*}$ be the largest value such that 
\[
\int_{x=-M}^{x^{*}}\Delta|_{V^{i-1}}(V_{>x})dx=\int_{x=-M}^{x^{*}}\delta_{H}(V_{>x})dx,
\]
which must exist since $x^{*}=-M$ works. Next, we claim that we must
have
\begin{equation}
\Delta|_{V^{i-1}}(V_{>x^{*}})>\delta_{H}(V_{>x^{*}}).\label{eq:1}
\end{equation}
Otherwise, for small enough $\epsilon>0$ we would have $\int_{x=-M}^{x^{*}+\epsilon}\Delta|_{V^{i-1}}(V_{>x})dx\le\int_{x=-M}^{x^{*}+\epsilon}\delta_{H}(V_{>x})dx$,
and since $\int_{x=-M}^{M}\Delta|_{V^{i-1}}(V_{>x})dx>\int_{x=-M}^{M}\delta_{H}(V_{>x})dx$,
there is another choice of $x^{*}$ between $x^{*}+\epsilon$ and
$M$ that achieves equality, a contradiction.

We now claim that $t\notin V_{>x^{*}}$. Otherwise, since $\Delta|_{V^{i-1}}(V(H))=0$
and $\Delta|_{V^{i-1}}(t)$ is the only negative entry, we would have
$\Delta|_{V^{i-1}}(V_{>x^{*}})\le0$ which would violate (\ref{eq:1}).
Since $t\notin V_{>x^{*}}$ and $\phi_{t}^{i}=0$, we conclude that
$x^{*}\ge0$.

Let $\overline{\phi}^{i}=\min\{\phi^{i},x^{*}\}$ as $\phi^{i}$ truncated
to a maximum of $x^{*}$. Then, similar to (\ref{eq:chain}), we obtain
\begin{equation}
\langle\overline{\phi}^{i},\Delta|_{V^{i-1}}\rangle=\int_{x=-M}^{x^{*}}\Delta|_{V^{i-1}}(V_{>x})dx=\int_{x=-M}^{x^{*}}\delta_{H}(V_{>x})dx=\langle\overline{\phi}^{i},Bf^{i}\rangle.\label{eq:chain2}
\end{equation}
Define our deletion set as $D^{i}\triangleq V_{>x^{*}}$, so $t\notin D^{i}$
and \Cref{i1} follows from (\ref{eq:1}). We now prove the chain of
relations 
\[
\langle\phi^{i},\Delta|_{V^{i-1}\setminus D^{i}}\rangle=\langle\overline{\phi}^{i},\Delta|_{V^{i-1}\setminus D^{i}}\rangle\le\langle\overline{\phi}^{i},\Delta|_{V^{i-1}}\rangle=\langle\overline{\phi}^{i},Bf^{i}\rangle\le\langle\phi^{i},Bf^{i}\rangle,
\]
which would fulfill \Cref{i2}. For the first relation, if $\phi_{v}^{i}\ne\overline{\phi}_{v}^{i}$
then $v\in D^{i}$, which means that $\Delta|_{V^{i-1}\setminus D^{i}}(v)=0$.
For the second relation, we use $\overline{\phi}_{t}^{i}=\phi_{t}^{i}=0$
to obtain 
\begin{align*}
\langle\overline{\phi}^{i},\Delta|_{V^{i-1}\setminus D^{i}}\rangle=\sum_{v\in V(H)\setminus t}\overline{\phi}^{i}(v)\Delta|_{V^{i-1}\setminus D^{i}}(v)= & \sum_{v\in V(H)\setminus t}\overline{\phi}^{i}(v)\Delta|_{V^{i-1}}(v)-x^{*}\Delta|_{V^{i-1}}(D^{i})\\
 & =\langle\overline{\phi}^{i},\Delta|_{V^{i-1}}\rangle-x^{*}\Delta|_{V^{i-1}}(D^{i})
\end{align*}
which is at most $\langle\overline{\phi}^{i},\Delta|_{V^{i-1}}\rangle$
since $x^{*}\ge0$. The third relation follows from (\ref{eq:chain2}).
For the last relation, we have 
\[
\langle\overline{\phi}^{i},Bf^{i}\rangle=\sum_{(u,v)\in E(H)}c_{H}(u,v)|\overline{\phi}_{u}^{i}-\overline{\phi}_{v}^{i}|\le\sum_{(u,v)\in E(H)}c_{H}(u,v)|\phi_{u}^{i}-\phi_{v}^{i}|=\langle\phi^{i},Bf^{i}\rangle.
\]
This concludes \Cref{i2}.

Finally, we claim the running time $O(m+n\log n)$. The only nontrivial
step in the algorithm is computing $x^{*}$. We first sort the values
$\phi_{v}^{i}$ in $O(n\log n)$ time. Then, by sweeping through the
sorted list, we can compute $\Delta|_{V^{i-1}}(V_{>x})-\delta_{H}(V_{>x})$
for all $x\in\{\phi_{v}^{i}:v\in V(H)\}$ in $O(m)$ time. The function
$\Delta|_{V^{i-1}}(V_{>x})-\delta_{H}(V_{>x})$ is linear between
consecutive values of $\phi_{v}^{i}$, so we can locate the largest
value $x^{*}$ for which the function is $0$.\thatchaphol{Need to check if how to make this parallel.}

\section{From Almost Fair Cuts to Fair Cuts}
\label{sec:from_almost_to_fair}
In this section, we prove \Cref{thm:fair} using the $\almostfair$
subroutine. 

\subsection{Algorithm}
\label{sec:alg fair}
Let $(G,s,t,\alpha)$ be the input and we want to compute a $(1+\alpha)$-fair
$(s,t)$-cut in $G$. Let $c_{\min}=\min_{e}c(e)$ and let $C=c(E)/c_{\min}$
be the ratio between total capacity and the minimum capacity. Recall
that we assume $C=\poly(n)$. We also assume $\alpha\ge\frac{1}{\poly(n)}$,
otherwise we could solve the problem using exact max flow algorithms.

Our algorithm runs in iterations where in iteration $j$ we compute
$(S^{j},T^{j},k^{j},\exx^{j})$ where $(S^{j},T^{j})$ is an $(s,t)$-cut
where $s\in S^{j}$ and $t\in T^{j}$, $k^{j}\in\mathbb{Z}_{\ge0}$,
and $\exx^{j}\in\mathbb{R}_{\ge0}$ represents an upper bound of the
\emph{deficit} which will be explained in the analysis. Define $\beta=\Theta(\alpha/\log n)$
and $\epsilon=\beta/16$. Initially, $(S^{0},T^{0})$ is an arbitrary
$(s,t)$-cut,  $\exx^{0}=\delta_{G}(S^{0},T^{0})$, and  $k^0=0$.

While $\exx^{j}>\beta c_{\min}$, do the following starting from $j=0,1,2,\dots$
\begin{enumerate}
\item Compute 
\begin{align*}
(P_{s}^{j},S^{j}\setminus P_{s}^{j}) & =\almostfair(G,S^{j},s,\epsilon,(k^{j}+1)\beta),\text{ and}\\
(P_{t}^{j},S^{j}\setminus P_{t}^{j}) & =\almostfair(G,T^{j},t,\epsilon,(k^{j}+1)\beta)
\end{align*}
by calling \Cref{thm:almost fair}.
\item \label{TWO-SIDED:item:Case 1 new flow} If $\max\{\delta_{G}(P_{s}^{j},T^{j}),\delta_{G}(P_{t}^{j},S^{j})\}\le\exx^{j}/40$,
then we update 
\begin{align*}
k^{j+1} & =k^{j}+1,\text{ and}\\
\exx^{j+1} & =\exx^{j}/2.
\end{align*}
Then, we set $T^{j+1}=T^{j}\setminus P_{t}^{j}$ and $S^{j+1}=V\setminus T^{j+1}$.\footnote{We could also symmetrically set $S^{j+1}=S^{j}\setminus P_{s}^{j}$
and $T^{j+1}=V\setminus S^{j+1}$. This choice is arbitrary.}
\item \label{TWO-SIDED:item:Case 2 old flow} Else, if $\max\{\delta_{G}(P_{s}^{j},T^{j}),\delta_{G}(P_{t}^{j},S^{j})\}>\exx^{j}/40$,
then we update 
\begin{align*}
k^{j+1} & =k^{j},\text{ and}\\
\exx^{j+1} & =(1-\beta/80)\exx^{j}
\end{align*}
 If $\delta_{G}(P_{s}^{j},T^{j})>\exx^{j}/40$, then, we set $S^{j+1}=S^{j}\setminus P_{s}^{j}$
(and $T^{j+1}=V\setminus S^{j+1}$). Otherwise, we set $T^{j+1}=T^{j}\setminus P_{t}^{j}$
(and $S^{j+1}=V\setminus T^{j+1}$). 
\end{enumerate}
After the while loop, we return $(S^{j},T^{j})$ as a $(1+\alpha)$-fair
$(s,t)$-cut. As $\exx^{0}\le c(E)$ we have that $\exx^{j}\le(1-\beta/80)^{j}c(E)$
for all $j$. So there are at most $O(\log(C/\beta)/\beta)$ iterations
before $\exx^{j}<\beta\cdot c_{\min}$. Therefore, the algorithm takes
$O(\log(C/\beta)/\beta)\times\Otil(m/\epsilon^{2})=\Otil(m/\alpha^{3})$
total time by \Cref{thm:almost fair}. It remains to show the correctness
of the algorithm.

\subsection{Analysis}

For convenience, whenever we refer to an edge $(a,b)\in E(A,B)$,
we mean $a\in A$ and $b\in B$. Only for the analysis, we construct
a feasible flow $f^{j}$ in $G$ on each iteration $j$, and ensure
that $f^{j}$ satisfies the following two properties: 
\begin{enumerate}
\item Define the \emph{deficit} of flow $f^{j}$ as $\ex^{j}(f^{j})=\sum_{(u,v)\in E(S^{j},T^{j})}\max\{0,(1-k^{j}\beta)c(u,v)-f^{j}(u,v)\}$.
We maintain an invariant that $\ex^{j}(f^{j})\le\exx^{j}$.\label{TWO-SIDED:eq:1} 
\item For all $R\subseteq V\setminus\{s,t\}$, we require that $|f^{j}(R)|\le\epsilon\delta_{G}(R)$.
Equivalently, $f^{j}$ $\epsilon$-satisfies an $(s,t)$-demand function
in $G$.\label{TWO-SIDED:eq:2} 
\end{enumerate}
In words, each cut edge $(u,v)\in E(S^{j},T^{j})$ contributes to
the deficit of flow $f^{j}$ when the flow in $f^{j}$ \emph{from
$u$ to $v$} is less than $(1-k^{j}\beta)$-fraction of its capacity.
With our definition of deficit in Property \ref{TWO-SIDED:eq:1},
we have that the cut is fair whenever the deficit is very small:
\begin{prop}
If $\exx^{j}<\beta c_{\min}$, then $(S^{j},T^{j})$ is a $(1+\alpha)$-fair
$(s,t)$-cut.
\end{prop}

\begin{proof}
First we claim that $k^{j}=O(\log n)$. This is because everytime
$k^{j}$ increments, $\exx$ is halved. So at the end of the algorithm,
we have $\frac{\beta c_{\min}}{2}<\exx^{j}<c(E)/ 2^{k^{j}}$, which
implies $k^{j}=O(\log(C/\beta))=O(\log n)$. Now, by the assumption
and Property \ref{TWO-SIDED:eq:1}, for all $(u,v)\in E(S^{j},T^{j})$,
we have $(1-k^{j}\beta)c(u,v)-f^{j}(u,v)<\beta\cdot c_{\min}$ and
so 
\[
f^{j}(u,v) > (1-(k^{j}+1)\beta)c(u,v)\ge\frac{1}{(1+\alpha/2)}c(u,v)
\]
where the last inequality is because $k^{j}=O(\log n)$ and we can
set the constant in $\beta=\Theta(\alpha/\log n)$ to be small enough.
Since $f^j$  $\epsilon$-satisfies an $(s,t)$-demand function, by the observation below \Cref{fact:mfmc}, there exists $f_{aug}$ with congestion $\epsilon$ such that $f^* = f^j + f_{aug}$ is an $(s,t)$-flow.
Now, we have that for all $(u,v)\in E(S^{j},T^{j})$,
\[
f^*(u,v) \ge f^{j}(u,v) -\epsilon c(u,v) \ge \frac{1}{(1+\alpha)} c(u,v)
\] because $\epsilon = \beta/16 = \Theta(\alpha/\log n)$ and the constant in it is small enough.
Therefore, $f^*$ certifies that $(S^{j},T^{j})$ is a $(1+\alpha)$-fair
$(s,t)$-cut.
\end{proof}
Initially, we set $f^{0}$ as the zero flow, which satisfies both
properties since $\exx^{0}=\delta_{G}(S^{0},T^{0})$. Property \ref{TWO-SIDED:eq:2}
will help us show the following inductive step, which would conclude
the correctness of \Cref{thm:fair}.
\begin{lem}
\label{TWO-SIDED:thm:excess decreases} Suppose there exists a feasible flow
$f^{j}$ satisfying Properties~\ref{TWO-SIDED:eq:1} and~\ref{TWO-SIDED:eq:2}
for $j$. Then, we can construct a feasible flow $f^{j+1}$ satisfying Properties~\ref{TWO-SIDED:eq:1}
and~\ref{TWO-SIDED:eq:2} for $j+1$. 
\end{lem}

We analyze the two cases based on $\max\{\delta_{G}(P_{s}^{j},T^{j}),\delta_{G}(P_{t}^{j},S^{j})\}$
in the subsections below. 

\subsubsection*{Case 1: $\max\{\delta_{G}(P_{s}^{j},T^{j}),\delta_{G}(P_{t}^{j},S^{j})\}\leq\protect\exx^{j}/40$}

Let $S'^{j}=S^{j}\setminus P_{s}^{j}$. By the guarantees of $\almostfair(G,S^{j},s,\epsilon,(k^{j}+1)\beta)$,
let $\Delta_{s}$ be the $S'^{j}$-boundary demand function satisfied
by a flow $f_{s}$ in $G\{S'^{j}\}$ with congestion $(1+\epsilon)$.
As $k^{j+1}=k^{j}+1$ in this case, by \Cref{thm:almost fair}, we
have $f_{s}(v)=\Delta_{s}(v)=(1-k^{j+1}\beta)\deg_{G\{S^{j}\}}(v)$
for all old boundary vertices $v\in N\langle S^{j}\rangle\cap N\langle S'^{j}\rangle.$
Let $T'^{j},\Delta_{t},f_{t}$ be defined symmetrically. From $f_{s}$
and $f_{t}$, we will construct a new flow $f^{j+1}$ in three steps.

\paragraph{Step 1: Concatenate. Get $\hat{f}$.}

Consider the ``concatenation'' of $f_{s}$ and $f_{t}$, denoted by $f_{st}$, where
we reverse the direction of $f_{s}$ so that the flow is sent out
of $s$. The concatenated flow $f_{st}$ is on the graph $G\{S'^{j}\}\cup G\{T'^{j}\}$
where the two graphs share $N\langle S'^{j}\rangle\cap N\langle T'^{j}\rangle$
as common boundary vertices. Now, we want to define a flow $\hat{f}$
on $G$ that corresponds to $f_{st}$ in a natural way. See \Cref{fig:define flow}. 
\begin{enumerate}
    \item For each edge $e\in E(G[S'^{j}])\cup E(G[T'^{j}])$ in the ``interior'' of $S'^{j}$ or $T'^{j}$, we set $\hat{f}(e)=f_{st}(e)$.
    \item For each common boundary vertex $x_{e}\in N\langle S'^{j}\rangle\cap N\langle T'^{j}\rangle$ where $e=(u,v)\in E(S'^{j},T'^{j})$, we have $f_{st}(u,x_{e})=f_{st}(x_{e},v)=(1-k^{j+1}\beta)c(e)$ and so we set $\hat{f}(e)=(1-k^{j+1}\beta)c(e)$.
    \item  For each new boundary vertex $x_{e}\in(N\langle S'^{j}\rangle\setminus N\langle S^{j}\rangle)\cup(N\langle T'^{j}\rangle\setminus N\langle T^{j}\rangle)$ where $e=(u,v)\in E(S'^{j},P_{s}^{j})\cup E(T'^{j},P_{t}^{j})$, we set $\hat{f}(e)=f_{st}(u,x_{e})$.
    \item For each old boundary vertex $x_{e}\in N\langle S{}^{j}\rangle\cap N\langle T{}^{j}\rangle$ incident to the pruned set $P_{s}^{j}$ or $P_{t}^{j}$ on one side, i.e., $e=(u,v)\in E(S'^{j},P_{t}^{j})\cup E(T'^{j},P_{s}^{j})$, we set $\hat{f}(e)=f_{st}(u,x_{e})$. 
    \item For each old boundary vertex $x_{e}\in N\langle S{}^{j}\rangle\cap N\langle T{}^{j}\rangle$ incident to the pruned set $P_{s}^{j}$ or $P_{t}^{j}$ on both sides, i.e., $e=(u,v)\in E(P_{s}^{j},P_{t}^{j})$, we set $\hat{f}(e)=0$. 
    \item For each edge in the ``interior'' of $P_{s}^{j}$ or $P_{t}^{j}$, we set $\hat{f}(e)=0$. 
\end{enumerate}
By construction, $\hat{f}$ satisfies some demand function $\hat{\Delta}$ where $\hat{\Delta}(v)=0$ for $v\notin\{s,t\}\cup V(P_{s}^{j})\cup V(P_{t}^{j})$.

\begin{figure}[h]
\begin{centering}
\includegraphics[width=0.4\textwidth]{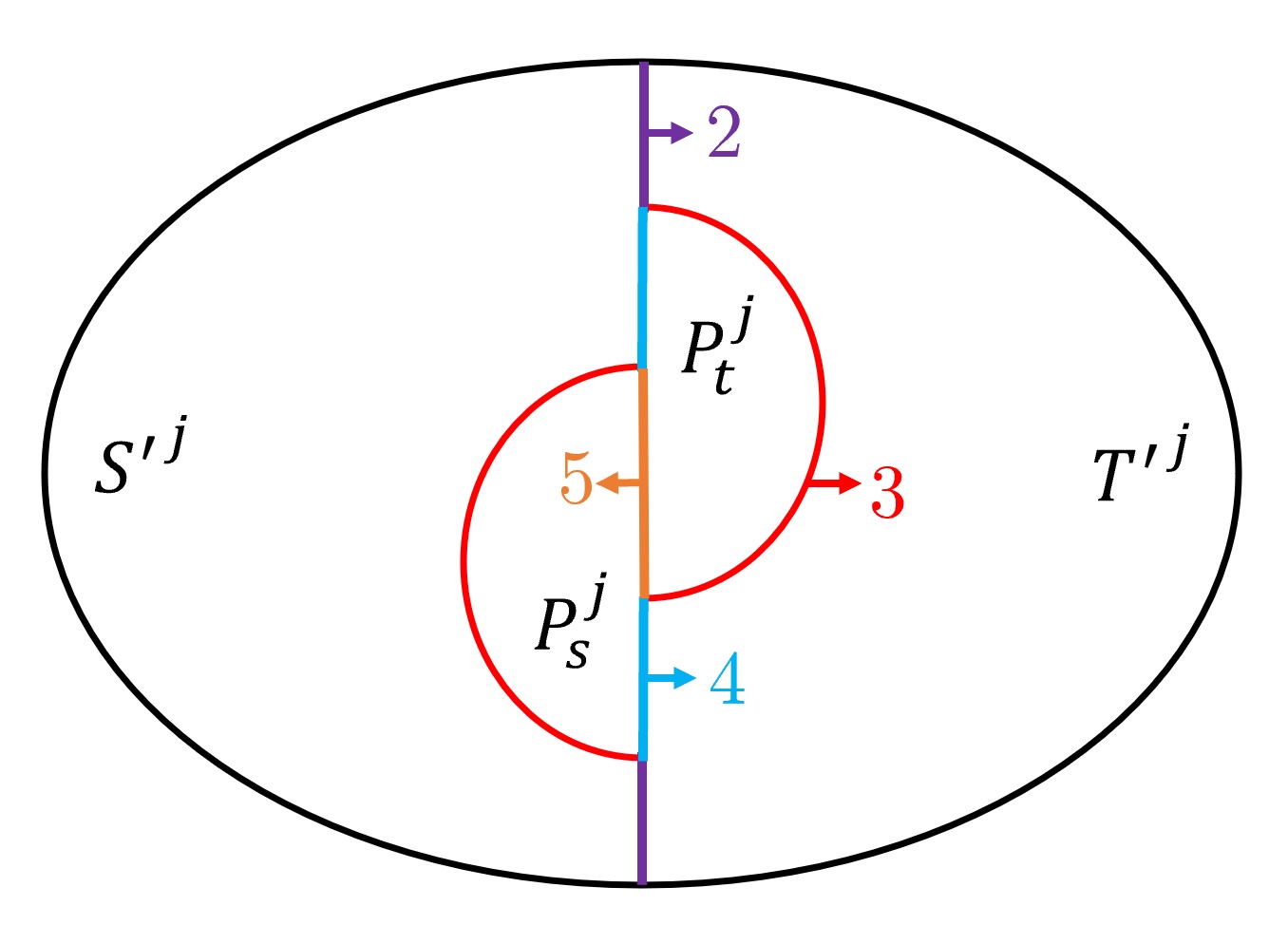}
\par\end{centering}
\caption{A diagram indicating the cases for defining $\hat{f}$ from $f_{st}$.\label{fig:define flow}}
\end{figure}

\paragraph{Step 2: Remove Flow Paths Through New Boundaries. Get $\hat{f}'$.}

Take a path decomposition of $\hat{f}$ in $G$, and then remove all
paths starting or ending at vertices in $V(P_{s}^{j})\cup V(P_{t}^{j})$;
let the resulting flow be $\hat{f}'$, which satisfies some demand
function that is only nonzero at $s,t$. That is, $\hat{f}'$ is an
$(s,t)$-flow. Note that $\hat{f}'$ still has congestion at most $(1+\epsilon)$. 

\paragraph{Step 3: Truncate to a Feasible Flow. Get $f^{j+1}$.}

Finally, for any edges congested by more than $1$ in $\hat{f}'$,
lower the flow along that edge to congestion exactly $1$. We define
$f^{j+1}$ as the resulting flow. 

\paragraph{Proving Properties of $f^{j+1}$.}

Since $f^{j+1}$ is obtained from the $(s,t)$-flow $\hat{f}'$ by removing a flow of
congestion at most $\epsilon$, Property~\ref{TWO-SIDED:eq:2} is
satisfied. Now, we prove Property~\ref{TWO-SIDED:eq:1}. We write
the deficit of $f^{j+1}$ as follows
\begin{align*}
 & \ex^{j+1}(f^{j+1})\\
 & =\sum_{e\in E(S^{j+1},T^{j+1})}\max\{0,(1-k^{j+1}\beta)c(e)-f^{j+1}(e)\}\\
 & \le\sum_{e\in E(S^{j+1},T^{j+1})}\left(\max\{0,(1-k^{j+1}\beta)c(e)-\hat{f}(e)\}+|\hat{f}(e)-\hat{f}'(e)|+|\hat{f}'(e)-f^{j+1}(e)|\right)\\
 & =\sum_{e\in E(S^{j+1},T^{j+1})}\max\{0,(1-k^{j+1}\beta)c(e)-\hat{f}(e)\}+\\&\sum_{e\in E(S^{j+1},T^{j+1})}|\hat{f}(e)-\hat{f}'(e)|+\sum_{e\in E(S^{j+1},T^{j+1})}|\hat{f}'(e)-f^{j+1}(e)|
\end{align*}
Now, we bound each of the three terms above. We use the fact $T^{j+1}=T'^{j}$
and $S^{j+1}=S'^{j}\cup P_{s}^{j}\cup P_{t}^{j}$. 

For the first term, we consider the concatenated flow $\hat{f}$.
We have $\hat{f}(e)=(1-k^{j+1}\beta)c(e)$ for each old boundary edge
$e\in E(S'^{j},T'^{j})$. So, the first term is bounded by 
\begin{align*}
\sum_{e\in E(S'^{j}\cup P_{s}^{j}\cup P_{t}^{j},T'^{j})}\max\{0,(1-k^{j+1}\beta)c(e)-\hat{f}(e)\} & \le\sum_{e\in E(P_{s}^{j}\cup P_{t}^{j},T'^{j})}(1-k^{j+1}\beta)c(e)-\hat{f}(e)\\
 & \le((1-k^{j+1}\beta)+(1+\epsilon))\cdot\delta(P_{s}^{j}\cup P_{t}^{j},T'^{j})\\
 & \le(2+\epsilon)\cdot\delta(P_{s}^{j}\cup P_{t}^{j},T'^{j})
\end{align*}
where the second inequality is because $\hat{f}$ has $(1+\epsilon)$
congestion.

For the second term, consider the flow $\hat{f}'$ obtained by the
flow-path removal. We rewrite the second term as  
\[
\sum_{e\in E(P_{s}^{j}\cup P_{t}^{j},T'^{j})}|\hat{f}(e)-\hat{f}'(e)|+\sum_{e\in E(S'^{j},T'^{j})}|\hat{f}(e)-\hat{f}'(e)|.
\]
Trivially, we have 
\[
\sum_{e\in E(P_{s}^{j}\cup P_{t}^{j},T'^{j})}|\hat{f}(e)-\hat{f}'(e)|\le(1+\epsilon)\delta(P_{s}^{j}\cup P_{t}^{j},T'^{j})
\]
because the flow has congestion $(1+\epsilon)$. Now, we claim that
\[
\sum_{e\in E(S'^{j},T'^{j})}|\hat{f}(e)-\hat{f}'(e)|\le\sum_{e\in E(P_{s}^{j}\cup P_{t}^{j},S'^{j}\cup T'^{j})}|\hat{f}(e)-\hat{f}'(e)|\le(1+\epsilon)\delta(P_{s}^{j}\cup P_{t}^{j},S'^{j}\cup T'^{j}).
\]
To see this, consider each flow-path $P$ removed from $\hat{f}$
to obtain $\hat{f}'$. Observe that $P$ cannot cross directly from
$T'^{j}$ to $S'^{j}$ because, for \emph{every} edge $e\in E(S'^{j},T'^{j})$,
the flow is directed from $S'^{j}$ to $T'^{j}$ as $\hat{f}(e)=(1-k^{j+1}\beta)c(e)$.
Thus, between any two consecutive times that $P$ crosses from $S'^{j}$
to $T'^{j}$, $P$ must have crossed from $T'^{j}$ to $P_{s}^{j}\cup P_{t}^{j}$.
Also, note that the first edge of $P$ is from $E(P_{s}^{j}\cup P_{t}^{j},S'^{j}\cup T'^{j})$.
Therefore, we can charge the flow changes in edges of $E(S'^{j},T'^{j})$
to the changes in edges of $E(P_{s}^{j}\cup P_{t}^{j},S'^{j}\cup T'^{j})$.
So $\sum_{e\in E(S'^{j},T'^{j})}|\hat{f}(e)-\hat{f}'(e)|\le\sum_{e\in E(P_{s}^{j}\cup P_{t}^{j},S'^{j}\cup T'^{j})}|\hat{f}(e)-\hat{f}'(e)|$
as claimed. 

Finally, for the third term, we consider the truncated flow $f^{j+1}$
with congestion at most $1$ on all edges. Again, we have $\hat{f}'(e)-f^{j+1}(e)=0$
for all $e\in E(S'^{j},T'^{j})$ because $0\le\hat{f}'(e)\le(1-k^{j+1}\beta)c(e)$.
In particular, the congestion on $e$ was already less than $1$. Also,
we have $|\hat{f}'(e)-f^{j+1}(e)|\le\epsilon c(e)$ for any edges
$e$ as $\hat{f}'$ has congestion $1+\epsilon$. Hence, we have
\[
\sum_{e\in E(S^{j+1},T^{j+1})}|\hat{f}'(e)-f^{j+1}(e)|\le\sum_{e\in E(P_{s}^{j}\cup P_{t}^{j},T'^{j})}\epsilon c(e)=\epsilon\delta(P_{s}^{j}\cup P_{t}^{j},T'^{j}).
\]

From the above bounds, we obtain
\begin{align*}
\ex^{j+1}(f^{j+1}) & \le((2+\epsilon)+(1+\epsilon)+(1+\epsilon)+\epsilon)\delta(P_{s}^{j}\cup P_{t}^{j},S'^{j}\cup T'^{j}).
\end{align*}
Now, write $\delta(P_{s}^{j}\cup P_{t}^{j},S'^{j}\cup T'^{j})=\delta(P_{s}^{j},S'^{j})+\delta(P_{t}^{j},S'^{j})+\delta(P_{s}^{j},T'^{j})+\delta(P_{t}^{j},T'^{j})$.
Note that $\delta(P_{t}^{j},T'^{j})\le\delta(P_{t}^{j},S^{j})$ and
$\delta(P_{s}^{j},S'^{j})\le\delta(P_{s}^{j},T^{j})$ by the guarantee
of $\almostfair$. Trivially, we also have $\delta(P_{t}^{j},S'^{j})\le\delta(P_{t}^{j},S^{j})$
and $\delta(P_{s}^{j},T'^{j})\le\delta(P_{s}^{j},T{}^{j})$. But we
have $\delta(P_{t}^{j},S{}^{j}),\delta(P_{s}^{j},T{}^{j})\le\exx^{j}/40$
by the assumption of this case. So we have, as $\epsilon\le1/4$,
\[
\ex^{j+1}(f^{j+1})\le(4+4\epsilon)\cdot4\cdot\frac{\exx^{j}}{40}\le\exx^{j}/2=\exx^{j+1}
\]
fulfilling Property~\ref{TWO-SIDED:eq:1}.

\subsubsection*{Case 2: $\max\{\delta_{G}(P_{s}^{j},T^{j}),\delta_{G}(P_{t}^{j},S^{j})\}>\protect\exx^{j}/40$}

In this case, we set $f^{j+1}$ as the same old flow $f^{j}$. So
Property~\ref{TWO-SIDED:eq:2} of $f^{j+1}$ trivially continues to
hold. For Property \ref{TWO-SIDED:eq:1}, assume without loss of generality
the case $\delta_{G}(P_{t}^{j},S^{j})>\ex^{j}/40$, so $T^{j+1}=T^{j}\setminus P_{t}^{j}$.
(The case $\delta_{G}(P_{s}^{j},T^{j})>\ex^{j}/40$ is symmetric,
so we omit it.) As $f^{j+1}=f^{j}$ and $k^{j+1}=k^{j}$, we have
\begin{align*}
 & \ex^{j+1}(f^{j+1})\\
 & =\sum_{e\in E(S^{j+1},T^{j+1})}\max\{0,(1-k^{j}\beta)c(e)-f^{j}(e)\}\\
 & =\ex^{j}(f^{j})-\sum_{e\in E(S^{j},P_{t}^{j})}\max\{0,(1-k^{j}\beta)c(e)-f^{j}(e)\}+\sum_{e\in E(P_{t}^{j},T^{j+1})}\max\{0,(1-k^{j}\beta)c(e)-f^{j}(e)\}.
\end{align*}
For the second term (without the minus sign), we can lower bound it
as

\[
\ge\sum_{e\in E(S^{j},P_{t}^{j})}(1-k^{j}\beta)c(e)-f^{j}(e)=(1-k^{j}\beta)\delta(S^{j},P_{t}^{j})-f^{j}(S^{j},P_{t}^{j}).
\]
For the third term, we can upper bound it as 
\[
\le\sum_{e\in E(P_{t}^{j},T^{j+1})}c(e)-f^{j}(e)=\delta(P_{t}^{j},T^{j+1})-f^{j}(P_{t}^{j},T^{j+1}).
\]
where the first inequality is because $0\le c(e)-f^{j}(e)$ as $f^{j}$
is feasible. Putting these together, we have 
\[
\ex^{j+1}(f^{j+1})\le\ex^{j}(f^{j})-\left((1-k^{j}\beta)\delta(S^{j},P_{t}^{j})-\delta(P_{t}^{j},T^{j+1})\right)+\left(f^{j}(S^{j},P_{t}^{j})-f^{j}(P_{t}^{j},T^{j+1})\right).
\]
That is, the increase in deficit can be upper bounded as follows.
It will decrease proportional to $(1-k^{j}\beta)\delta(S^{j},P_{t}^{j})-\delta(P_{t}^{j},T^{j+1})$ which is
related cut size. It may increase proportional to $f(S^{j},P_{t}^{j})-f(P_{t}^{j},T^{j+1})$
which is related to flow. 

For the decrease caused by cut size, $\almostfair(G,T^{j},t,\epsilon,(k^{j}+1)\beta)$
guarantees that $\delta(P_{t}^{j},T^{j+1})\le(1-(k^{j}+1)\beta)\delta(S^{j},P_{t}^{j})$.
So the deficit must decrease by at least $\left((1-k^{j}\beta)-(1-(k^{j}+1)\beta)\right)\delta(S^{j},P_{t}^{j})\ge\beta\delta(S^{j},P_{t}^{j}).$
For the increase caused by flow, we have that $f^{j}(S^{j},P_{t}^{j})-f^{j}(P_{t}^{j},T^{j+1})=f^{j}(S^{j},P_{t}^{j})+f^{j}(T^{j+1},P_{t}^{j})=-f^{j}(P_{t}^{j})$
is exactly the net flow of $f^{j}$ into $P_{t}^{j}$. As $|f^{j}(P_{t}^{j})|\le\epsilon\delta_{G}(P_{t}^{j})$
by Property~\ref{TWO-SIDED:eq:2} on $P_{t}^{j}$, we now have 
\begin{align*}
\ex^{j+1}(f^{j+1}) & \le\ex^{j}(f^{j})-\beta\delta(S^{j},P_{t}^{j})+\epsilon\delta_{G}(P_{t}^{j}).
\end{align*}
Observe that $\delta_{G}(P_{t}^{j})=\delta_{G}(S^{j},P_{t}^{j})+\delta_{G}(P_{t}^{j},T^{j+1})$
but $\delta(P_{t}^{j},T^{j+1})\le\delta(S^{j},P_{t}^{j})$ by $\almostfair$
again. So $\epsilon\delta_{G}(P_{t}^{j})\le2\epsilon\delta_{G}(S^{j},P_{t}^{j})\le\frac{\beta}{2}\delta_{G}(S^{j},P_{t}^{j})$
because $\epsilon\le\beta/4$. Therefore, 
\[
\ex^{j+1}(f^{j+1})\le\ex^{j}(f^{j})-\frac{\beta}{2}\delta(S^{j},P_{t}^{j})\le(1-\frac{\beta}{80})\ex^{j}(f^{j})=\exx^{j+1}
\]
because $\delta_{G}(S^{j},P_{t}^{j})>\ex^{j}/40$ by our initial assumption.

\section{Approximate Isolating Cuts and Steiner Cut}
\label{sec:isolating}
The focus of this section is to compute approximate isolating cuts and show its application in the Steiner mincut problem. 

\subsection{Approximate Minimum Isolating Cuts}

The approximate minimum isolating cuts problem is defined below.

\begin{definition}
    Given an undirected graph $G = (V, E)$ with non-negative edge weights and a set of terminals $T\subseteq V$, a cut $\emptyset \subset S \subset V$ is said to be an {\em isolating cut} for a terminal $t\in T$ if $T\cap S = \{t\}$. 
    A {\em minimum} isolating cut for $t$ is a minimum value cut among all the isolating cuts for $t$. Similarly, a $(1+\e)$-approximate minimum isolating cut for $t$ is an isolating cut for $t$ whose value is at most $(1+\e)$ times that of a minimum isolating cut for $t$.
\end{definition}

Below is our main theorem. We state our result in general before plugging in the current best runtime from \Cref{thm:fair}.

\begin{theorem}
\label{thm:isolating}
    We can compute $(1+\epsilon)$ approximate minimum isolating cuts in $\tilde O(m)$ time. 
    
    More precisely, fix any $\e < 1$. Given an undirected graph $G = (V, E)$ on $m$ edges and $n$ vertices with non-negative edge weights and a set of terminals $T\subseteq V$, there is an algorithm that outputs a $(1+\e)$-approximate minimum isolating cut $S_t$ for every terminal $t\in T$ in $O(m)$ time plus a set of $(1+\gamma)$-fair $(s,t)$-cut calls on undirected graphs that collectively  contain $O(m \log |T|)$ edges and $O(n\log |T|)$ vertices, where $\gamma =  \frac{\e}{4\lceil\lg |T|\rceil}$.
    Moreover, the sets $S_t$ are disjoint, and for each $t\in T$, the cut $(S_t,V\setminus S_t)$ is a $t$-sided $(1+\gamma)$-fair cut.
    Using \Cref{thm:fair} to compute $(1+\gamma)$-fair $(s,t)$-cuts, our algorithm for $(1+\e)$-approximate minimum isolating cuts runs in $\tO(m/\e^3)$ time. 
\end{theorem}

\begin{algorithm}
    \caption{$(1+\e)$-approximate Minimum Isolating Cuts Algorithm on terminal set $T$}
    \label{alg:isolating}
    \begin{algorithmic}[1]
        \STATE {Arbitrarily order the terminals in $T = \{t_1, t_2, \ldots, t_{|T|}\}$}
        \STATE{\bf Phase 1:}
        \FOR {$i=1$ to $\lceil\lg |T|\rceil$}
            \STATE {$X_i \gets \{v_j\in T: i^{\rm th} \text{ bit in } j \text{ is } 1\}$} 
            \STATE {$Y_i \gets \{v_j\in T: i^{\rm th} \text{ bit in } j \text{ is } 0\}$}
            \STATE {Use \Cref{thm:fair} to find a $(1+\gamma)$-fair  $(X_i, Y_i)$-cut $S_i$}
        \ENDFOR
        \STATE{\bf Phase 2:}
        \FOR {every terminal $t\in T$}
            \STATE {Let $S_t$ be the connected component containing $t$ in $G\setminus\cup_i\delta S_i$, i.e., the graph where we delete all the edges in cuts $\delta S_i$ for all $i$.}
            \STATE {$G_t$ is obtained from $G$ by contracting all vertices in $V\setminus S_t$ into a single vertex $\bar{s}_t$.}
            \COMMENT{To implement this step efficiently, we construct a new graph that is identical to $G_t$ instead of contracting $G$.}
            \STATE {Find a $(1+\beta)$-approximate minimum $(t,\bar{s}_t)$-cut in graph $G_t$; call this cut $C_t$\label{line:iso:solve mincut}}
        \ENDFOR        
        \STATE {Return the cuts $\{C_t: t\in T\}$}
    \end{algorithmic}
\end{algorithm}
To establish \Cref{thm:isolating}, we describe \Cref{alg:isolating} for finding $(1+\e)$-approximate isolating cuts.
First, we establish correctness of the algorithm by showing that the cut $C_t$ returned by \Cref{alg:isolating} for a terminal $t\in T$ is indeed a $(1+\e)$-approximate minimum isolating cut for $T$. The following claim establishes an approximate version of the standard uncrossing property of minimum cuts, and is crucial for the correctness of our algorithm.

\begin{lemma}
\label{lem:intersect}
    Let $A$ be a $(1+\alpha)$-approximate minimum isolating cut for some terminal $t$ and let $B$ be a $(1+\gamma)$-fair $(X, Y)$-cut where $X\cup Y = T$, $t\in X$, and $X\subseteq B$. Then, $A\cap B$ is a $(1+\alpha)(1+\gamma)$-approximate minimum isolating cut for $t$.
\end{lemma}
\begin{proof}
    First, note that since $A$ is an isolating cut for $t$ and $t\in X, X\subseteq B$, it follows that $A\setminus B$ does not contain any terminal and $A\cap B$ contains a single terminal $t$. Now, consider the two cuts $A$ and $A\cap B$. Using the notation $\uplus$ for disjoint union, we can write 
    \begin{align*}
        E(A, V\setminus A) &= E(A\cap B, V\setminus (A\cup B)) \uplus E(A\cap B, B\setminus A) \uplus E(A\setminus B, V\setminus A) \\
        E(A\cap B, V\setminus (A\cap B)) &= E(A\cap B, V\setminus (A\cup B)) \uplus E(A\cap B, B\setminus A) \uplus E(A\cap B, A\setminus B).
    \end{align*}
    Since the first two sets are identical, we only need to compare the third sets $E(A\setminus B, V\setminus A)$ and $E(A\cap B, A\setminus B)$. Since $B$ is a $(1+\gamma)$-fair  $(X, Y)$-cut, there is a feasible flow from $X$ to $Y$ that, for each edge in $E(B, V\setminus B)$, sends at least $\frac1{1+\gamma}$ times capacity in the direction from $B$ to $V\setminus B$. Now, consider the flow on the subset of edges $E(A\cap B, A\setminus B) \subseteq E(B, V\setminus B)$. Since the flow must end at $Y$ and since $Y\cap (A\setminus B) = \emptyset$, it follows that this flow must exit the set $A\setminus B$ on the edges in $E(A\setminus B, V\setminus (A\cup B))$. Thus, 
    $$\delta(A\cap B, A\setminus B) \le (1+\gamma) \cdot \delta(A\setminus B, V\setminus (A\cup B)) \le (1+\gamma) \cdot \delta(A\setminus B, V\setminus A).$$
    It follows that $\delta(A\cap B) \le (1+\gamma)\cdot \delta(A)$, which proves the lemma.
\end{proof}

\begin{lemma}
    For $\gamma = \frac{\e}{4\lceil\lg |T|\rceil}$ and $\beta = \frac{\e}{4}$, the cut $C_t$ returned by \Cref{alg:isolating} is a $(1+\e)$-approximate minimum isolating cut for every $t\in T$.
\end{lemma}
\begin{proof}
    \Cref{lem:intersect} implies that in \Cref{alg:isolating}, the minimum isolating cut of $t$ in graph $G_t$, i.e., the minimum $t-\bar{s}_t$ cut, is a $(1+\gamma)^{\lceil\lg |T|\rceil}$-approximate minimum isolating cut of $t$ in the input graph $G$. Since $C_t$ is a $(1+\beta)$-approximate minimum $t-\bar{s}_t$ cut, it follows that $C_t$ is a $(1+\gamma)^{\lceil\lg |T|\rceil}\cdot (1+\beta)$-approximate minimum isolating cut of $t$ in the input graph $G$. Using the values of $\gamma$ and $\beta$, we have
    $$\left(1+\frac{\e}{4\lceil\lg |T|\rceil}\right)^{\lceil\lg |T|\rceil}\cdot \left(1+\frac{\e}{4}\right) \le e^{\e/4}\cdot e^{\e/4} = e^{\e/2} \le 1 + \e \text{ since } \e < 1.$$
\end{proof}

For the $(1+\beta)$-approximate mincut in Step~\ref{line:iso:solve mincut}, we can use \Cref{thm:fair} to compute a $(1+\gamma)$-fair cut, which is also a $(1+\beta)$-approximate mincut since $\gamma\le\beta$. This also guarantees that the cut $C_t$ is a $t$-sided $(1+\gamma)$-fair cut. Finally, it is clear from the algorithm that all cuts $C_t$ are disjoint. 

The runtime analysis is identical to that in \cite{LiP20}, so we omit it for brevity.  

\subsection{$(1+\e)$-approximate Minimum Steiner Cut}
\label{sec:steiner}

As an immediate application of our isolating cut result, we can solve the Steiner cut problem below efficiently. 

\begin{definition}
    Given an undirected graph $G = (V, E)$ with non-negative edge weights and a set of terminals $T\subseteq V$, a minimum Steiner cut is a cut of minimum value among all cuts $\emptyset \subset S \subset V$ that satisfy $\emptyset\subset S\cap T \subset T$.
\end{definition}

Using \Cref{thm:isolating}, we give the following algorithm for finding a $(1+\e)$-approximate minimum Steiner cut.

\begin{algorithm}
    \caption{$(1+\e)$-approximate minimum Steiner cut Algorithm on terminal set $T$}
  \label{alg:steiner}
    \begin{algorithmic}
        \FOR {$i=1$ to $\lceil\lg |T|\rceil$}
            \FOR {$j=1$ to $\lceil\log_{8/7} n\rceil$}        
                \STATE {$T_{ij}$ is drawn i.i.d. from $T$ where every vertex $t\in T$ appears in $T_{ij}$ with probability $1/2^i$}
                \STATE {Use \Cref{thm:isolating} to find isolating cuts $\cS_{ij} = \{S_t: t\in T_{ij}\}$ for the terminal set $T_{ij}$}
            \ENDFOR
        \ENDFOR
        \STATE {Return $\arg\min \{\delta(S): S\in \cS_{ij}, i\in[\lceil\lg |T|\rceil], j\in [\lceil\log_{8/7} n\rceil]\}$}
    \end{algorithmic}
\end{algorithm}

\begin{theorem}
    Given an undirected graph $G = (V, E)$ on $m$ edges and $n$ vertices and with non-negative edge weights and a set of terminals $T\subseteq V$, \Cref{alg:steiner} computes a $(1+\e)$-minimum Steiner cut for $T$ wuth probability at least $1-1/n$ in $\tO(m)$ time.
\end{theorem}
\begin{proof}
    Fix a minimum Steiner cut for the terminal set $T$ and let $S$ denote the side of this cut such that $|T\cap S| \le |T\setminus S|$. Let $i\in[\lceil\lg |T|\rceil]$ such that $2^{i-1} \le |S\cap T| < 2^i$. Then, $T_{ij}$ contains exactly one vertex in $T\cap S$ with probability
    $$|T\cap S|\cdot \frac{1}{2^i}\cdot \left(1 - \frac{1}{2^i}\right)^{|T\cap S|-1}
    \ge 2^{i-1}\cdot \frac{1}{2^i}\cdot \left(1 - \frac{1}{2^i}\right)^{2^i}
    \ge \frac{1}{2}\cdot \frac{1}{4}
    = \frac{1}{8}.$$
    This implies that the probability that there is no index $j \in [\lceil\log_{8/7} n\rceil]$ such that $T_{ij}$ contains exactly one terminal in $T\cap S$ is at most $1/n$, thereby establishing the correctness of the algorithm.
    
    The running time bound follows from \Cref{thm:isolating}.
\end{proof}

\newcommand{\lar}{\textup{large}}

\section{Approximate Gomory-Hu Tree Algorithm}
\label{sec:ghtree}
The main result in this section is the near-linear time algorithm for computing an approximate Gomory-Hu tree. In fact, our algorithm can solve a more general problem called approximate Gomory-Hu Steiner tree defined below. (The definition is copied verbatim from \cite{LiP21}.)

\begin{definition}[Approximate Gomory-Hu Steiner tree]
Given a graph $G=(V,E)$ and a set of terminals $U\subseteq V$, the $(1+\epsilon)$-approximate Gomory-Hu Steiner tree is a weighted tree $T$ on the vertices $U$, together with a function $f:V\to U$, such that
 \begin{itemize}
 \item For all $s,t\in U$, consider the minimum-weight edge $(u,v)$ on the unique $s-t$ path in $T$. Let $U'$ be the vertices of the connected component of $T-(u,v)$ containing $s$.
Then, the set $f^{-1}(U')\subseteq V$ is a $(1+\epsilon)$-approximate $(s,t)$-mincut, and its value is $w_T(u,v)$.
 \end{itemize}
\end{definition}

Our main result is stated below. Recall that we assume that the ratio between the largest and lowest edge weights are $\poly(n)$.

\begin{restatable}{theorem}{ApproxW}\label{thm:approx-w}
Let $G$ be a weighted, undirected graph, and let $U$ be a subset of vertices. There is a randomized algorithm that w.h.p., outputs a $(1+\epsilon)$-approximate Gomory-Hu Steiner tree in $\tilde{O}(m \cdot \poly(1/\epsilon))$ time.
\end{restatable}

The algorithm and analysis are similar to those in~\cite{LiP21}, except we replace (exact) minimum isolating cuts with an approximate version, which requires overcoming a few more technical issues. For completeness, we redo all the proofs. We also restate \Cref{thm:isolating} below in the form we precisely need.

\begin{theorem}
\label{thm:isolating2}
    Fix any $\epsilon < 1$. Given an undirected graph $G = (V, E)$ on $m$ edges and $n$ vertices with non-negative edge weights and a set of terminals $T\subseteq V$, there is an algorithm that outputs a $(1+\epsilon)$-approximate minimum isolating cut $S_t\subseteq V$ for every terminal $t\in T$ in $\tilde{O}(m/\epsilon^{O(1)})$ time. Moreover, the sets $S_t$ are disjoint, and for each $t\in T$, the set $S_t$ is a $t$-sided $(1+\gamma)$-fair $(\{t\},T\setminus\{t\})$-cut.
\end{theorem}

\subsection{Cut Threshold Step Algorithm}

We begin with the following ``Cut Threshold Step'' subroutine from \cite{LiP21}, described in Algorithm~\ref{alg:step} below. Loosely speaking, the algorithm inputs a source vertex $s$ and a threshold $W$, and aims to find a large fraction of vertices whose mincut from $s$ is approximately at most $W$.

\begin{algorithm}
    \caption{$(1+\gamma)$-approximate ``Cut Threshold Step'' on inputs $(G,U,W,s)$}
    \label{alg:step}
    \begin{algorithmic}[1]
        \STATE {Initialize $D\gets\emptyset$}
        \FOR {independent iteration $i\in\{0,1,2,\ldots,\lfloor\lg|U|\rfloor\}$}
            \STATE {$R^{i}\gets$ sample of $U$ where each vertex in $U\setminus \{s\}$ is sampled independently with probability $1/2^i$, and $s$ is sampled with probability $1$}            \STATE {Compute $(1+\frac\gamma{2\lceil\lg|U|\rceil})$-approximate minimum isolating cuts $\{S^i_v:v\in R^i\}$ on inputs $G$ and $R^i$ with the additional guarantees of \Cref{thm:isolating2}} (for large enough constant $c>0$)
                        \STATE {Let $\mathcal F^i$ be the family of sets $S^i_v$ satisfying $\delta S^i_v\le(1+\gamma)W$, and let $D^i\gets\bigcup_{S_v^i\in\mathcal F^i}S_v^i\cap U$}\label{line:add-to-Di}
                        \STATE {Let $\widetilde R^i\subseteq R^i$ be the set of all $v\in R^i$ satisfying $\delta S^i_v\le(1+\gamma)W$}
        \ENDFOR
        \STATE{Let $i_{\max}$ be the index $i$ maximizing $|D^i|$}
        \STATE{Return $D\gets D^{i_{\max}}$, $R\gets\widetilde R^{i_{\max}}$, and $\mathcal F\gets\mathcal F^{i_{\max}}$}
    \end{algorithmic}
\end{algorithm}

\begin{lemma}\label{lem:add}
For any $i$, each set $S^i_v$ added to $D^i$ satisfies $\lambda(s,v)\le(1+\gamma)W$.
\end{lemma}
\begin{proof}
For each $v\in D^i$, the corresponding set $S^i_v$ on line~\ref{line:add-to-Di} contains $v$ and not $s$, so $\lambda(s,v)\le\delta S^i_v\le(1+\gamma)W$.
\end{proof}

\begin{lemma}\label{lem:step}
Let $D^*$ be all vertices $v\in U\setminus s$ for which
there exists an $(s,v)$-cut in $G$ of weight at most $W$ whose side containing $v$ has at most $|U|/2$ vertices in $U$. Then,
$\mathbb E[|D|]=\Omega(|D^*|/\log|U|)$.
\end{lemma}
\begin{proof}
We will show that
\BALN \mathbb E\left[\sum_{i=0}^{\lfloor\lg|U|\rfloor}|D^i|\right]\ge\Omega(|D^*|) \label{eq:expected} ,\EALN
 which is sufficient, since the largest $D^i$ will have at least $1/(\lfloor\lg|U|\rfloor+1)$ fraction of the total size. Fix a vertex $v\in D^*$. 
For each $0\le j\le\lceil\lg|U|\rceil$, define $C_v^j\subseteq V$ as the $(s,v)$-cut of weight at most $(1+\frac\gamma{2\lceil\lg|U|\rceil})^jW$ that minimizes $|C_v^j\cap U|$, which must exist since $v\in D^*$. By construction, $|C^j_v\cap U|$ is decreasing in $j$.

We focus on a value $j^*$ ($0\le j^*<\lceil\lg|U|\rceil$) satisfying $|C_v^{j^*+1}\cap U|\ge|C_v^{j^*}\cap U|/2$, which is guaranteed to exist. Consider sampling iteration $i=\lfloor\lg|C^{j^*}_v\cap U|\rfloor$, where each vertex in $U\setminus\{s\}$ is sampled with probability $1/2^i$. With probability $\Omega(1/|C_v^{j^*}\cap U|)$, we have $C_v^{j^*}\cap R^i=\{v\}$, i.e., we sampled $v$ and nothing else in $C_v^{j^*}\cap U$. If this occurs, then $C_v^{j^*}$ is a valid isolating cut separating $v$ from $R^i\setminus\{v\}$. Since $S^i_v$ is a $(1+\frac\gamma{2\lceil\lg|U|\rceil})$-approximate minimum isolating cut, we have
\[ \delta S^i_v \le \left(1+\frac\gamma{2\lceil\lg|U|\rceil}\right)\delta C_v^{j^*}\le\left(1+\frac\gamma{2\lceil\lg|U|\rceil}\right)^{j^*+1}W \le\left(1+\frac\gamma{2\lceil\lg|U|\rceil}\right)^{\lceil\lg|U|\rceil}W \le e^{\gamma/2}W \le (1+\gamma)W ,\]
so $S^i_v\cap U$ is added to $D^i$ on line~\ref{line:add-to-Di}.
By definition of $C_v^{j^*+1}$, we have $|S^i_v\cap U|\ge|C_v^{j^*+1}\cap U|$, which is at least $|C_v^{j^*}\cap U|/2$ by our choice of $j^*$. In other words, if $C_v^{j^*}\cap R^i=\{v\}$, which occurs with probability $\Omega(1/|C_v^{j^*}\cap U|)$, then $v$ is ``responsible'' for adding at least $|C_v^{j^*}\cap U|/2$ vertices to $D^i$.

Thus, each vertex $v\in D^*$ is responsible for adding $\Omega(1)$ vertices in expectation to some $D^i$, which increases $\mathbb E\left[\sum_{i=0}^{\lfloor\lg|U|\rfloor}|D^i|\right]$ by $\Omega(1)$ in expectation. Finally, (\ref{eq:expected}) follows by linearity of expectation over all $v\in D^*$.
\end{proof}

For our approximate Gomory-Hu tree algorithm, we actually need a bound on $\mathbb E[|D\cap D^*|]$, not $\mathbb E[|D|]$, since we want to remove $D$ from $U$ and claim that the size of the new $D^*$ drops by a large enough factor. Unfortunately, it is possible that $D$ is largely disjoint from $D^*$, so a bound on $\mathbb E[|D|]$ does not directly translate to a bound on $\mathbb E[|D\cap D^*|]$. Therefore, we wrap Algorithm~\ref{alg:step} into another routine that achieves a good bound on $\mathbb E[|D\cap D^*|]$. We actually prove the stronger guarantee that $D^*$ can be any \emph{subset} of all vertices $v\in U\setminus s$ for which $\lambda(s,v)\le W$, which is needed in our Gomory-Hu tree algorithm.

\begin{algorithm}
    \caption{$(1+\gamma)$-approximate Gomory-Hu Steiner tree ``step'' on inputs $(G,U_0,W_0,s)$}
    \label{alg:step2}
    \begin{algorithmic}
      \STATE {Initialize $U\gets U_0$}
      \FOR {$O(\log^3n)$ sequential iterations}
        \FOR {independent iteration $j\in\{0,1,2,\ldots,\lceil\lg|U|\rceil-1\}$}
            \STATE Call Algorithm~\ref{alg:step} on parameter $\frac\gamma{2\lceil\lg|U|\rceil}$ and inputs $(G,U,(1+\frac\gamma{2\lceil\lg|U|\rceil})^jW_0,s)$ and let $(D_j,R_j,\mathcal F_j)$ be the output
        \ENDFOR
        \STATE Update $U\gets U\setminus \bigcup_jD_j$ for the values $D_j$ computed on this sequential iteration
      \ENDFOR
      \STATE{Return an output $(D,R,\mathcal F)$ selected uniformly at random out of the  $O(\log^3n\log|U|)$ calls to Algorithm~\ref{alg:step}.}
    \end{algorithmic}
\end{algorithm}

\begin{lemma}\label{lem:add2}
Each set $S\in\mathcal F$ in the output $(D,R,\mathcal F)$ of Algorithm~\ref{alg:step2} satisfies $\delta S\le(1+\gamma)W_0$.
\end{lemma}
\begin{proof}
By \Cref{lem:add} applied to any $j\in\{0,1,2,\ldots,\lceil\lg|U|\rceil-1\}$, each set $S\in\mathcal F_j$ satisfies 
\[ \delta S\le\left(1+\frac\gamma{2\lceil\lg|U|\rceil}\right)\cdot\left(1+\frac\gamma{2\lceil\lg|U|\rceil}\right)^jW_0\le\left(1+\frac\gamma{2\lceil\lg|U|\rceil}\right)^{\lceil\lg|U|\rceil}W_0 \le e^{\gamma/2}W_0\le(1+\gamma)W_0.\] So the same holds for the randomly chosen output $(D,R,\mathcal F)$.
\end{proof}

\begin{lemma}\label{lem:step2}
Let $D^*$ be an arbitrary set of vertices $v\in U\setminus s$ satisfying $\lambda(s,v)\le W_0$. The output $(D,R,\mathcal F)$ satisfies $\mathbb E[D\cap D^*]\ge\Omega(|D^*|/\log^4n)$.
\end{lemma}
\begin{proof}
We claim that after $O(\log^3n)$ iterations of the main for loop, the set $D^*\cap U$ becomes empty. This would mean that $D^*$ is contained in the union of all $O(\log^4n)$ sets $D_j$ computed over all iterations, so a random set $D_j$ must contain a $\Omega(1/\log^4n)$ fraction of $D^*$ in expectation.
For the rest of the proof, we prove this claim. 

For each $0\le j\le\lceil\lg|U|\rceil$, let $D^*_j$ be all vertices $v\in U\setminus s$ for which $\lambda(s,v)\le(1+\frac\gamma{2\lceil\lg|U|\rceil})^jW_0$. By construction, $D^*\subseteq D^*_0\subseteq D^*_1\subseteq\cdots\subseteq D^*_{\lceil\lg|U|\rceil}$. We track the sets $D^*_j\cap U$ throughout the algorithm. Consider the set $U$ at the beginning of one of the $O(\log^3|U|)$ sequential iterations. We focus on a value $j^*$ ($0\le j^*<\lceil\lg|U|\rceil$) satisfying $|D^*_{j^*}\cap U|\ge|D^*_{j^*+1}|/2$. Consider iteration $j^*$ of the inner for loop. By \Cref{lem:add}, we have $\lambda(s,v)\le(1+\frac\gamma{2\lceil\lg|U|\rceil})\cdot(1+\frac\gamma{2\lceil\lg|U|\rceil})^{j^*}W_0=(1+\frac\gamma{2\lceil\lg|U|\rceil})^{j^*+1}W_0$, so in particular, $D_{j^*}\subseteq D^*_{j^*+1}$. By \Cref{lem:step}, we have $\mathbb E[|D_{j^*}|]\ge\Omega(|D^*_{j^*}|/\log|U|) \ge \Omega(|D^*_{j^*+1}|/\log|U|)$. Therefore, once we delete $\bigcup_jD_j$ at the end of this sequential iteration, the size of $D^*_{j^*+1}$ drops by factor $(1-\Omega(1/\log|U|))$ in expectation.

In other words, on each sequential iteration, there exists $j^*$ ($1\le j^*\le\lceil\lg|U|\rceil$) for which the size of $D^*_j\cap U$ drops by factor $(1-\Omega(1/\log|U|))$ in expectation. Since the other sets $D^*_{j'}\cap U$ can never increase in size, the product $\prod_{j=1}^{\lceil\lg|U|\rceil}|D^*_j\cap U|$ decreases by factor $(1-\Omega(1/\log|U|))$ in expectation. Since the product is at most $|U|^{\lceil\lg|U|\rceil}\le2^{O(\log^2n)}$ initially, it follows that after $O(\log^3n)$ sequential iterations, the product becomes zero w.h.p. Therefore, at the end of the algorithm, there exists $j$ ($1\le j^*\le\lceil\lg|U|\rceil$) with $D^*_j\cap U=\emptyset$. Since $D^*\subseteq D^*_j$, we also get $D^*\cap U=\emptyset$, which proves the claim.
\end{proof}

\subsection{The Algorithm for Approximating Gomory-Hu Steiner Tree}

We present our approximate Gomory-Hu tree algorithm in Algorithm~\ref{alg:approxGH}. It uses Algorithm~\ref{alg:step2} as a subroutine. See Figure~\ref{fig:ghtree-alg} for a visual guide to the algorithm. Once again, the algorithm and analysis closely follow those in \cite{LiP21}. 

We require the lemma below for both running time and approximation guarantee analysis.

\begin{algorithm}
\caption{$(1+\epsilon)$-approximate Gomory-Hu Steiner tree on inputs $(G_0,U)$. Assume $\epsilon<1/100$.}\label{alg:approxGH}
\begin{algorithmic}[1]
\STATE If $|U|=1$, then return the trivial Gomory-Hu Steiner tree $(T,f)$ where $T$ is the empty tree on the single vertex $u\in U$, and $f(v)=u$ for all vertices $v$. Otherwise, if $|U|>1$, then do the steps below.
\STATE $\gamma\gets \epsilon^2/\log^6n$ 
\STATE {$\lambda\gets(1+\epsilon)$-approximate global Steiner mincut of $G$ with terminals $U$, so that the Steiner mincut is in the range $[(1-\epsilon)\lambda,\lambda]$}
\STATE {$s\gets$ uniformly random vertex in $U$}
\STATE {Construct graph $G'$ by starting with $G$ and adding an edge $(s,u)$ of weight $18\epsilon\lambda/|U|$ for each $u\in U$}
\STATE Call Algorithm~\ref{alg:step2} on parameter $\gamma$ and inputs $(G',U,(1+10\epsilon)\lambda,s)$, and let $(D,R,\mathcal F)$ be the output. Write $\mathcal F=\{S_v:v\in R\}$ where $v\in S_v$ for all $v\in R$.\label{line:F}
\STATE{\bf Phase 1: Construct recursive graphs and apply recursion}
\FOR{each $v\in  R$} 
 \STATE Let $G_v$ be the graph $G$ with vertices $V\setminus S_v$ contracted to a single vertex $x_v$ 
 \STATE Let $U_v\gets S_v\cap U$
 \STATE Recursively call $(G_v,U_v)$ to obtain output $(T_v,f_v)$ \label{line:recursive}
\ENDFOR
\STATE Let $G_\lar$ be the graph $G$ with (disjoint) vertex sets $S_v$ contracted to single vertices $y_v$ for all $v\in R$
\STATE Let $U_\lar\gets U\setminus\bigcup_{v\in R}(S_v\cap U)$
\STATE Recursively call $(G_\lar,U_\lar)$ to obtain $(T_\lar,f_\lar)$
\STATE {\bf Phase 2: Merge the recursive Gomory-Hu Steiner trees}
\STATE Construct $T$ by starting with the disjoint union $T_\lar\cup\bigcup_{v\in R}T_v$ and, for each $v\in R$, adding an edge between $f_v(x_v)\in U_v$ and $f_\lar(y_v)\in U_\lar$ of weight $w(\partial_GS_v)$\label{line:combine-T}
\STATE Construct $f:V\to U$ by $f(v')=f_\lar(v')$ if $v'\in U_\lar$ and $f(v')=f_v(v')$ if $v'\in U_v$ for some $v\in R$\label{line:combine-f}
\STATE Return $(T,f)$
\end{algorithmic}
\end{algorithm}

\begin{lemma}\label{lem:Sv}
Each set $S\in\mathcal F$ satisfies $\delta_GS\le(1+\gamma)(1+10\epsilon)\lambda$ and $|S\cap U|\le2|U|/3$.
\end{lemma}
\begin{proof}
By \Cref{lem:add2} on the call to Algorithm~\ref{alg:step2} (line~\ref{line:F}), each set $S\in\mathcal F$ satisfies $\delta_{G'}S\le(1+\gamma)\cdot(1+10\epsilon)\lambda$, so $\delta_GS\le\delta_{G'}S\le(1+\gamma)(1+10\epsilon)\lambda$. We now prove the second statement. By construction, the cut $\partial _{G'}S$ has $|S\cap U|$ edges of weight $18\epsilon\lambda/|U|$ that were added to $G'$. Since $\partial_GS$ is a valid Steiner cut in $G$ and the Steiner mincut is at least $(1-\epsilon)\lambda$, the cut $\partial _{G'}S$ has at least $(1-\epsilon)\lambda$ weight of edges from $G$. So $\delta _{G'}S\ge (1-\epsilon)\lambda+|S\cap U|\cdot18\epsilon\lambda/|U|$. Suppose for contradiction that $|S\cap U|>2|U|/3$; then, this becomes $\delta_{G'}S>(1-\epsilon)\lambda+12\epsilon\lambda=(1+11\epsilon)\lambda$, which contradicts the earlier statement $\delta_{G'}S\le(1+\gamma)(1+10\epsilon)\lambda$.
\end{proof}

\begin{figure}\label{fig:ghtree-alg}\centering
\includegraphics[scale=.5]{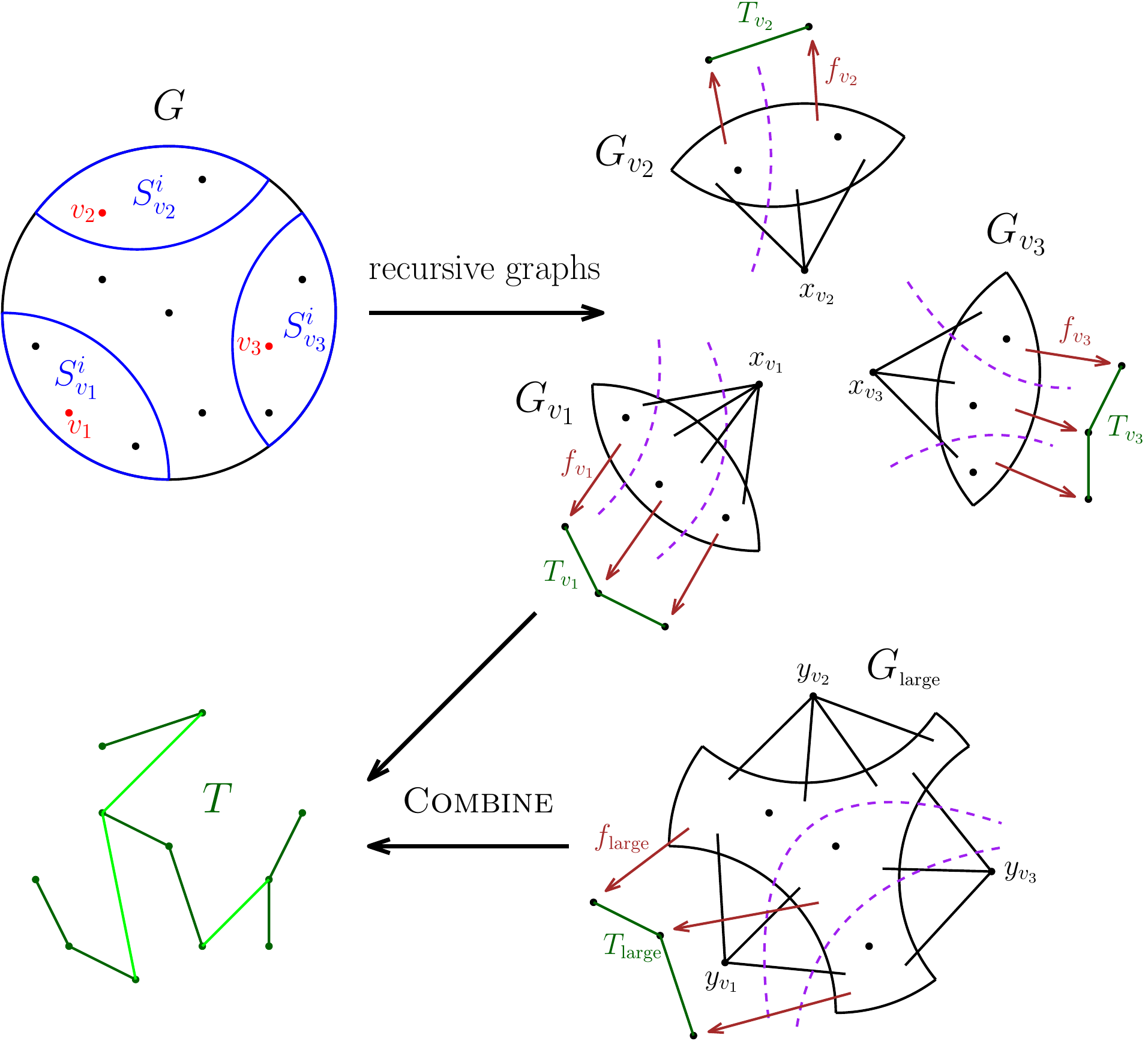}
\caption{
Recursive construction of $G_\lar$ and $G_v$ for $v\in R$. Here, $R=\{v_1,v_2,v_3\}$, denoted by red vertices on the top left. The dotted blue curves on the right mark the boundaries of the regions $f_{v_i}^{-1}(u):u\in U_{v_i}$ and $f_{v_\lar}^{-1}(u):u\in U_\lar$. The light green edges on the bottom left are the edges $(f_{v_i}(x_{v_i}),f_\lar(y_{v_i}))$ added on line~\ref{line:combine-T}.
}
\end{figure}

\subsection{Running Time Bound}

Let $P(G,U,W)$ be the set of unordered pairs of distinct vertices whose mincut is at most $W$:
\[ P(G,U,W) = \bigg\{ \{u,v\}\in\binom U2:\lambda_G(u,v)\le W \bigg\} .\]
In particular, we will consider its size $|P(G,U,W)|$, and show the following expected reduction:

\begin{restatable}{lemma}{LemD}\label{lem:P}
For any $W$ that is at most $(1+\epsilon)$ times the Steiner mincut of $G$, we have
\[ \mathbb E[|P(G_\lar,U_\lar,W)|] \le \left(1-\Omega\left(\frac1{\log^4 n}\right)\right)|P(G,U,W)| ,\]
where the expectation is taken over the random selection of $s$ and the randomness in Algorithm~\ref{alg:step2}.
\end{restatable}

Before we prove \Cref{lem:P}, we show how it implies progress on the recursive call for $G_\lar$.
\begin{corollary}\label{cor:mincut-increase}
Let $\lambda_0$ be the global Steiner mincut of $G$.
W.h.p., after $\Omega(\log^5n)$ recursive calls along $G_\lar$ (replacing $G\gets G_\lar$ each time), the global Steiner mincut of $G$ is at least $(1+\epsilon)\lambda_0$ (where $\lambda_0$ is still the global Steiner mincut of the initial graph).
\end{corollary}
\begin{proof}
Let $W=(1+\epsilon)\lambda_0$.
Initially, we trivially have $|P(G,U,W)|\le\binom{|U|}2$. The global Steiner mincut can only increase in the recursive calls, since $G_\lar$ is always a contraction of $G$, so $W$ is always at most $(1+\epsilon)$ times the current Steiner mincut of $G$. By \Cref{lem:P}, the value $|P(G,U,W)|$ drops by factor $1-\Omega(\frac1{\log^4 n})$ in expectation on each recursive call, so after $\Omega(\log^5n)$ calls, we have
\[ \mathbb E[|P(G,U,W)|]\le\binom{|U|}2\cdot\left(1-\Omega\left(\frac1{\log^4n}\right)\right)^{\Omega(\log^5n)}\le\frac1{\textup{poly}(n)} .\]
In other words, w.h.p., we have $|P(G,U,W)|=0$ at the end, or equivalently, the Steiner mincut of $G$ is at least $(1+\epsilon)\lambda_0$.
\end{proof}

Combining both recursive measures of progress together, we obtain the following bound on the recursion depth:
\begin{lemma}\label{lem:depth}
W.h.p., each path down the recursion tree of Algorithm~\ref{alg:approxGH} has $O(\log n)$ calls on a graph $G_v$, and between two consecutive such calls, there are  $O(\epsilon^{-1}\log^6n)$ calls on the graph $G_\lar$.
\end{lemma}
\begin{proof}
For any $\Theta(\log^5n)$ successive recursive calls down the recursion tree, either one call was on a graph $G_v$, or all $\Theta(\log^5n)$ of them were on the graph $G_\lar$. In the former case, $|U|$ drops by a constant factor by \Cref{lem:Sv}, so it can happen $O(\log n)$ times total. In the latter case, by \Cref{cor:mincut-increase}, the global Steiner mincut increases by factor $(1+\epsilon)$. Let $w_{\min}$ and $w_{\max}$ be the minimum and maximum weights in $G$, so that $\Delta=w_{\max}/w_{\min}$, which we assume to be $\poly(n)$. Note that for any recursive instance $(G',U')$ and any $s,t\in U'$, we have $w_{\min}\le\lambda_{G'}(s,t)\le w(\partial(\{s\}))\le nw_{\max}$, so the global Steiner mincut of $(G',U')$ is always in the range $[w_{\min},nw_{\max}]$. It follows that the global Steiner mincut can increase by factor $(1+\epsilon)$ at most $O(\epsilon^{-1}\log(nw_{\max}/w_{\min}))=O(\epsilon^{-1}\log n)$ times. Therefore, there are at most $O(\epsilon^{-1}\log^6n)$ consecutive calls on $G_\lar$ before a call on some $G_v$ must occur.
\end{proof}

\begin{lemma}\label{lem:runtime}
For an unweighted/weighted graph $G=(V,E)$, and terminals $U\subseteq V$, Algorithm~\ref{alg:approxGH} takes time $\tilde{O}(m\epsilon^{-1})$ plus calls to \Cref{thm:isolating2} with parameter $\gamma=\epsilon^2/\log^6n$ on unweighted/weighted instances with  a total of $\tilde{O}(n\epsilon^{-1})$ vertices and $\tilde{O}(m\epsilon^{-1})$ edges.
\end{lemma}
\begin{proof}
For a given recursion level, consider the instances $\{ (G_i,U_i,W_i)\}$ across that level. By construction, the terminals $U_i$ partition $U$. Moreover, the total number of vertices over all $G_i$ is at most $n+2(|U|-1)=O(n)$ since each branch creates $2$ new vertices and there are at most $|U|-1$ branches. 

To bound the total number of edges, we consider the unweighted and weighted cases separately, starting with the unweighted case. The total number of new edges created is at most the sum of weights of the edges in the final $(1+\epsilon)$-approximate Gomory-Hu Steiner tree. For an unweighted graph, this is $O(m)$ by the following well-known argument. Root the Gomory-Hu Steiner tree $T$ at any vertex $r\in U$; for any $v\in U\setminus r$ with parent $u$, the cut $\partial\{v\}$ in $G$ is a $(u,v)$-cut of value $\deg(v)$, so $w_T(u,v)\le(1+\epsilon)\lambda_G(u,v)\le(1+\epsilon)\deg(v)$. Overall, the sum of the edge weights in $T$ is at most $(1+\epsilon)\sum_{v\in U}\deg(v)\le(1+\epsilon)\cdot2m$.

For the weighted case, define a \emph{parent} vertex in an instance as a vertex resulting from either (1)~contracting $V\setminus S_v$ in some previous recursive $G_v$ call, or (2)~contracting a component containing a parent vertex in some previous recursive call. There are at most $O(\log n)$ parent vertices: at most $O(\log n)$ can be created by~(1) since each $G_v$ call decreases $|U|$ by a constant factor (\Cref{lem:Sv}), and (2)~cannot increase the number of parent vertices. Therefore, the total number of edges adjacent to parent vertices is at most $O(\log n)$ times the number of vertices. Since there are $O(n)$ vertices in a given recursion level, the total number of edges adjacent to parent vertices is $O(n\log n)$ in this level. Next, we bound the number of edges not adjacent to a parent vertex by $m$. To do so, we first show that on each instance, the total number of these edges over all recursive calls produced by this instance is at most the total number of such edges in this instance. Let $P\subseteq V$ be the parent vertices; then, each $G_v$ call has exactly $|E(G[S_v\setminus P])|$ edges not adjacent to parent vertices (in the recursive instance), and the $G_\lar$ call has at most $|E(G[V\setminus P]) \setminus \bigcup_{v\in R}E(G[S_v\setminus P])|$, and these sum to $|E(G[V\setminus P])|$, as promised. This implies that the total number of edges not adjacent to a parent vertex at the next level is at most the total number at the previous level. Since the total number at the first level is $m$, the bound follows.

Therefore, there are $O(n)$ vertices and $\tilde{O}(m)$ edges in each recursion level. By \Cref{lem:depth}, there are $O(\epsilon^{-1}\log^6n)$ levels, for a total of $\tilde{O}(n\epsilon^{-1})$ vertices and $\tilde{O}(m\epsilon^{-1})$ edges. In particular, the instances to the max-flow calls have $\tilde{O}(n\epsilon^{-1})$ vertices and $\tilde{O}(m\epsilon^{-1})$ edges in total.
\end{proof}

Finally, we prove \Cref{lem:P}, restated below.
\LemD*
\begin{proof}
Define $D^*$ as the set of vertices $v\in U\setminus s$ for which there exists an $(s,v)$-cut in $G$ of weight at most $W$ whose side containing $v$ has at most $|U|/2$ vertices in $U$. Let $P_{\text{ordered}}(G,U,W)$ be the set of ordered pairs $(u,v):u,v\in V$ for which there exists a $(u,v)$-mincut of weight at most $W$ with at most $|U|/2$ vertices in $U$ on the side $S(u,v)\subseteq V$ containing $u$. 
We now state and prove the following four properties:

\begin{enumerate}
\item[(a)] For all $u,v\in U$, $\{u,v\}\in P(G,U,W)$ if and only if either $(u,v)\in P_{\text{ordered}}(G,U,W)$ or $(v,u)\in P_{\text{ordered}}(G,U,W)$ (or both).
\item[(b)] For each pair $(u,v)\in P_{\text{ordered}}(G,U,W)$, we have $u\in D^*$ with probability at least $1/2$,
\item[(c)] For each $u\in D^*$, there are at least $|U|/2$ vertices $v\in U$ for which $(u,v)\in P_{\text{ordered}}(G,U,W)$.
\item[(d)] Over the randomness in Algorithm~\ref{alg:step} on $(G,U,(1+\epsilon)\lambda)$, $\mathbb E[|D\cap D^*|]\ge\Omega(|D^*|/\log^4|U|)$.
\end{enumerate}

Property (a) follows by definition.
Property~(b) follows from the fact that $u\in D^*$ whenever $s\notin S(u,v)$, which happens with probability at least $1/2$. 
Property~(c) follows because any vertex $v\in U\setminus S(u,v)$ satisfies $(u,v)\in P_{\text{ordered}}(G,U,W)$, of which there are at least $|U|/2$. For property~(d), observe by construction of $G'$ that for each vertex $v\in D^*$, the $(s,v)$-mincut has weight at most $W+|U|/2\cdot18\epsilon\lambda/|U|$. This is at most $(1+\epsilon)\lambda + 9\epsilon\lambda = (1+10\epsilon)\lambda$ since $W$ is at most $(1+\epsilon)$ times the Steiner mincut of $G$ (which is at most $\lambda$). It follows that each $v\in D^*$ satisfies $\lambda_{G'}(s,v)\le(1+10\epsilon)\lambda$. Property~(d) follows from \Cref{lem:step2} applied to input $(G,U,(1+10\epsilon)\lambda,s)$ and set $D^*$.

With properties (a) to (d) in hand, we now finish the proof of \Cref{lem:P}. For any vertex $u\in D$, all pairs $(u,v)\in P_{\text{ordered}}(G,U,W)$ (over all $v\in U$) disappear from $P_{\text{ordered}}(G,U,W)$, which is at least $|U|/2$ many by (c). In other words, 
\BAL & |P_{\text{ordered}}(G,U,W)\setminus P_{\text{ordered}}(G_\lar,U_\lar,W)| 
\ge \frac{|U|\cdot|D|}2 
.\EAL
Taking expectations and applying (d), 
\BAL & 
\mathbb E[|P_{\text{ordered}}(G,U,W)\setminus P_{\text{ordered}}(G_\lar,U_\lar,W)|] 
\ge\frac{|U|\cdot\mathbb E[|D|]}{2}    
 \ge \Omega\left(\frac{|U|\cdot|D^*|}{\log^4|U|}\right)  .\EAL
Moreover,
\[ |U|\cdot|D^*| \ge \mathbb E\big[\big| \{(u,v): u\in D^*\} \big|\big] \ge \frac12|P_{\text{ordered}}(G,U,W)|,  \]
where the second inequality follows by (b). Putting everything together, we obtain
\BAL &\mathbb E[|P_{\text{ordered}}(G,U,W)\setminus P_{\text{ordered}}(G_\lar,U_\lar,W)|] 
\ge\Omega\left(\frac{|P_{\text{ordered}}(G,U,W)|}{\log^4|U|}\right)   .\EAL
Finally, applying (a) gives
\BG \mathbb E[|P(G,U,W) \setminus P(G_\lar,U_\lar,W)|] \ge \Omega\left(\frac{|P(G,U,W)|}{\log^4|U|}\right) .\nonumber
\EG
Finally, we have $P(G_\lar,U_\lar,W) \subseteq P(G,U,W)$ since the $(u,v)$-mincut for $u,v\in U_\lar$ can only increase in $G_\lar$ due to $G_\lar$ being a contraction of $G$. Therefore,
\BAL &|P(G,U,W)| - |P(G_\lar,U_\lar,W)| 
= |P(G,U,W) \setminus P(G_\lar,U_\lar,W)| ,\EAL
and combining with the bound on  $\mathbb E[|P(G,U,W) \setminus P(G_\lar,U_\lar,W)|]$ concludes the proof.
\end{proof}

\subsection{Approximation}

We first prove the two lemmas below before concluding the approximation guarantee.

\begin{lemma}\label{lem:large}
For any distinct vertices $p,q\in U_\lar$, we have $\lambda_G(p,q)\le\lambda_{G_\lar}(p,q)\le(1+\gamma)\lambda_G(p,q)$.
\end{lemma}
\begin{proof}
Since $G_\lar$ is a contraction of $G$, we have $\lambda_G(p,q)\le\lambda_{G_\lar}(p,q)$. To show the other inequality, fix any $(p,q)$-mincut $(A,B)$ in $G$. We iteratively ``uncross'' the cut $(A,B)$ with each set $S_v\in\mathcal F$ ($v\in R$) as follows: if $v\in A$, then replace $(A,B)$ with $(A\cup S_v,B\setminus S_v)$, and if $v\in B$, then replace $(A,B)$ with $(A\setminus S_v,B\cup S_v)$. By construction, the final cut is a $(p,q)$-cut that contains each $S_v$ on one side of the cut, so it survives upon contraction into $G_\lar$ and is a valid $(p,q)$-cut in $G_\lar$. We claim that the final cut has weight at most $(1+\gamma)\lambda_G(p,q)$, which would prove $\lambda_{G_\lar}(p,q)\le(1+\gamma)\lambda_G(p,q)$.

Let $(A,B)$ be the current cut in the iterative process, and let $S_v$ be the next cut we wish to uncross. Since $S_v$ is a $v$-sided $(1+\gamma)$-fair cut on $G'$, there is a feasible flow with no source/sink in $S_v\setminus\{v\}$ and which saturates $\partial _{G'}S_v$ up to factor $\frac1{1+\gamma}$ (in the direction from $S_v$ to $V\setminus S_v$). By ignoring the flow outside $G'[S_v]\cup\partial_{G'}S_v$, we can view it as a flow from $v$ to the boundary $\partial_{G'}S_v$ that saturates the boundary up to $\frac1{1+\gamma}$ factor. Decompose the flow into paths and ignore the paths ending at edges in $G'-G$ (which are all in $\partial_{G'}S_v$), obtaining a feasible flow from $v$ to $\delta_GS_v$ that saturates $\partial_GS_v$ to factor $\frac1{1+\gamma}$. 

Suppose first that $v\in B$. Further restrict the flow paths to only those ending at the edges in the subset $E_G(A\setminus S_v,A\cap S_v)$ of $\partial_GS_v$. Each of these paths must cross $E_G(A\cap S_v,B\cap S_v)$. There is at least $\frac1{1+\gamma}w(E_G(A\setminus S_v,A\cap S_v))$ flow along these paths, and they must cross a total capacity of $w(E_G(A\cap S_v,B\cap S_v))$. Since the flow is feasible, we conclude that $\frac1{1+\gamma}w(E_G(A\setminus S_v,A\cap S_v))\le w(E_G(A\cap S_v,B\cap S_v)$. In the operation that uncrosses $S_v$, the newly cut edges are precisely $E_G(A\setminus S_v,A\cap S_v)$, and all edges in $E_G(A\cap S_v,B\cap S_v)$ disappear. We \emph{charge} the newly cut edges $E_G(A\setminus S_v,A\cap S_v)$ to the deleted edges $E_G(A\cap S_v,B\cap S_v)$ at a $1+\gamma$ to $1$ ratio. Finally, if $v\in A$, then the argument is symmetric by replacing $A$ and $B$, and the charging is identical.

Since the sets $S_v:v\in R$ are disjoint, 
each edge is either charged to or charged from, but not both. If the total weight of charged-to edges is $W$, then the total weight of newly cut edges is at most $(1+\gamma)W$, so the final cut has weight at most $\lambda_G(p,q)-W+(1+\gamma)W\le(1+\gamma)\lambda_G(p,q)$, as promised.
\end{proof}

\begin{lemma}\label{lem:small}
For any $v\in  R$ and any distinct vertices $p,q\in U_v$, we have $\lambda_G(p,q)\le\lambda_{G_v}(p,q)\le(1+13\epsilon)\lambda_G(p,q)$.
\end{lemma}
\begin{proof}
The lower bound $\lambda_G(p,q)\le\lambda_{G_v}(p,q)$ holds because $G_v$ is a contraction of $G$, so we focus on the upper bound.
Fix any $(p,q)$-mincut in $G$, and let $S$ be the side of the mincut not containing $s$ (recall that $s\in U$ and $s\notin S_v$). Since $S_v\cup S$ is a $(p, s)$-cut (and also a $(q, s)$-cut), it is in particular a Steiner cut for terminals $U$, so $\delta_G(S_v\cup S)\ge(1-\epsilon)\lambda$. Also, 
$\delta_GS_v\le(1+\gamma)(1+10\epsilon)\lambda\le(1+11\epsilon)\lambda$
 by \Cref{lem:Sv}. Together with the submodularity of cuts, we obtain
\BAL (1+11\epsilon)\lambda+\delta_GS \ge \delta_GS_v + \delta_GS \ge \delta_G(S_v\cup S) + \delta_G(S_v\cap S) \ge (1-\epsilon)\lambda+ \delta_G(S_v\cap S) ,\EAL
The set $S_v\cap S$ stays intact under the contraction from $G$ to $G_v$, so $\delta_{G_v}(S_v\cap S)=\delta_G(S_v\cap S)$. Therefore,
\BAL \lambda_{G_v}(p,q)\le \delta_{G_v}(S_v\cap S)=\delta_G(S_v\cap S) \le \delta_GS+12\epsilon\lambda = \lambda_G(p,q) +12\epsilon\lambda .\EAL
Finally, $\lambda_G(p,q)$ is at least the Steiner mincut of $G$, which is at least $(1-\epsilon)\lambda$, so the above is at most $\lambda_G(p,q)+12\epsilon\cdot\lambda_G(p,q)/(1-\epsilon) \le (1+13\epsilon)\lambda_G(p,q)$, as promised.
\end{proof}

Combining the lemmas above, we can conclude the following. 

\begin{lemma}\label{lem:approx}
Algorithm~\ref{alg:approxGH} outputs a $\big((1+13\epsilon)(1+\gamma)^{O(\epsilon^{-1}\log^6n)}\big)^{\log_{1.5}|U|}$-approximate Gomory-Hu Steiner tree. With $\gamma=\epsilon^2/\log^6n$, the approximation factor is $(1+\epsilon)^{O(\log|U|)}$.
\end{lemma}
\begin{proof}
To avoid clutter, define $\alpha=C\epsilon^{-1}\log^6n$ for large enough constant $C>0$. Consider the path down the recursion tree leading up to the current recursive instance, and let $k$ be the number of consecutive recursive calls of type $G_\lar$ directly preceding the current instance.
We apply induction on $|U|$ and $k$ to prove an $((1+13\epsilon)(1+\gamma)^\alpha)^{\log_{1.5}|U|}(1+\gamma)^{-k}$-approximation factor. By \Cref{lem:Sv}, we have $|U_v|\le2|U|/3$ for all $v\in R$, so by induction, the recursive outputs $(T_v,f_v)$ are Gomory-Hu Steiner trees with approximation $((1+13\epsilon)(1+\gamma)^\alpha)^{\log_{1.5}|U_v|}\le((1+13\epsilon)(1+\gamma)^\alpha)^{\log_{1.5}|U|-1}$.  By definition, this means that for all $s,t\in U_v$ and the minimum-weight edge $(u,u')$ on the $s$--$t$ path in $T_v$, letting $U'_v\subseteq U_v$ be the vertices of the connected component of $T_v-(u,u')$ containing $s$, we have that $f^{-1}_v(U'_v)$ is a $((1+13\epsilon)(1+\gamma)^\alpha)^{\log_{1.5}|U|-1}$-approximate $(s,t)$-mincut in $G_v$ with value $w_T(u,u')$. Define $U'\subseteq U$ as the vertices of the connected component of $T-(u,u')$ containing $s$. By construction of $(T,f)$ (lines~\ref{line:combine-T}~and~\ref{line:combine-f}), the set $f^{-1}(U')$ is simply $f^{-1}_v(U'_v)$ with the vertex $x_v$ replaced by $V\setminus S_v$ in the case that $x_v\in f^{-1}(U')$. Since $G_v$ is simply $G$ with all vertices $V\setminus S_v$ contracted to $x_v$, we conclude that $\delta_{G_v}( f^{-1}_v(U'_v)) = \delta_G( f^{-1}(U'))$. By \Cref{lem:small}, the values $\lambda_G(s,t)$ and $\lambda_{G_v}(s,t)$ are within factor $(1+13\epsilon)$ of each other, so $\delta_G( f^{-1}(U'))$ approximates the $(s,t)$-mincut in $G$ to a factor $(1+13\epsilon)\cdot((1+13\epsilon)(1+\gamma)^\alpha)^{\log_{1.5}|U|-1}$, which we want to show is at most $((1+13\epsilon)(1+\gamma)^\alpha)^{\log_{1.5}|U|}(1+\gamma)^{-k}$. This follows by \Cref{lem:depth} since w.h.p., we always have $k\le C\epsilon^{-1}\log^6n=\alpha$ for large enough constant $C>0$. Thus, the Gomory-Hu Steiner tree condition for $(T,f)$ is satisfied for all $s,t\in U_v$ for some $v\in R$.

We now focus on the case $s,t\in U_\lar$. By induction, the recursive output $(T_\lar,f_\lar)$ is a Gomory-Hu Steiner tree with approximation $((1+13\epsilon)(1+\gamma)^\alpha)^{\log_{1.5}|U|}(1+\gamma)^{-(k+1)}$. Again, consider $s,t\in U_\lar$ and the minimum-weight edge $(u,u')$ on the $s$--$t$ path in $T_\lar$, and let $U'_\lar\subseteq U_\lar$ be the vertices of the connected component of $T_\lar-(u,u')$ containing $s$. Define $U'\subseteq U$ as the vertices of the connected component of $T-(u,u')$ containing $s$. By a similar argument, we have $\delta_{G_\lar}(f^{-1}_\lar(U'_\lar)) = \delta_G(f^{-1}(U'))$. By \Cref{lem:large}, we also have $\lambda_{G_\lar}(s,t)=(1+\gamma)\lambda_G(s,t)$, so $\delta_G( f^{-1}(U'))$ is a $\left(((1+13\epsilon)(1+\gamma)^\alpha)^{\log_{1.5}|U|}(1+\gamma)^{-(k+1)}\cdot(1+\gamma)\right)$-approximate $(s,t)$-mincut in $G$, fulfilling the Gomory-Hu Steiner tree condition for $(T,f)$ in the case $s,t\in U_\lar$.

There are two remaining cases: $s\in U_v$ and $t\in U_{v'}$ for distinct $v,v'\in R$, and $s\in U_v$ and $t\in U_\lar$; we treat both cases simultaneously. Since $G$ has Steiner mincut at least $\lambda$, each of the contracted graphs $G_\lar$ and $G_v$ also has Steiner mincut at least $\lambda$. Since all edges on the approximate Gomory-Hu Steiner tree correspond to actual cuts in the graph, every edge in $T_v$ and $T_\lar$ has weight at least $\lambda$. By construction, the $s$--$t$ path in $T$ has at least one edge of the form $(f_v(x_v),f_\lar(y_v))$, added on line~\ref{line:combine-T}; this edge has weight $\delta_GS_v\le(1+\epsilon)(1+\gamma)\lambda$ by \Cref{lem:Sv}. Therefore, the minimum-weight edge on the $s$--$t$ path in $T$ has weight at least $\lambda$ and at most $(1+\epsilon)(1+\gamma)\lambda$; in particular, it is a $(1+\epsilon)(1+\gamma)$-approximation of $\lambda_G(s,t)$, which fits the bound since $|U|\ge2$. If the edge is of the form $(f_v(x_v),f_\lar(y_v))$, then by construction, the relevant set $f^{-1}(U')$ is exactly $S_v$, which is a $(1+\epsilon)$-approximate $(s,t)$-mincut in $G$. If the edge is in $T_\lar$ or $T_v$ or $T_{v'}$, then we can apply the same arguments used previously.
\end{proof}

Finally, we can reset $\epsilon\gets\Theta(\epsilon/\log n)$ so that the $(1+\epsilon)^{O(\log|U|)}$-approximation becomes $(1+\epsilon)$. This concludes \Cref{thm:approx-w}.

\section{Expander Decomposition}
\label{sec:expdecomp}
In this section, we show how the fair cut algorithm implies a near-optimal expander decomposition algorithm, following the framework of Saranurak and Wang~\cite{SaranurakW19}. We first begin with some notation exclusive to this section. Define the \emph{volume} of a set of vertices $S$ as $\textbf{\textup{vol}}(S)=\sum_{v\in S}\deg(v)$, and let $G\{S\}$ denote the subgraph $G[S]$ with (weighted) self-loops added to vertices so that all vertex degrees are preserved, i.e., $\deg_G(v)=\deg_{G\{S\}}(v)$ for all $v\in S$. For a graph $G$, define its \emph{conductance} as
\[ \Phi_G = \min_{\emptyset\subsetneq S\subsetneq V} \frac { c(E(S,V\setminus S)) } { \min\{\textbf{\textup{vol}}(S),\textbf{\textup{vol}}(V\setminus S) \} } .\]
We call $G$ a \emph{$\phi$-expander} if $\Phi_G\ge\phi$.

\begin{restatable}[Near-linear expander decomposition]{theorem}{expdecomp}
\label{thm:exp_decomp}
Given a graph $G=(V,E)$ and a parameter $\phi$, there is a randomized $\tilde{O}(m)$-time algorithm that with high probability finds a partitioning of $V$ into $V_1,\ldots,V_k$ such that $\Phi_{G\{V_i\}}\ge\phi$ for all $i\in[k]$ and $\sum_i\delta(V_i)=\tilde{O}(\phi m)$.
\end{restatable}

Note that if $G\{V_i\}$ is a $\phi$-expander, then so is the induced subgraph $G[V_i]$ (which is sometimes more directly applicable). We also remark that \cite{SaranurakW19} prove almost the exact same theorem, except their running time is $\tilde{O}(m/\phi)$, and is therefore slower for small values of $\phi$.

At a high level, we use the same high-level recursive approach, except we replace the flow subroutines in their \emph{trimming} and \emph{cut-matching} steps of \cite{SaranurakW19} with a fair cut computation.
We note that there are known black-box reductions from expander decomposition to computing (approximately) most-balanced sparse cuts. But these reductions have some drawbacks and do not lead to near-optimal algorithms as in \Cref{thm:exp_decomp}. 
The first reduction is implicit by Spielman and Teng \cite{SpielmanT04}. However, they can only obtain a \emph{weak} expander decomposition from most-balanced sparse cut algorithms. It is weak in the sense that each part is only guaranteed to be contained in some expanders, but may not induce an expander itself. Another reduction by Nanongkai and Saranurak \cite{NanongkaiS17} suffers from an inherent $n^{o(1)}$ factor loss in both quality and running time. 
By refining the non-blackbox approach of \cite{SaranurakW19} via fair cuts, we successfully obtain the first expander decomposition algorithm that are optimal up to polylogarithmic factors. 

\subsection{Algorithm overview}

We begin by describing the recursive algorithm of~\cite{SaranurakW19} at a high level. There are two main subroutines, \emph{cut-matching} and \emph{trimming}, to be described later. On input graph $G=(V,E)$ and parameter $\phi$, the algorithm Decomp$(G,\phi)$ outputs a partition of $V$ as follows.
 \begin{enumerate}
 \item Call Cut-Matching$(G,\phi)$, which either certifies that $\Phi_G\ge\phi$ or finds a cut $(A,R)$
 \item If we certify $\Phi_G\ge\phi$, then return $\{V\}$ (the trivial partition)
 \item Else if we find a relatively balanced cut $(A,R)$, where $\textbf{\textup{vol}}(A)$ and $\textbf{\textup{vol}}(R)$ are both $\Omega(\textbf{\textup{vol}}(V)/\log^2m)$:
  \begin{enumerate}
  \item Return Decomp$(G\{A\},\phi)$ $\cup$ Decomp$(G\{R\},\phi)$
  \end{enumerate}
 \item Else, suppose that $\textbf{\textup{vol}}(R) \le O(\textbf{\textup{vol}}(V)/\log^2m)$:
  \begin{enumerate}
  \item $A'=\text{Trimming}(G,A,\phi)$
  \item Return $\{A'\} \cup \text{Decomp}(G\{A'\},\phi)$
  \end{enumerate}
 \end{enumerate}
If Cut-Matching and Trimming run in $T$ time, then the entire recursive algorithm takes $\tilde{O}(T)$ time. In~\cite{SaranurakW19}, the two subroutines are solved in $\tilde{O}(m/\phi)$ time. In this section, we improve both running times to $\tilde{O}(m)$ by substituting their flow subroutines with fair cuts/flows.

\subsection{Trimming step}

To describe the trimming step formally, we need the concept of a \emph{nearly expander}.

\begin{definition}[nearly $\phi$-expander]
Given $G=(V,E)$ and a set of vertices $A\subseteq V$, we say that $A$ is a nearly $\phi$-expander in $G$ if for all subsets $S\subseteq A$ with $\textbf{\textup{vol}}(S)\le\textbf{\textup{vol}}(A)/2$, we have $c(E(S,V\setminus S))\ge\phi\textbf{\textup{vol}}(S)$.
\end{definition}

In the trimming step, we are given a set $A\subseteq V$ such that $A$ is a nearly $\phi$-expander in $G$, and the goal is to ``trim'' $A$ to a subset $A'\subseteq A$ such that $G\{A'\}$ is a $\phi/6$-expander. The formal subroutine is described in the theorem below, copied almost identically to Theorem~2.1 of~\cite{SaranurakW19} except for the improved $\tilde{O}(m)$ running time.

\begin{theorem}[Trimming, Theorem~2.1 of~\cite{SaranurakW19}]
Given graph $G=(V,E)$ and $A\subseteq V$ such that
 \begin{enumerate}
 \item $A$ is a nearly $\phi$-expander in $G$, and
 \item $c(E(A,\overline A))\le\phi\textbf{\textup{vol}}(A)/10$,
 \end{enumerate}
the trimming step finds $A'\subseteq A$ in time $\tilde{O}(m)$ such that $\Phi_{G\{A'\}}\ge\phi/6$. Moreover, $\textbf{\textup{vol}}(A')\ge\textbf{\textup{vol}}(A)-4c(E(A,\overline A))/\phi$ and $c(E(A',\overline{A'}))\le2c(E(A,\overline A))$.
\end{theorem}
\begin{proof}
Consider the following $(s,t)$-flow problem on a new graph $H=(V_H,E_H)$. Start from $G\{A\}$, and contract $V\setminus A$ into a single vertex and label it the source $s$. Next, multiply the capacity of each edge by $3/\phi$. Finally, add a new sink vertex $t$ and connect it to each vertex $v\in A$ with an edge of capacity $\deg_{G\{A\}}(v)$. Let $\alpha=0.1$, and compute a $(1+\alpha)$-fair cut $(S,T)$. Let $A'=T\setminus\{t\}$, which we now show satisfies the properties of the lemma.

First, suppose for contradiction that $G\{A'\}$ is not a $\phi/6$-expander. Then, there is a violating set $U\subseteq A'$ satisfying
\[ c(E(U,A'\setminus U))\le\frac\phi6\textbf{\textup{vol}}(U) .\]
Since $A$ is a nearly $\phi$-expander,
\[ c(E(U,V\setminus U))\ge\phi\textbf{\textup{vol}}(U) .\]
Taking the difference of the two inequalities above,
\[ c(E(U,V\setminus A'))=c(E(U,V\setminus U))-c(E(U,A'\setminus U))\ge\frac{5\phi}6\textbf{\textup{vol}}(U) .\]
Since $(S,T)$ is a $(1+\alpha)$-fair cut, there is a feasible flow $f$ that saturates each edge of $E_H(S,T)$ to factor $\frac1{1+\alpha}$. Each edge $(u,v)$ in $E(U,V\setminus A')$ corresponds to an edge in $E_H(S,T)$ of capacity $\frac3\phi c_{G\{A\}}(u,v)$, and the flow $f$ must send at least $\frac1{1+\alpha}\cdot\frac3\phi c_{G\{A\}}(u,v) \ge \frac2\phi c_{G\{A\}}(u,v)$ flow along that edge (in the direction from $S$ to $T$). In total, the amount of flow entering $U$ in $H$ is at least
\[ \frac2\phi c_{G\{A\}}(E(U,V\setminus A')) \ge \frac2\phi \cdot \frac{5\phi}6\textbf{\textup{vol}}(U) = \frac53\textbf{\textup{vol}}(U) .\]
On the other hand, at most $\textbf{\textup{vol}}(U)$ flow can leave $U$ along the edges incident to $t$, and at most 
\[ \frac3\phi c_{G\{A\}}(E(U,A'\setminus U)) \le \frac3\phi\cdot\frac\phi6\textbf{\textup{vol}}(U)=\frac12\textbf{\textup{vol}}(U) \]
 flow can cross from $U$ to $A'\setminus U$. This totals at most $\frac32\textbf{\textup{vol}}(U)$ flow that can exit $U$, which is strictly less than the $\ge \frac53\textbf{\textup{vol}}(U)$ flow that enters $U$, a contradiction. Thus, $G\{A'\}$ is a $\phi/6$-expander.

Finally, we show the properties $\textbf{\textup{vol}}(A')\ge\textbf{\textup{vol}}(A)-4c(E(A,\overline A))/\phi$ and $c(E(A',\overline{A'}))\le2c(E(A,\overline A))$ promised by the lemma. Since $(S,T)$ is a $(1+\alpha)$-fair cut, it is in particular a $(1+\alpha)$-approximate $(s,t)$-mincut. Since $(\{s\},V_H\setminus\{s\})$ is an $(s,t)$-cut of capacity $\frac3\phi c(E(A,\overline A))$, it follows that the cut $(S,T)$ has capacity at most $(1+\alpha)\cdot\frac3\phi c(E(A,\overline A))$. To prove the first property above, note that each vertex $v\in A\setminus A'$ is on the $S$-side of the cut $(S,T)$, and therefore contributes $\deg_{G\{A\}}(v)$ to the cut $(S,T)$ from the edge $(v,t)$. Summing over all $v\in A\setminus A'$, we obtain
\[ \textbf{\textup{vol}}(A\setminus A') \le c_H(E(S,T)) \le (1+\alpha)\cdot\frac3\phi c(E(A,\overline A)) \le \frac4\phi c(E(A,\overline A)) ,\]
which proves the first property. For the second property above, note that each edge $(u,v)$ in $E(A',\overline{A'})$ corresponds to an edge in $E(S,T)$ with $3/\phi$ times the capacity, so summing over all such edges,
\[ \frac3\phi c(E(A',\overline{A'})) \le c_H(E(S,T)) \le (1+\alpha)\cdot\frac3\phi c(E(A,\overline A)) ,\]
which proves the second property.
\end{proof}

\subsection{Cut-matching step}

In the cut-matching step, the goal is to either certify that the input graph is an expander, or find a low-conductance cut with a special property: either it is balanced, or if not, we guarantee that the larger side is a nearly expander. The name ``cut-matching'' comes from the \emph{cut-matching game} framework~\cite{khandekar2009graph} that this step uses, though its description is not required in this section.

The formal subroutine is described in the theorem below, copied almost identically to Theorem~2.2 of~\cite{SaranurakW19} except for the improved $\tilde{O}(m)$ running time.

\begin{theorem}[Cut-Matching, Theorem~2.2 of~\cite{SaranurakW19}]
Given a graph $G=(V,E)$ and a parameter $\phi$, the \emph{cut-matching} step takes $\tilde{O}(m)$ time and must end with one of the three cases:
 \begin{enumerate}
 \item We certify $G$ has conductance $\Phi_G\ge\phi$.
 \item We find a cut $(A,\overline A)$ in $G$ of conductance $\Phi_G(\overline A)=O(\phi^2m)$, and $\textbf{\textup{vol}}(A),\textbf{\textup{vol}}(\overline A)$ are both $\Omega(m/\log^2m)$, i.e., we find a relatively balanced low conductance cut.
 \item We find a cut $(A,\overline A)$ with $\Phi_G(\overline A)\le c_0\phi\log^2m$ for some constant $c_0$, and $\textbf{\textup{vol}}(\overline A)\le m/(10c_0\log^2m)$, and $A$ is a nearly $\phi$-expander.
 \end{enumerate}
\end{theorem}

We will not present the entire proof of this theorem, since most of the steps remain unchanged from~\cite{SaranurakW19}. The only step that takes $\tilde{O}(m/\phi)$ time in~\cite{SaranurakW19} is their subroutine Lemma~B.6, so it suffices to describe it and improve its running time to $\tilde{O}(m)$.

First, we introduce some notation from~\cite{SaranurakW19}. Given a graph $G=(V,E)$ and a subset of vertices $A\subseteq V$, denote by $G\{S\}$ the induced subgraph $G[S]$ but with self-loops added to vertices so that any vertex in $S$ has the same degree as its degree in $G$. Given a multi-graph $G=(V,E)$, its \emph{subdivision graph} $G_E=(V',E')$ is the graph where we put a \emph{split node} $x_e$ on each edge $e\in E$ (including the self-loops). Formally, $V'=V\cup X_E$ where $X_E=\{x_e\mid e\in E\}$, and $E'=\{(u,x_e),(v,x_e)\mid e=(u,v)\in E\}$. While~\cite{SaranurakW19} only defines the subdivision graph for unweighted graphs, we can extend the definition to weighted graphs by assigning the edges $(u,x_e),(v,x_e)$ ro have capacity $c(e)$ for each edge $e=(u,v)\in E$. For a split node $x_{(u,v)}$, we abuse notation and define its \emph{capacity} $c(x_{(u,v)})$ to be the capacity $c(u,v)$ of the edge $(u,v)$ in $G$. For a set of split nodes $S$, its total capacity $c(S)$ is the sum of the capacities of the split nodes in $S$.

The input to the subroutine of Lemma~B.6 is
 \begin{enumerate}
 \item A set of vertices $A\subseteq V'$,
 \item A set of source split nodes $A^l\subseteq A\cap X_E$ of total capacity at most $c_{G\{A\}}(A\cap X_E)/8$, and\label{bound-l}
 \item A set of target split nodes $A^r\subseteq A\cap X_E$ of total capacity at least $c_{G\{A\}}(A\cap X_E)/2$.\label{bound-r}
 \end{enumerate}
For any graph $H$ and positive number $U$, let $H^U$ be the graph where each edge has its capacity multiplied by $U$. Let $U=1/(\phi\log^2m)$, and consider a flow problem on $(G_E\{A\})^U$ where each split node $x_{(u,v)}\in A^l$ is a source of $c(u,v)$ units of mass (where $c(u,v)$ is the original capacity in $G_E$, not multiplied by $U$) and each split node $x_{(u,v)}\in A^r$ is a sink with capacity $c(u,v)$. The task is to either find
 \begin{enumerate}
 \item A feasible flow $f$ for the above problem, or
 \item A cut $S$ where $\Phi_{G\{A\}}(S)=O(\phi\log^2m)$ and a feasible flow for the above flow problem when only split nodes $x_{(u,v)}$ in $A^l\setminus S$ are sources of $c(u,v)$ units.
 \end{enumerate}
Lemma~B.6 of~\cite{SaranurakW19} uses a push-relabel or blocking-flow algorithm that runs in $O(m/(\phi\log m))$ time. Using fair cuts, we improve the running time to $\tilde{O}(m)$, independent of $\phi$, in the lemma below.

\begin{lemma}
We can solve the task above in $\tilde{O}(m)$ time.
\end{lemma}
\begin{proof}
Let $\alpha=0.1$, and consider the flow problem on the graph $H=(G_E\{A\})^{U/(1+\alpha)}$ instead. First, convert it to an $(s,t)$-flow problem by adding a source vertex $s$, connected to each $x_{(u,v)}\in A^l$ with capacity $c_{G\{A\}}(u,v)$, and a sink vertex $t$, connected to each $x_{(u,v)}\in A^r$ with capacity $c_{G\{A\}}(u,v)/(1+\alpha)$. Next, we compute a $(1+\alpha)$-fair cut $(S,T)$ and corresponding feasible flow $f'$ in $\tilde{O}(m)$ time. There are two cases below:
 \begin{enumerate}
 \item $S=\{s\}$. In this case, by definition of fair cuts, the flow $f'$ sends at least $c_{G\{A\}}(u,v)/(1+\alpha)$ flow out of each edge from $s$. By computing a path decomposition and removing paths accordingly, we can modify $f'$ to a new feasible flow $f''$ that sends \emph{exactly} $c_{G\{A\}}(u,v)/(1+\alpha)$ flow along each edge out of $s$, and at most $c_{G\{A\}}(u,v)/(1+\alpha)$ flow along each edge into $t$. Finally, we let flow $f$ be $f''$ multiplied by $(1+\alpha)$, and then restricted to graph $(G_E\{A\})^U$. Since $f''$ is feasible on the edges in $(G_E\{A\})^{U/(1+\alpha)}$, we conclude that $f$ is feasible on $(G_E\{A\})^U$. 
 \item $S\ne\{s\}$. In this case, let $E_s\subseteq E_H(S,T)$ be the edges of the cut incident to $s$, let $E_t\subseteq E_H(S,T)$ be those incident to $t$, and let $E_m=E_H(S,T)\setminus(E_s\cup E_t)$ be the remaining cut edges. Recall that edges in $E_s$ and $E_t$ retain their original capacity from $G_E\{A\}$, while edges in $E_m$ have their capacity scaled by $U/(1+\alpha)$. Also, note that $E_m$ is, up to this scaling factor, exactly the cut $E(S\setminus\{s\},T\setminus\{t\})$ in the original graph $G\{A\}$.\footnote{We show later that the degenerate case $T=\{t\}$ cannot happen.} In other words,
\begin{gather}
c_{G\{A\}}(E(S\setminus\{s\},T\setminus\{t\})) = \frac{1+\alpha}U \cdot c_H(E_m) .\label{eq:expander1}
\end{gather}

Let $\overline E_s$ be the edges incident to $s$ that are not in $E_s$. Since $(S,T)$ is a $(1+\alpha)$-fair cut, there is a flow $f$ from $s$ to $t$ that saturates each edge in $E_H(S,T)$ to fraction at least $\frac1{1+\alpha}$. In particular, this means that the sub-flow from $s$ starting from edges $\overline E_s$ must saturate edges in $E_H(S,T)\setminus E_s$ to fraction at least $\frac1{1+\alpha}$. This implies that $c_H(E_H(S,T)\setminus E_s) \le (1+\alpha) c_H(\overline E_s)$. Moreover, for each edge $(s,x_e)\in \overline E_s$, the split node $x_e$ is on the $S\setminus\{s\}$ side of the cut $E(S\setminus\{s\},T\setminus\{t\})$ in $G\{A\}$, so
\begin{gather}
\textbf{\textup{vol}}_{G\{A\}}(S\setminus\{s\}) \ge \sum_{(s,x_e)\in\overline E_s}\deg_{G\{A\}}(x_e) = 2 c_H(\overline E_s) \ge \frac2{1+\alpha}c_H(E_H(S,T)\setminus E_s) \ge \frac2{1+\alpha}c_H(E_m) .\label{eq:expander2}
\end{gather}
Putting (\ref{eq:expander1}) and (\ref{eq:expander2}) together, we obtain
\begin{gather}
\textbf{\textup{vol}}_{G\{A\}}(S\setminus\{s\}) \ge \frac{2U}{(1+\alpha)^2} c_{G\{A\}}(E(S\setminus\{s\},T\setminus\{t\})) ,\label{eq:expander3}
\end{gather}
so we would be done as long as we show that $\textbf{\textup{vol}}_{G\{A\}}(S\setminus\{s\}) \le O(\textbf{\textup{vol}}_{G\{A\}}(T\setminus\{t\}))$.

Consider now the edges $E_t$. Their capacities are scaled down by $1/(1+\alpha)$, so their total original capacity is at most $(1+\alpha)^2c_{G\{A\}}(A^l)$, which is at most $(1+\alpha)^2c_{G\{A\}}(A\cap X_E)/8$ by property~(\ref{bound-l}). On the other hand, the total capacity of edges incident to $t$ is $c_{G\{A\}}(A^r)/(1+\alpha)$, which is at least $c_{G\{A\}}(A\cap X_E)/(2(1+\alpha))$ by property~(\ref{bound-r}). It follows that at least 
\[  c_{G\{A\}}(A\cap X_E)/(2(1+\alpha))  - (1+\alpha)^2c_{G\{A\}}(A\cap X_E)/8 \ge \Omega(c_{G\{A\}}(A\cap X_E)) \]
total capacity of edges incident to $t$ are not in $E_t$. In other words, their corresponding split nodes are on the $T\setminus\{t\}$ side of the cut $E(S\setminus\{s\},T\setminus\{t\})$, which means that $\textbf{\textup{vol}}_{G\{A\}}(T\setminus\{t\}) \ge \Omega(c_{G\{A\}}(A\cap X_E))$. Now observe that $c_{G\{A\}}(A\cap X_E)$ is a constant fraction of the total volume of the graph $G\{A\}$, so $\textbf{\textup{vol}}_{G\{A\}}(T\setminus\{t\}) \ge \Omega(\textbf{\textup{vol}}_{G\{A\}}(A))$. Together with (\ref{eq:expander3}), we obtain the desired
\[ 
\Phi_{G\{A\}}(S\setminus\{s\}) = \frac{c_{G\{A\}}(S\setminus\{s\},T\setminus\{t\})}{\min\{\textbf{\textup{vol}}_{G\{A\}}(S\setminus\{s\}),\textbf{\textup{vol}}_{G\{A\}}(T\setminus\{t\})\}} \le O(1/U) = O(\phi\log^2m).
\]
\end{enumerate}
\end{proof}

\section*{Acknowledgements}
This project has received funding from the European Research Council (ERC) under the European Union's Horizon 2020 research and innovation programme under grant agreement No 715672. Danupon Nanongkai was also supported by the Swedish Research Council (Reg. No. 2019-05622). Debmalya Panigrahi was supported in part by NSF grants CCF-1750140 (CAREER Award) and CCF-1955703.

\bibliographystyle{alpha}
\bibliography{ref}

\newcommand{\etalchar}[1]{$^{#1}$}
\begin{thebibliography}{vdBLN{\etalchar{+}}20}

\bibitem[AHK12]{AroraHK12}
Sanjeev Arora, Elad Hazan, and Satyen Kale.
\newblock The multiplicative weights update method: a meta-algorithm and
  applications.
\newblock {\em Theory of Computing}, 8(1):121--164, 2012.

\bibitem[AKL{\etalchar{+}}21]{AbboudKLPST21-GHtreeSubcubic}
Amir Abboud, Robert Krauthgamer, Jason Li, Debmalya Panigrahi, Thatchaphol
  Saranurak, and Ohad Trabelsi.
\newblock Gomory-hu tree in subcubic time.
\newblock {\em CoRR}, abs/2111.04958, 2021.

\bibitem[AKT20a]{AbboudKT20focs}
Amir Abboud, Robert Krauthgamer, and Ohad Trabelsi.
\newblock Cut-equivalent trees are optimal for min-cut queries.
\newblock In {\em 61st {IEEE} Annual Symposium on Foundations of Computer
  Science, {FOCS} 2020, Durham, NC, USA, November 16-19, 2020}, pages 105--118.
  {IEEE}, 2020.

\bibitem[AKT20b]{AbboudKT20soda}
Amir Abboud, Robert Krauthgamer, and Ohad Trabelsi.
\newblock New algorithms and lower bounds for all-pairs max-flow in undirected
  graphs.
\newblock In Shuchi Chawla, editor, {\em Proceedings of the 2020 {ACM-SIAM}
  Symposium on Discrete Algorithms, {SODA} 2020, Salt Lake City, UT, USA,
  January 5-8, 2020}, pages 48--61. {SIAM}, 2020.

\bibitem[AKT21a]{AbboudKT21-focs}
Amir Abboud, Robert Krauthgamer, and Ohad Trabelsi.
\newblock {APMF} {\textless} apsp? gomory-hu tree for unweighted graphs in
  almost-quadratic time.
\newblock In {\em 62nd {IEEE} Annual Symposium on Foundations of Computer
  Science, {FOCS} 2021, Denver, CO, USA, February 7-10, 2022}, pages
  1135--1146. {IEEE}, 2021.

\bibitem[AKT21b]{AbboudKT22-soda}
Amir Abboud, Robert Krauthgamer, and Ohad Trabelsi.
\newblock Friendly cut sparsifiers and faster gomory-hu trees.
\newblock {\em CoRR}, abs/2110.15891, 2021.

\bibitem[AKT21c]{AbboudKT21}
Amir Abboud, Robert Krauthgamer, and Ohad Trabelsi.
\newblock Subcubic algorithms for gomory-hu tree in unweighted graphs.
\newblock In Samir Khuller and Virginia~Vassilevska Williams, editors, {\em
  {STOC} '21: 53rd Annual {ACM} {SIGACT} Symposium on Theory of Computing,
  Virtual Event, Italy, June 21-25, 2021}, pages 1725--1737. {ACM}, 2021.

\bibitem[BBG{\etalchar{+}}20]{bernstein2020fully}
Aaron Bernstein, Jan van~den Brand, Maximilian~Probst Gutenberg, Danupon
  Nanongkai, Thatchaphol Saranurak, Aaron Sidford, and He~Sun.
\newblock Fully-dynamic graph sparsifiers against an adaptive adversary.
\newblock {\em arXiv preprint arXiv:2004.08432}, 2020.

\bibitem[BHKP07]{BHKP07}
Anand Bhalgat, Ramesh Hariharan, Telikepalli Kavitha, and Debmalya Panigrahi.
\newblock An {$\tilde{O}(mn)$} {G}omory-{H}u tree construction algorithm for
  unweighted graphs.
\newblock In {\em 39th Annual ACM Symposium on Theory of Computing}, STOC'07,
  pages 605--614, 2007.

\bibitem[CGL{\etalchar{+}}20]{Chuzhoy20det}
Julia Chuzhoy, Yu~Gao, Jason Li, Danupon Nanongkai, Richard Peng, and
  Thatchaphol Saranurak.
\newblock A deterministic algorithm for balanced cut with applications to
  dynamic connectivity, flows, and beyond.
\newblock In {\em 2020 IEEE 61st Annual Symposium on Foundations of Computer
  Science (FOCS)}, pages 1158--1167. IEEE, 2020.

\bibitem[CH03]{ColeH03}
Richard Cole and Ramesh Hariharan.
\newblock A fast algorithm for computing steiner edge connectivity.
\newblock In Lawrence~L. Larmore and Michel~X. Goemans, editors, {\em
  Proceedings of the 35th Annual {ACM} Symposium on Theory of Computing, June
  9-11, 2003, San Diego, CA, {USA}}, pages 167--176. {ACM}, 2003.

\bibitem[CKL{\etalchar{+}}22]{ChenKLPGS22}
Li~Chen, Rasmus Kyng, Yang~P Liu, Richard Peng, Maximilian Probst~Gutenberg,
  and Sushant Sachdeva.
\newblock Maximum flow and minimum-cost flow in almost-linear time.
\newblock March 2022.

\bibitem[CLP22]{CenLP22}
Ruoxu Cen, Jason Li, and Debmalya Panigrahi.
\newblock Augmenting edge connectivity via isolating cuts.
\newblock In {\em Proceedings of the 2022 ACM-SIAM Symposium on Discrete
  Algorithms (SODA)}, 2022.

\bibitem[CQ21a]{ChekuriQ21}
Chandra Chekuri and Kent Quanrud.
\newblock Isolating cuts, (bi-)submodularity, and faster algorithms for
  connectivity.
\newblock In Nikhil Bansal, Emanuela Merelli, and James Worrell, editors, {\em
  48th International Colloquium on Automata, Languages, and Programming,
  {ICALP} 2021, July 12-16, 2021, Glasgow, Scotland (Virtual Conference)},
  volume 198 of {\em LIPIcs}, pages 50:1--50:20. Schloss Dagstuhl -
  Leibniz-Zentrum f{\"{u}}r Informatik, 2021.

\bibitem[CQ21b]{CQ20}
Chandra Chekuri and Kent Quanrud.
\newblock Isolating cuts, (bi-)submodularity, and faster algorithms for
  connectivity.
\newblock In {\em {ICALP}}, volume 198 of {\em LIPIcs}, pages 50:1--50:20.
  Schloss Dagstuhl - Leibniz-Zentrum f{\"{u}}r Informatik, 2021.

\bibitem[CRJ17]{CohenRD17}
Jaime Cohen, Luiz~A. Rodrigues, and Elias P.~Duarte Jr.
\newblock Parallel cut tree algorithms.
\newblock {\em J. Parallel Distributed Comput.}, 109:1--14, 2017.

\bibitem[CS19]{chang2019improved}
Yi-Jun Chang and Thatchaphol Saranurak.
\newblock Improved distributed expander decomposition and nearly optimal
  triangle enumeration.
\newblock In {\em Proceedings of the 2019 ACM Symposium on Principles of
  Distributed Computing}, pages 66--73, 2019.

\bibitem[DV94]{DinitzV94}
Yefim Dinitz and Alek Vainshtein.
\newblock The connectivity carcass of a vertex subset in a graph and its
  incremental maintenance.
\newblock In Frank~Thomson Leighton and Michael~T. Goodrich, editors, {\em
  Proceedings of the Twenty-Sixth Annual {ACM} Symposium on Theory of
  Computing, 23-25 May 1994, Montr{\'{e}}al, Qu{\'{e}}bec, Canada}, pages
  716--725. {ACM}, 1994.

\bibitem[GG18]{geissmann2018parallel}
Barbara Geissmann and Lukas Gianinazzi.
\newblock Parallel minimum cuts in near-linear work and low depth.
\newblock In {\em Proceedings of the 30th on Symposium on Parallelism in
  Algorithms and Architectures}, pages 1--11, 2018.

\bibitem[GKK{\etalchar{+}}15]{ghaffari2015near}
Mohsen Ghaffari, Andreas Karrenbauer, Fabian Kuhn, Christoph Lenzen, and Boaz
  Patt-Shamir.
\newblock Near-optimal distributed maximum flow.
\newblock In {\em Proceedings of the 2015 ACM Symposium on Principles of
  Distributed Computing}, pages 81--90, 2015.

\bibitem[GLP21]{GaoLP21}
Yu~Gao, Yang~P. Liu, and Richard Peng.
\newblock Fully dynamic electrical flows: Sparse maxflow faster than
  goldberg-rao.
\newblock {\em FOCS}, 2021.

\bibitem[GR98]{GoldbergR98}
Andrew~V. Goldberg and Satish Rao.
\newblock Beyond the flow decomposition barrier.
\newblock {\em J. {ACM}}, 45(5):783--797, 1998.

\bibitem[GRST21]{goranci2021expander}
Gramoz Goranci, Harald R{\"a}cke, Thatchaphol Saranurak, and Zihan Tan.
\newblock The expander hierarchy and its applications to dynamic graph
  algorithms.
\newblock In {\em Proceedings of the 2021 ACM-SIAM Symposium on Discrete
  Algorithms (SODA)}, pages 2212--2228. SIAM, 2021.

\bibitem[HKP07]{HariharanKP07}
Ramesh Hariharan, Telikepalli Kavitha, and Debmalya Panigrahi.
\newblock Efficient algorithms for computing all low \emph{s-t} edge
  connectivities and related problems.
\newblock In {\em Proceedings of the Eighteenth Annual {ACM-SIAM} Symposium on
  Discrete Algorithms, {SODA} 2007, New Orleans, Louisiana, USA, January 7-9,
  2007}, pages 127--136, 2007.

\bibitem[Kar00]{Karger2000minimum}
David~R Karger.
\newblock Minimum cuts in near-linear time.
\newblock {\em Journal of the ACM (JACM)}, 47(1):46--76, 2000.

\bibitem[KLOS14]{KelnerLOS14}
Jonathan~A Kelner, Yin~Tat Lee, Lorenzo Orecchia, and Aaron Sidford.
\newblock An almost-linear-time algorithm for approximate max flow in
  undirected graphs, and its multicommodity generalizations.
\newblock In {\em Proceedings of the twenty-fifth annual ACM-SIAM symposium on
  Discrete algorithms}, pages 217--226. SIAM, 2014.

\bibitem[KRV09]{khandekar2009graph}
Rohit Khandekar, Satish Rao, and Umesh Vazirani.
\newblock Graph partitioning using single commodity flows.
\newblock {\em Journal of the ACM (JACM)}, 56(4):1--15, 2009.

\bibitem[LF80]{ladner1980parallel}
Richard~E Ladner and Michael~J Fischer.
\newblock Parallel prefix computation.
\newblock {\em Journal of the ACM (JACM)}, 27(4):831--838, 1980.

\bibitem[LNP{\etalchar{+}}21]{LiNPSY21}
Jason Li, Danupon Nanongkai, Debmalya Panigrahi, Thatchaphol Saranurak, and
  Sorrachai Yingchareonthawornchai.
\newblock Vertex connectivity in poly-logarithmic max-flows.
\newblock In {\em {STOC}}, pages 317--329. {ACM}, 2021.

\bibitem[LP20]{LiP20}
Jason Li and Debmalya Panigrahi.
\newblock Deterministic min-cut in poly-logarithmic max-flows.
\newblock In {\em 61st {IEEE} Annual Symposium on Foundations of Computer
  Science, {FOCS} 2020}. {IEEE} Computer Society, 2020.

\bibitem[LP21]{LiP21}
Jason Li and Debmalya Panigrahi.
\newblock Approximate {Gomory-Hu} tree is faster than $n-1$ max-flows.
\newblock In {\em Proceedings of the 53rd Annual {ACM} Symposium on Theory of
  Computing}, 2021.

\bibitem[LPS21]{LiPS21}
Jason Li, Debmalya Panigrahi, and Thatchaphol Saranurak.
\newblock A nearly optimal all-pairs min-cuts algorithm in simple graphs.
\newblock In {\em 62nd {IEEE} Annual Symposium on Foundations of Computer
  Science, {FOCS} 2021, Denver, CO, USA, February 7-10, 2022}, pages
  1124--1134. {IEEE}, 2021.

\bibitem[LS20]{LiuS20}
Yang~P. Liu and Aaron Sidford.
\newblock Faster energy maximization for faster maximum flow.
\newblock In Konstantin Makarychev, Yury Makarychev, Madhur Tulsiani, Gautam
  Kamath, and Julia Chuzhoy, editors, {\em Proccedings of the 52nd Annual {ACM}
  {SIGACT} Symposium on Theory of Computing, {STOC} 2020, Chicago, IL, USA,
  June 22-26, 2020}, pages 803--814. {ACM}, 2020.

\bibitem[Mad11]{madry2011graphs}
Aleksander Madry.
\newblock {\em From graphs to matrices, and back: new techniques for graph
  algorithms}.
\newblock PhD thesis, Massachusetts Institute of Technology, 2011.

\bibitem[MCJ20]{MaskeCD20}
Charles Maske, Jaime Cohen, and Elias P.~Duarte Jr.
\newblock Speeding up the gomory-hu parallel cut tree algorithm with efficient
  graph contractions.
\newblock {\em Algorithmica}, 82(6):1601--1615, 2020.

\bibitem[MN20]{MukhopadhyayN20}
Sagnik Mukhopadhyay and Danupon Nanongkai.
\newblock Weighted min-cut: sequential, cut-query, and streaming algorithms.
\newblock In {\em {STOC}}, pages 496--509. {ACM}, 2020.

\bibitem[MN21]{MukhopadhyayN21-submodular}
Sagnik Mukhopadhyay and Danupon Nanongkai.
\newblock A note on isolating cut lemma for submodular function minimization.
\newblock {\em CoRR}, abs/2103.15724, 2021.

\bibitem[NS17]{NanongkaiS17}
Danupon Nanongkai and Thatchaphol Saranurak.
\newblock Dynamic spanning forest with worst-case update time: adaptive, las
  vegas, and o(n\({}^{\mbox{1/2 - {\(\epsilon\)}}}\))-time.
\newblock In {\em {STOC}}, pages 1122--1129. {ACM}, 2017.

\bibitem[NSW17]{NanongkaiSW17}
Danupon Nanongkai, Thatchaphol Saranurak, and Christian Wulff{-}Nilsen.
\newblock Dynamic minimum spanning forest with subpolynomial worst-case update
  time.
\newblock In {\em {FOCS}}, pages 950--961. {IEEE} Computer Society, 2017.

\bibitem[Pen16]{Peng16}
Richard Peng.
\newblock Approximate undirected maximum flows in o (m polylog (n)) time.
\newblock In {\em Proceedings of the twenty-seventh annual ACM-SIAM symposium
  on Discrete algorithms}, pages 1862--1867. SIAM, 2016.

\bibitem[RST14]{RackeST14}
Harald R{\"a}cke, Chintan Shah, and Hanjo T{\"a}ubig.
\newblock Computing cut-based hierarchical decompositions in almost linear
  time.
\newblock In {\em Proceedings of the twenty-fifth annual ACM-SIAM symposium on
  Discrete algorithms}, pages 227--238. SIAM, 2014.

\bibitem[She13]{Sherman13}
Jonah Sherman.
\newblock Nearly maximum flows in nearly linear time.
\newblock In {\em 2013 IEEE 54th Annual Symposium on Foundations of Computer
  Science}, pages 263--269. IEEE, 2013.

\bibitem[She17]{Sherman2017area}
Jonah Sherman.
\newblock Area-convexity, $l_\infty$ regularization, and undirected
  multicommodity flow.
\newblock In {\em Proceedings of the 49th Annual ACM SIGACT Symposium on Theory
  of Computing}, pages 452--460, 2017.

\bibitem[ST04]{SpielmanT04}
Daniel~A. Spielman and Shang{-}Hua Teng.
\newblock Nearly-linear time algorithms for graph partitioning, graph
  sparsification, and solving linear systems.
\newblock In {\em Proceedings of the 36th Annual {ACM} Symposium on Theory of
  Computing, Chicago, IL, USA, June 13-16, 2004}, pages 81--90, 2004.

\bibitem[SW19]{SaranurakW19}
Thatchaphol Saranurak and Di~Wang.
\newblock Expander decomposition and pruning: Faster, stronger, and simpler.
\newblock 2019.
\newblock To appear in SODA'19.

\bibitem[vdBLL{\etalchar{+}}21]{BrandLLSSSW21minimum}
Jan van~den Brand, Yin~Tat Lee, Yang~P. Liu, Thatchaphol Saranurak, Aaron
  Sidford, Zhao Song, and Di~Wang.
\newblock Minimum cost flows, mdps, and $\ell_1$-regression in nearly linear
  time for dense instances.
\newblock 2021.
\newblock arXiv:2101.05719.

\bibitem[vdBLN{\etalchar{+}}20]{BrandLNPSS0W20-matching}
Jan van~den Brand, Yin~Tat Lee, Danupon Nanongkai, Richard Peng, Thatchaphol
  Saranurak, Aaron Sidford, Zhao Song, and Di~Wang.
\newblock Bipartite matching in nearly-linear time on moderately dense graphs.
\newblock In {\em {FOCS}}, pages 919--930. {IEEE}, 2020.

\bibitem[Wul17]{Wulff-Nilsen17}
Christian Wulff{-}Nilsen.
\newblock Fully-dynamic minimum spanning forest with improved worst-case update
  time.
\newblock In {\em {STOC}}, pages 1130--1143. {ACM}, 2017.

\bibitem[Zha21a]{Zhang-simple}
Tianyi Zhang.
\newblock Faster cut-equivalent trees in simple graphs.
\newblock {\em CoRR}, abs/2106.03305, 2021.

\bibitem[Zha21b]{Zhang-weighted}
Tianyi Zhang.
\newblock Gomory-hu trees in quadratic time.
\newblock {\em CoRR}, abs/2112.01042, 2021.

\end{thebibliography}

\appendix

\section{Parallel Algorithms}

The goal of this section to prove \Cref{thm:ghtree-parallel-intro}. Along the way, we will show that all algorithmic components we use and develop can be parallelized.

\subsection{Congestion Approximators}
\label{sec:parallel_congest}

The first thing we need is a parallel construction of congestion approximators (see \Cref{thm:congest}).

\begin{theorem}
[Parallel Congestion approximator]\label{thm:congest parallel}
There is a randomized algorithm that, given an unweighted graph $G=(V,E)$ with $n$
vertices and $m$ edges, constructs in $m^{1+o(1)}$ work and $m^{o(1)}$ depth with high probability same laminar as in \Cref{thm:congest} except that $\gamma_{\mathcal{S}}=n^{o(1)}$.
\end{theorem}

We only state the result for unweighted graphs as it follows quite easily from \cite{chang2019improved,goranci2021expander}. We believe that known techniques also imply the same for weighted graphs. Below, we sketch the proof of \Cref{thm:congest parallel}.

First, we need a definition of \emph{boundary-linked expander decomposition} introduced in \cite{goranci2021expander}.
For any graph $G=(V,E)$ and any set $S\subset V$, let $G[S]$ denote
the subgraph of $G$ induced by $S$. For any $w\ge0$, let $G[S]^{w}$
be obtained from $G[S]$ by adding $w$ self-loops to each vertex
$v\in S$ for every boundary edge $(v,x)$, $x\notin S$. 
\begin{defn}
For any graph $G=(V,E)$ with $m$ edges, a \emph{$(\epsilon,\phi,\alpha)$-boundary-linked
expander decomposition} is partition ${\cal U}=(U_{1},\dots,U_{k})$
of vertex set $V$ such that $\sum_{i}|E(U_{i},V\setminus U_{i})|\le\epsilon m$
and $G[U_{i}]^{\alpha/\phi}$ is a $\phi$-expander for all $i$. 
\end{defn}

Note that $(\epsilon,\phi,0)$-boundary-linked expander decomposition
is the standard $(\epsilon,\phi)$-expander decomposition. A parallel
algorithm for computing an expander decomposition of an unweighted
graph was explicitly shown in \cite{chang2019improved}. In fact, the algorithm works even
in the distributed model called CONGEST. 
\begin{theorem}[\cite{chang2019improved}] For any positive integer $k$, $\epsilon\in(0,1)$, and
$\phi\ge(\epsilon/\log n)^{2^{O(k)}}$, there is an algorithm for
computing an $(\epsilon,\phi)$-expander decomposition of an unweighted
graph in CONGEST in $O(n^{2/k}\poly(1/\phi,\log n))$ rounds w.h.p.
In fact, this algorithm has $n^{1/O(\log\log\log n)}$-depth and $m^{1+o(1)}$
work.
\end{theorem}

We will choose $k=\log\log\log n$ from now on. This algorithm can
be easily extended to compute a $(\epsilon,\phi,\epsilon)$-boundary-linked
expander decomposition. The idea is as follows: whenever we find a
$\phi$-sparse cut, for each cut edge $(u,v)$, we add $(\alpha/\phi)$
self-loops on both $u$ and $v$ before recursing on both sides. The
largest boundary-linked parameter $\alpha$ we can get can be derived
by setting $\epsilon=1/O(\log n)$ and see the largest value of $\phi$
we can get. In this case, it is $1/2^{\Theta(\log\log n)^{2})}$ when
$\epsilon=1/O(\log n)$ and $k=\log\log\log n$. From this, it implies
the following:

\begin{theorem}
\label{thm:parallel expand decomp}
When $\epsilon=1/2^{\Theta(\sqrt{\log n})}$, $\phi\ge(\epsilon/\log n)^{2^{O(\log\log\log)}}\ge1/2^{\Theta(\sqrt{\log n}\cdot\log\log n)}$,
and $\alpha\ge1/2^{\Theta(\log\log n)^{2})}$, there is an algorithm
that w.h.p. computes a $(\epsilon,\phi,\alpha)$-boundary-linked expander
decomposition in $n^{1/O(\log\log\log n)}$-depth and $m^{1+o(1)}$
work. (In fact, the algorithm is implementable in CONGEST in $n^{1/O(\log\log\log n)}$
rounds.)
\end{theorem}

In \cite{goranci2021expander}, it is shown that constructing congestion approximators can
be reduced to computing boundary-linked expander decomposition a few
times, which is summarized as follows:

\begin{lem}
By calling an algorithm for computing a $(\epsilon,\phi,\alpha)$-boundary-linked
expander decomposition for $O(\log_{(1/\epsilon)}m)$ times, one can
construct a congestion approximator ${\cal S}$ with quality $\gamma_{{\cal S}}=O((1/\phi)\cdot(1/\alpha)^{\log_{(1/\epsilon)}m})$.
\end{lem}

Plugging \Cref{thm:parallel expand decomp} into the above lemma, this implies an algorithm for \Cref{thm:congest parallel} where $n^{1/O(\log\log\log n)}$ depth and $m^{1+o(1)}$ work that computes a congestion approximator ${\cal S}$ with quality
$\gamma_{{\cal S}}=2^{\Theta(\sqrt{\log n}\cdot(\log\log n)^{2})} = n^{o(1)}$.

\subsection{Fair Cuts}

Given the above parallel construction for congestion approximator, we can obtain the following parallel fair cut algorithm:

\begin{theorem}[Parallel Fair Cut]\label{thm:fair parallel}
Given an unweighted graph $G=(V,E)$, two vertices
$s,t\in V$, and $\epsilon\in(0,1]$, we can compute with high probability
a $(1+\epsilon)$-fair $(s,t)$-cut in 
$n^{o(1)}/\poly(\epsilon)$ depth and $m^{1+o(1)}/\poly(\epsilon)$ work.
\end{theorem}
Before proving the above theorem, we first argue how to obtain a parallel version of the $\almostfair$ algorithm.

See the running time analysis of $\almostfair$ in \Cref{sec:time almostfair}. We can parallelize it as follows. 
We initialize by computing a congestion approximator $\cal{S}$ with quality $\gamma_{\cal{S}}=n^{o(1)}$ via \Cref{thm:congest parallel}. The other initialization steps consist of elementary operations which can be parallelized in $\Otil(1)$ depth and $\Otil(m)$ work. 

For each round of the multiplicative weight update algorithm, the only non-trivial step is to a compute the ``deletion set'' $D^i$ via a sweep cut (\Cref{lem:1}).  

We will prove the below claim at the end.
\begin{claim}
\label{claim:sweep cut parallel}
\Cref{lem:1} admits a parallel implementation with $\Otil(1)$ depth and $\Otil(m)$  work.
\end{claim}

Since our multiplicative weight update algorithm consists of  $T=O(\log(n)/\alpha^2) = m^{o(1)}/\poly(\epsilon)$ rounds (recall that $\alpha = \epsilon/\gamma_{\cal{S}}$), we can implement the $\almostfair$ algorithm from \Cref{thm:almost fair} in $m^{o(1)}/\poly(\epsilon)$ depth and  $m^{1+o(1)}/\poly(\epsilon)$ work.

Given the parallel implementation of the $\almostfair$ algorithm, we are almost done. The algorithm for computing fair cuts in \Cref{sec:alg fair}  simply calls the $\almostfair$ subroutine for $O(\log(C/\beta)/\beta)$ times where we set $\beta = \Theta(\alpha/\log n)$. Therefore, the algorithm require  $m^{o(1)}/\poly(\epsilon)$ depth and  $m^{1+o(1)}/\poly(\epsilon)$ work. This concludes \Cref{thm:fair parallel}.

\begin{proof}[Proof of \Cref{claim:sweep cut parallel}]
Recall that the problem is to compute $x^*$ which is the largest $x$ such that $\Delta|_{V^{i-1}}(V_{>x})-\delta_{H}(V_{>x}) > 0$ where $V_{>x}=\{v\in V(H):\phi_{v}^{i}>x\}$.

We start by parallel sorting vertices $v$ according to their potential $\phi^i_v$ in decreasing order. Let $v_1,\dots,v_n$ be the vertices after sorting. Let $S_k = \{v_1,\dots,v_k\}$.
We can compute the list of values of $\Delta|_{V^{i-1}}(S_k)$ for all $k\in [n]$ in $O(\log n)$ depth and $O(n)$ work using a classic parallel prefix sum algorithm \cite{ladner1980parallel}.

Observe that our goal is equivalent to  finding the largest $k$ where  $\delta_{H}(S_k) - \Delta|_{V^{i-1}}(S_k) < 0$.
By binary search, we can reduce the problem to checking if there is $k$ where  $\delta_{H}(S_k) - \Delta|_{V^{i-1}}(S_k) < 0$.

Now, this problem can be solved using a parallel 1-respecting mincut algorithm by Karger \cite{Karger2000minimum} (see also Lemma 11 of \cite{geissmann2018parallel}) with $O(\log n)$ depth and $O(m)$ work. The reduction is as follows.
Let $H'$ be the graph obtained from $H$ by inserting the tree $P = (v_1,\dots,v_n)$, which is a path. 
Let $M$ be a big number such that $M - \Delta|_{V^{i-1}}(S_k) > 0$. 
Each tree edge $(v_k,v_{k+1})\in P$, we set its weight to be $M - \Delta|_{V^{i-1}}(S_k)$. 
By computing a mincut in $H'$ that $1$-respect the tree $P$, we will obtain $k$ such that $\delta_{H'}(S_k)$ is minimized. Since $\delta_{H'}(S_k) = \delta_{H}(S_k) + M - \Delta|_{V^{i-1}}(S_k)$, we can just check if $\delta_{H'}(S_k) -M < 0$. 
\end{proof}

\subsection{Isolating Cuts and Gomory-Hu Tree}
Here, we finally prove \Cref{thm:ghtree-parallel-intro}.
We first briefly explain how the approximate isolating cuts algorithm (Algorithm~\ref{alg:isolating}) and Gomory-Hu tree algorithm (Algorithm~\ref{alg:approxGH}) can be parallelized to run in $\tilde O(m)$ work and $\textup{polylog}(n)$ parallel time.

For approximate isolating cuts, Phase~1 of Algorithm~\ref{alg:isolating} requires $O(\log n)$ many calls to $(1+\gamma)$-fair cut, which has a parallel algorithm by \Cref{thm:fair parallel}. For Phase~2, the sets $S_t$ and graphs $G_t$ can be constructed independently for different $t$ in parallel, and for the $(1+\beta)$-approximate minimum cut computation, we can use the parallel $(1+\beta)$-fair cut algorithm of \Cref{thm:fair parallel}, which is also a $(1+\beta)$-approximate minimum cut.

For Gomory-Hu tree, there are a few additional algorithms that need to be investigated. For the ``Cut Threshold Step'' algorithm (Algorithm~\ref{alg:step}), the $O(\log n)$ independent iterations can be executed in parallel, so the entire algorithm can as well. The $(1+\gamma)$-approximate Gomory-Hu Steiner tree ``step'' (Algorithm~\ref{alg:step2}) makes $O(\log^3n)$ (sequential) calls to Algorithm~\ref{alg:step}, so it can also be parallelized. The Gomory-Hu tree algorithm itself (Algorithm~\ref{alg:approxGH}) makes one call to Algorithm~\ref{alg:step2} and, aside from the recursive call on line~\ref{line:recursive}, consists of elementary operations that can directly be parallelized. For the recursive calls, we use \Cref{lem:depth} to argue that the recursion tree has depth $\textup{polylog}(n)$ w.h.p., so the recursive calls can be parallelized as well. (We stop the recursion after a large enough $\textup{polylog}(n)$ many recursive calls, which is all we need w.h.p.)

\section{Proof of Uncrossing Property}\label{sec:uncrossing}

Here, we prove the uncrossing property (\Cref{lem:uncrossing-property}), restated below. We remark that the proof follows the same outline as the proof of \Cref{lem:intersect} for approximate isolating cuts.
\Uncrossing*
\begin{proof}
Let $(U,V\setminus U)$ be a $(u,v)$-mincut. Without loss of generality, assume that $t\notin U$. (Otherwise, we can swap $u$ and $v$ and use $V\setminus U$ in place of $U$.) Our goal is to show that $U\cap S$ is an $\alpha$-approximate $(u,v)$-mincut contained in $S$, so that setting $R=U\cap S$ proves the lemma. Equivalently, we want to show that $\delta(U\cap S) \le \alpha\cdot\delta(U)$.

Using the notation $\uplus$ for disjoint union, we can write
    \begin{align*}
        E(U, V\setminus U) &= E(U\cap S, V\setminus (U\cup S)) \uplus E(U\cap S, S\setminus U) \uplus E(U\setminus S, V\setminus U) \\
        E(U\cap S, V\setminus (U\cap S)) &= E(U\cap S, V\setminus (U\cup S)) \uplus E(U\cap S, S\setminus U) \uplus E(U\cap S, U\setminus S).
    \end{align*}
    Since the first two sets are identical, we only need to compare the third sets $E(U\setminus S, V\setminus U)$ and $E(U\cap S, U\setminus S)$. Since $(S,T)$ is an $\alpha$-fair  $(s,t)$-cut, there is a feasible flow from $s$ to $t$ that, for each edge in $E(S,T)$, sends at least $1/\alpha$ times capacity in the direction from $S$ to $T$. Now, consider the flow on the subset of edges $E(U\cap S, U\setminus S) \subseteq E(S, T)$. This flow must reach $t$ eventually, and it must exit $U\setminus S$ along the edges in $E(U\setminus S, V\setminus (U\cup S))$. Thus, 
    $$\delta(U\cap S, U\setminus S) \le \alpha \cdot \delta(U\setminus S, V\setminus (U\cup S)) \le \alpha \cdot \delta(U\setminus S, V\setminus U).$$
    It follows that $\delta(U\cap S) \le \alpha\cdot \delta(U)$, which proves the lemma.
\end{proof}

\end{document}